\newif\iflong
\newif\ifshort

\longtrue

\iflong
\else
\shorttrue
\fi

\documentclass[a4paper,USenglish,cleveref, autoref, thm-restate]{lipics-v2021}

\usepackage{paralist}
\usepackage[utf8]{inputenc}
\usepackage{amsmath}    %
\usepackage{amssymb}    %
\usepackage{mathtools}  %
\usepackage{microtype}  %
\usepackage{multirow}   %
\usepackage{booktabs}   %
\usepackage{subcaption} %
\usepackage[ruled,linesnumbered]{algorithm2e} %
\usepackage[usenames,dvipsnames,table]{xcolor} %
\usepackage{tikz}
\usetikzlibrary{shapes.geometric,arrows.meta,arrows}
\usetikzlibrary{calc}
\usepackage{nag}       %
\usepackage{todonotes} %
\usepackage{tcolorbox} %
\usepackage{hyperref}  %
\usepackage{floatrow}
\usepackage{wrapfig}

\newcommand{\sump}{\sqcup}

\newcommand{\cc}[1]{{\mbox{\textnormal{\textsf{#1}}}}\xspace} 
\newcommand{\NP}{\cc{NP}}

\newcommand{\XP}{\cc{XP}}

\newcommand{\ang}[1]{\langle #1 \rangle}
\newcommand{\sep}{\;|\;}
\newcommand{\seq}{\subseteq}

\newcommand{\ca}[1]{\mathcal{#1}}
\newcommand{\bigoh}{\mathcal{O}}
\newcommand{\tww}{\textnormal{tww}}
\newcommand{\dist}{\textnormal{dist}}

\title{Computing Twin-Width Parameterized by \\ the Feedback Edge Number}
\titlerunning{Computing Twin-Width Parameterized by the Feedback Edge Number}
\author{Jakub Balab\' an}{Faculty of Informatics, Masaryk University, Brno, Czech Republic}{485053@mail.muni.cz}{0000-0002-2475-8938}{}

\author{Robert Ganian}{Algorithms and Complexity Group, TU Wien, Vienna, Austria}{rganian@gmail.com}{0000-0002-7762-8045}{Robert Ganian acknowledges support by the FWF and WWTF Science Funds (FWF project Y1329 and WWTF project ICT22-029).}

\author{Mathis Rocton}{Algorithms and Complexity Group, TU Wien, Vienna, Austria}{mrocton@ac.tuwien.ac.at}{0000-0002-7158-9022}{Mathis Rocton acknowledges support by the \includegraphics[width=0.5cm]{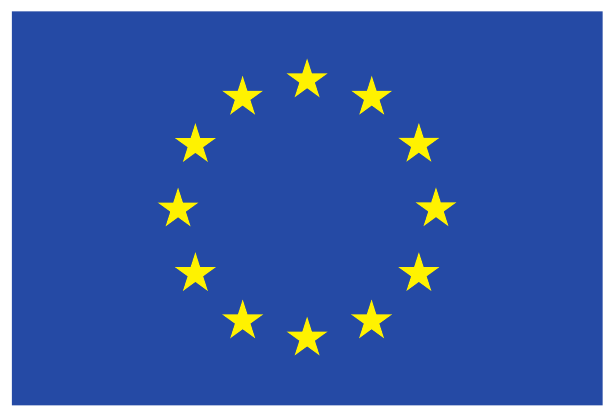} European Union's Horizon 2020 research and innovation COFUND programme (LogiCS@TUWien, grant agreement No 101034440), and the FWF Science Fund (FWF project Y1329).}
\authorrunning{Jakub Balab\' an, Robert Ganian, Mathis Rocton} 

\usepackage{nameref,cleveref}
\Crefname{splemma}{Lemma}{Lemmas}
\Crefname{sptheorem}{Theorem}{Theorems}
\Crefname{spdefinition}{Definition}{Definitions}
\Crefname{spproperty}{Property}{Properties}
\Crefname{spcorollary}{Corollary}{Corollaries}

\usepackage{regexpatch}
\makeatletter
\xpatchcmd\thmt@restatable{%
\csname #2\@xa\endcsname\ifx\@nx#1\@nx\else[{#1}]\fi
}{%
\ifthmt@thisistheone
\csname #2\@xa\endcsname\ifx\@nx#1\@nx\else[{#1}]\fi
\else
\iflong \csname #2\@xa\endcsname\fi
\ifshort \csname #2\@xa\endcsname[{$\clubsuit$}]\fi
\fi}{}{}
\makeatother

\begin{document}
\hideLIPIcs

\definecolor{turquoise}{rgb}{0.19, 0.84, 0.78}
\definecolor{ffxfqq}{rgb}{1.,0.4980392156862745,0.}
\definecolor{ududff}{rgb}{0.30196078431372547,0.30196078431372547,1.}
\definecolor{ffqqqq}{rgb}{1.,0.,0.}
\definecolor{yqyqyq}{rgb}{0.5019607843137255,0.5019607843137255,0.5019607843137255}
\definecolor{ududff}{rgb}{0.30196078431372547,0.30196078431372547,1.}
\definecolor{big}{rgb}{0.30196078431372547,0.30196078431372547,1.}

\newcommand{\cen}{\Psi}

\maketitle

\begin{abstract}
The problem of whether and how one can compute the twin-width of a graph---along with an accompanying contraction sequence---lies at the forefront of the area of algorithmic model theory. While significant effort has been aimed at obtaining a fixed-parameter approximation for the problem when parameterized by twin-width, here we approach the question from a different perspective and consider whether one can obtain (near-)optimal contraction sequences under a larger parameterization, notably the feedback edge number $k$. 
As our main contributions, under this parameterization we obtain
(1) a linear bikernel for the problem of 
either computing a $2$-contraction sequence or determining that none exists and (2) an approximate fixed-parameter algorithm which computes an $\ell$-contraction sequence (for an arbitrary specified $\ell$) or determines that the twin-width of the input graph is at least $\ell$.
These algorithmic results rely on newly obtained insights into the structure of optimal contraction sequences, and as a byproduct of these we also slightly tighten the bound on the twin-width of graphs with small feedback edge number.
\end{abstract}

\newpage
\section{Introduction}\label{sec:intro}
Since its introduction by Bonnet, Kim, Thomassé and Watrigant in 2020~\cite{BonnetKTW22},
 the notion of \emph{twin-width} has made an astounding impact on the field of algorithmic model-checking~\cite{DBLP:conf/soda/BonnetGKTW21,DBLP:conf/icalp/BonnetG0TW21,DBLP:conf/stoc/BonnetGMSTT22,DBLP:conf/stacs/BonnetGMT23,DBLP:conf/soda/BonnetKRT22}. Indeed, it promises a unified explanation of why model-checking first order logic is fixed-parameter tractable on a number of graph classes which were, up to then, considered to be separate islands of tractability for the model-checking problem, including proper minor-closed graphs, graphs of bounded rank-width, posets of bounded width and map graphs~\cite{BonnetKTW22}; see also the recent works on other graph classes of bounded twin-width~\cite{DBLP:conf/iwpec/BalabanH21,DBLP:conf/wg/BalabanHJ22,EppsteinConfluent}. Beyond this, twin-width was shown to have fundamental connections to rank-width and path-width~\cite{DBLP:conf/soda/BonnetKRT22} 
as well as to matrix theory~\cite{DBLP:conf/stacs/BonnetGMT23}, 
and has by now been studied even in areas such as graph drawing~\cite{EppsteinConfluent} 
and SAT Solving~\cite{DBLP:conf/sat/GanianPSSS22,DBLP:conf/ijcai/SchidlerS23}.

And yet, essentially all twin-width based algorithmic results known to date require a corresponding decomposition---a so-called \emph{contraction sequence}---to be provided as part of the input. The fact that the inner workings of these algorithms rely on a contraction sequence is not surprising; after all, the same reliance on a suitable decomposition is present in essentially all graph algorithms parameterized by classical width measures such as treewidth~\cite{RobertsonS86} or rank-width~\cite{Oum05,GanianH10}. But while optimal decompositions for treewidth and rank-width can be computed in fixed-parameter time when parameterized by the respective width measure~\cite{Bodlaender96,HlinenyO08} and even more efficient algorithms are known when aiming for decompositions that are only a constant-factor worse than optimal~\cite{DBLP:conf/focs/Korhonen21,DBLP:conf/stoc/FominK22}, the situation is entirely different in the case of twin-width. In particular, it is known that already deciding whether a graph has twin-width at most $4$, i.e., admits a $4$-contraction sequence, is \NP-hard~\cite{DBLP:conf/icalp/BergeBD22} (ruling out fixed-parameter as well as \XP\ algorithms for computing optimal contraction sequences). Moreover, whether one can at least compute approximately-optimal contraction sequences in fixed-parameter time is arguably the most prominent open question in contemporary research of twin-width. 

\smallskip
\noindent \textbf{Contribution.}\quad
Given the difficulty of computing (near-)optimal contraction sequences when parameterized by twin-width itself, in this article we ask whether one can at least compute such contraction sequences under a larger parameterization, i.e., when using an auxiliary parameter which yields stronger restrictions on the input graph\footnote{When approximating width parameters, it is desirable to aim for approximation errors which depend only on the targeted width parameter.}. Algorithms obtained under such stronger restrictions are not intended to be used as a pre-computation step prior to using twin-width for model checking, but rather aim to further our understanding of the fundamental problem of computing (near-)optimal contraction sequences. In this sense, our work follows in the footsteps of previous work on, e.g., treedepth parameterized by the vertex cover number~\cite{DBLP:conf/iwpec/KobayashiT16}, 
MIM-width parameterized by the feedback edge number and other parameters~\cite{DBLP:conf/innovations/EibenGHJK22}, 
treewidth parameterized by the feedback vertex number~\cite{DBLP:journals/siamdm/BodlaenderJK13}
and the directed feedback vertex number parameterized by the (undirected) feedback vertex number~\cite{DBLP:journals/algorithmica/BergougnouxEGOR21}.

As our two main contributions, we obtain the first non-trivial fixed-parameter algorithms for computing (near-)optimal contraction sequences.

\begin{restatable}{theorem}{thmone}
\label{thm:tww2}
The problem of deciding whether the twin-width of an input graph is at most $2$ admits a linear bikernel when parameterized by the feedback edge number $k$. Moreover, a $2$-contraction sequence for $G$ (if one exists) can be computed in time $2^{\bigoh(k\cdot \log k)}+n^{\bigoh(1)}$.
\end{restatable}

We remark that Theorem~\ref{thm:tww2} providing a \emph{bikernel}~\cite{DowneyF13} (instead of a kernel) is merely due to the output being a \emph{trigraph}~\cite{BonnetKTW22}. Our second result targets graphs of higher twin-width:

\begin{restatable}{theorem}{thmtwo}
\label{thm:tww-3+}
There is an algorithm which takes as input an $n$-vertex graph $G$ with feedback edge number $k$, runs in time $f(k) \cdot n^{\bigoh(1)}$ for a computable function $f$, and outputs a contraction sequence for $G$ of width at most $\tww(G)+1$.
\end{restatable}

We note that the graph parameter used in our results---the feedback edge number or equivalently the edge deletion distance to acyclicity---is highly restrictive and provides stronger structural guarantees on the input graph than not only twin-width itself, but also rank-width and treewidth. In a sense, it is one of the two most ``restrictive'' structural parameters used in the design of fixed-parameter algorithms~\cite{UhlmannW13,BannisterCE18,GanianO21,GanianK21,FichteGHSO23}, with the other being the vertex cover number, i.e., the minimum size of a vertex cover (see also Figure~\ref{fig:hierarchy}). %
 But while there is a trivial fixed-parameter algorithm for computing optimal contraction sequences w.r.t.\ the vertex cover number\footnote{In particular, this follows by repeatedly deleting vertices which are twins (an operation which is known to preserve twin-width~\cite{BonnetKTW22}) until one obtains a problem kernel.} and also a polynomial-time algorithm for solving the same problem on trees~\cite{BonnetKTW22}, lifting the latter to our more general setting parameterized by the feedback edge number is far from a trivial undertaking and goes hand in hand with the development of new insights into optimal contraction sequences of highly structured graphs. 

\begin{figure}
\floatbox[{\capbeside\thisfloatsetup{capbesideposition={left},capbesidewidth=8cm}}]{figure}[\FBwidth]
{
\scalebox{0.6}{\begin{tikzpicture}
\clip (-5.75, -8.75) rectangle (5.25, 0.75);
  \node [draw, fill=red!30, rectangle, minimum width=2.5cm, minimum height=1cm, rounded corners] (twin) at (-1,-2) {\textbf{Twin-width}};
  
  \node [draw, rectangle, minimum width=2.5cm, minimum height=1cm, rounded corners] (clique) at (-1,-3.5) {\textbf{Clique-width}};
  
  \node [draw, rectangle, minimum width=2.5cm, minimum height=1cm, rounded corners] (tw) at (1.5,-4.8) {\textbf{Treewidth}};
  
  \node [draw, fill=green!30, rectangle, minimum width=2.5cm, minimum height=1cm, rounded corners, align=center] (nd) at (-4,-6) {\textbf{Neighborhood}\\ \textbf{diversity}};
  
  \node [draw, rectangle, minimum width=2.5cm, minimum height=1cm, rounded corners] (td) at (-1,-6.4) {\textbf{\textbf{Treedepth}}};
  
  \node [draw, fill=green!30, rectangle, minimum width=2.5cm, minimum height=1cm, rounded corners, align=center] (vc) at (-3.5,-8) {\textbf{Vertex cover number}};
  
  \node [draw, fill=blue!30, rectangle, minimum width=2.5cm, minimum height=1cm, rounded corners, align=center] (fen) at (1,-8) {\textbf{Feedback edge number}};
  \draw[{Latex[length=2mm]}-] (twin.south) -- (clique.north);
  \draw[{Latex[length=2mm]}-] ([xshift=20pt]clique.south) -- ([xshift=-25pt]tw.north);
  \draw[{Latex[length=2mm]}-] ([xshift=-20pt]clique.south) -- ([xshift=20pt]nd.north);
  \draw[{Latex[length=2mm]}-] ([xshift=-15pt]tw.south) -- ([xshift=15pt]td.north);
  \draw[{Latex[length=2mm]}-] ([xshift=-15pt]td.south) -- ([xshift=20pt]vc.north);
  \draw[{Latex[length=2mm]}-] (nd.south) -- ([xshift=-14pt]vc.north);
  \draw[{Latex[length=2mm]}-] (tw.south) -- ([xshift=15pt]fen.north);
\end{tikzpicture}}
}
{\caption{Complexity of computing twin-width with respect to notable structural parameters. A directed path from parameter $x$ to parameter $y$ indicates that $y$ is upper-bounded by a function of $x$, i.e., that $x$ is more restrictive than $y$. Green marks parameters where the problem is trivially fixed-parameter tractable; red and white signify para-\NP-hardness and cases where the complexity is unknown, respectively; the parameter considered in this paper is highlighted in blue.\label{fig:hierarchy}~
\vspace{0.3cm}}}
\end{figure}

\smallskip
\noindent \textbf{Proof Overview and Techniques.}\quad
At its core, both of our results are kernelization routines in the sense that they apply polynomial-time reduction rules in order to transform the input instance into an ``equivalent'' instance whose size is upper-bounded by a function of the parameter alone. The required reduction rules are in fact very simple---the difficulty lies in proving that they are safe. Below, we provide a high-level overview of the proof; for brevity, we assume here that readers are familiar with basic terminology associated with twin-width (as also introduced in Section~\ref{sec:prelim}).

The first step towards both desired kernelization algorithms consists of a set of tree-pruning rules which allow us to ``cut off'' subtrees in the graph that are connected to the rest of the graph via a bridge, i.e., a single edge. Already this step (detailed in Section~\ref{sec:cutting-forest}) requires some effort in the context of twin-width, and replaces the cut-off subtrees by one of two kinds of \emph{stumps} which depend on the properties of the replaced subtree. After the exhaustive application of these general rules, in Section~\ref{sec:preprocessing} we apply a further cleanup step which uses the structural properties guaranteed by our parameter to deal with the resulting stumps. This reduces the instance to an equivalent trigraph consisting of $\bigoh(k)$-many vertices plus a set of ``\emph{dangling}'' paths connecting these vertices---paths whose internal vertices have degree $2$.
We note that some of the reduction rules developed here, and in particular those which allow us to safely remove trees connected via a bridge, provide techniques that could be lifted to remove other kinds of dangling subgraphs and hence may be of general interest. In fact, we use our reduction rules to improve the previously known twin-width $2$ upper bound from trees to graphs of feedback edge number $1$ (Theorem~\ref{thm:fenone}), which additionally yields a slightly tighter relationship between twin-width and the feedback edge number (Corollary~\ref{cor:fenmore}).

The structure of the trigraph at this point seems rather simple: it consists of a bounded-size part plus a small set of arbitrarily long dangling paths.
Intuitively, one would like to obtain a bikernel by showing that each sufficiently long dangling path can be replaced by a path of length bounded by some constant without altering the twin-width. In Section~\ref{sec:tww2}, we implement this approach by guaranteeing the existence of ``well-structured'' $2$-contraction sequences for graphs of twin-width $2$, in turn allowing us to complete the proof of Theorem~\ref{thm:tww2}.

The situation becomes significantly more complicated when aiming for contraction sequences for graphs of higher twin-width. In particular, not only does the approach used in Section~\ref{sec:tww2} not generalize, but we prove that there can exist no safe twin-width preserving rule to shorten dangling paths to a constant length, for any constant (see Proposition~\ref{prop:dangling}). Circumventing this issue---even when allowing for an additive error of one---in Section~\ref{sec:tww-3+} forms the most challenging part of our results. The core idea used in the proof here is to partition a hypothetical optimal contraction sequence into a bounded number of stages (defined via so-called \emph{blueprints}), whereas we show that the original contraction sequence can be transformed into a ``nice'' sequence where we retain control over the operations carried out in each stage, at the cost of allowing for a slightly higher width of the sequence. The structure in these nice sequences is defined by aggregating all the descendants of the dangling paths into so-called \emph{centipedes}. Afterwards, we use an iterative argument to show that such a nice sequence can also be used to deal with a kernelized trigraph where all the long dangling paths are replaced by paths whose length is not constant, but depends on a function of $k$.

A mindmap of our techniques and algorithmic results is provided in Figure~\ref{fig:intro}.

\ifshort
\smallskip
\noindent \emph{Statements where proofs or details are provided in the appendix are marked with~$\clubsuit$}.
\fi

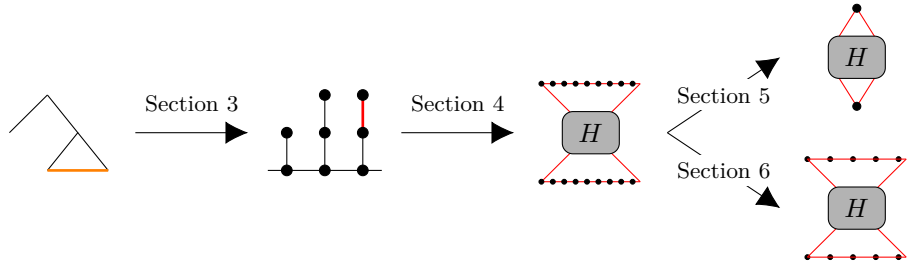
\begin{figure}[t]
\vspace{-0.4cm}
\begin{tikzpicture}

\begin{footnotesize}
\node [ rectangle, minimum width=2cm, minimum height=1.75cm, rounded corners, dashed] (a) at (5,5){};
\node [ rectangle, minimum width=2cm, minimum height=1.75cm, rounded corners, dashed] (b) at (8.5,5){};
\node [rectangle, minimum width=2cm, minimum height=1.75cm, rounded corners, dashed] (c) at (12,5){};
\node [rectangle, minimum width=2cm, minimum height=1.75cm, rounded corners, dashed] (d) at (15.5,6){};
\node [, rectangle, minimum width=2cm, minimum height=1.75cm, rounded corners, dashed] (e) at (15.5,4){};
\draw[-{Latex[length=3mm, width=3mm]}] (a.east)-- (b.west) node[midway, above=5pt] {Section 3};
\draw[-{Latex[length=3mm, width=3mm]}] (b.east)-- (c.west) node[midway, above=5pt] {Section 4};

\draw[-{Latex[length=3mm, width=3mm]}] (c.east)-- (d.west) node[midway, midway, fill=white] {Section 5};

\draw[-{Latex[length=3mm, width=3mm]}] (c.east)-- (e.west) node[midway, fill=white] {Section 6};
\end{footnotesize}

\draw[color=red](12,5)-- (11.35,5.65);
\draw[color=red](12,5)-- (12.65,5.65);
\draw[color=red](12,5)-- (12.65,4.35);
\draw[color=red](12,5)-- (11.35,4.35);
\draw[color=red](11.35,5.65)--(12.65,5.65);
\draw[color=red](11.35,4.35)--(12.65,4.35);
\node [draw, rectangle, minimum width=0.75cm, minimum height=0.56cm, rounded corners, fill=black!30] (cc) at (12,5){$H$};

\foreach \i in {0, 0.15, ..., 1.3}{
\draw [fill=black] (11.35+\i, 5.65) circle (0.9pt);
\draw [fill=black] (11.35+\i, 4.35) circle (0.9pt);
}

\foreach \i in {0, 0.3, ..., 1.3}{
\draw [fill=black] (14.85+\i, 4.65) circle (0.9pt);
\draw [fill=black] (14.85+\i, 3.35) circle (0.9pt);
}

\draw[color=red] (15.25,6.25)-- (15.5, 6.65)-- (15.75, 6.25);
\draw[color=red] (15.25,5.75)-- (15.5, 5.35)-- (15.75, 5.75);
\node [draw, rectangle, minimum width=0.75cm, minimum height=0.56cm, rounded corners, fill=black!30] (dd) at (15.5,6){$H$};
\draw [fill=black] (15.5, 6.65) circle (1.5pt);
\draw [fill=black] (15.5, 5.35) circle (1.5pt);

\draw[color=red](15.5,4)-- (14.85,4.65);
\draw[color=red](15.5,4)-- (16.15,4.65);
\draw[color=red](15.5,4)-- (16.15,3.35);
\draw[color=red](15.5,4)-- (14.85,3.35);
\draw[color=red](14.85,4.65)--(16.15,4.65);
\draw[color=red](14.85,3.35)--(16.15,3.35);
\node [draw, rectangle, minimum width=0.75cm, minimum height=0.56cm, rounded corners, fill=black!30] (dd) at (15.5,4){$H$};

\draw (4.85, 5.5)-- (5.25,5)-- (4.85, 4.5);
\draw (5.25, 5)-- (5.65,4.5);
\draw (4.85,5.5)-- (4.35, 5);
\draw[color=orange, line width = 1pt] (5.65,4.5)-- (4.85,4.5);

\draw(7.75, 4.5)-- (9.25, 4.5);
\draw(8,4.5)--(8,5);
\draw(8.5,4.5)--(8.5,5.5);
\draw(9,4.5)--(9,5);
\draw[color=red, line width=1pt](9,5)--(9,5.5);

\draw [fill=black] (8, 4.5) circle (2pt);
\draw [fill=black] (8.5, 4.5) circle (2pt);
\draw [fill=black] (9, 4.5) circle (2pt);
\draw [fill=black] (8, 5.) circle (2pt);
\draw [fill=black] (8.5, 5.) circle (2pt);
\draw [fill=black] (9, 5.) circle (2pt);
\draw [fill=black] (8.5, 5.5) circle (2pt);
\draw [fill=black] (9, 5.5) circle (2pt);

\end{tikzpicture}
\caption{\textbf{Left:} The input graph---a tree with $k$ extra edges. \textbf{Middle left:} Three kinds of stumps, which are obtained by cutting down dangling trees. \textbf{Middle right:} A small subtrigraph $H$ and long dangling red paths. \textbf{Top right:} If the target twin-width is 2, the paths can be shortened to single vertices. \textbf{Bottom right:} Paths can be shortened to bounded length while retaining a guarantee on the twin-width. A detailed example depicting the first two steps is provided in Figure~\ref{fig:preprocessing}.}
\label{fig:intro}
\end{figure}

\section{Preliminaries}\label{sec:prelim}

For integers $i$ and $j$, we let $[i,j] := \{n \in \mathbb N \sep i \le n \le j\}$ and $[i] := [1, i]$. 
We assume familiarity with basic concepts in graph theory~\cite{Diestel} and parameterized algorithmics~\cite{DowneyF13,CyganFKLMPPS15}. \ifshort ($\clubsuit$)\fi
\iflong
Given a vertex set $U$, we will use $G - U$ to denote the graph $G[V(G)\setminus U]$, and similarly for an edge set $F$, $G-F$ denotes $G$ after removing the edges in $F$. We use $\deg_G(u)$ to denote the degree of $u$ in $G$ and $N_G(u)$ to denote the neighborhood of $u$ in $G$.

An edge $ab$ in $G$ is called a \emph{bridge} if every path from $a$ in $b$ in $G$ uses the edge $ab$; in other words, $G-\{ab\}$ contains one more connected component than $G$.
\fi
\ifshort

\fi
The \emph{length} of a path is the number of edges it contains.
\iflong
The \emph{distance} between two vertices $u$ and $v$, denoted $\dist_G(u,v)$, is the length of the shortest path between them.

\fi
An edge set $F$ in an $n$-vertex graph $G$ is called a \emph{feedback edge set} if $G-F$ is acyclic, and the \emph{feedback edge number} of $G$ is the size of a minimum feedback edge set in $G$.
\iflong
 We remark that a minimum feedback edge set can be computed in time $\bigoh(n)$ as an immediate corollary of the classical (DFS- and BFS-based) algorithms for computing a spanning tree in an unweighted graph $G$. 
 \fi
A \emph{dangling path} in $G$ is a path of vertices which all have degree $2$ in $G$, and a \emph{dangling tree} in $G$ is an induced subtree in $G$ which can be separated from the rest of $G$ by a bridge (see, e.g., Figure~\ref{fig:preprocessing} later).

\smallskip
\noindent \textbf{Twin-Width.}\quad
A \emph{trigraph} $G$ is a graph whose edge set is partitioned into a set of \emph{black} and \emph{red} edges. The set of red edges is denoted $R(G)$, and the set of black edges $E(G)$.
The \emph{black (\emph{resp}.\ red) degree} of $u\in V(G)$ is the number of black (resp.\ red) edges incident to $u$ in $G$.
We extend graph-theoretic terminology to trigraphs by ignoring the colors of edges; for example, the degree of $u$ in $G$ is the sum of its black and red degrees\iflong~(in the literature, this is sometimes called the \emph{total degree})\fi.
We say a (sub)graph is \emph{black (\emph{resp}.\ red)} if all of its edges are black (resp. red); for example, $P$ is a red path in $G$ if it is a path containing only red edges. Without a color adjective, the path (or a different kind of subgraph) may contain edges of both colors. We use $G[Q]$ to denote the subtrigraph of $G$ induced on $Q\subseteq V(G)$. 

Given a trigraph $G$, a \emph{contraction} of two distinct vertices $u,v\in V(G)$ is the operation which produces a new trigraph by (1) removing $u, v$ and adding a new vertex $w$, (2) adding a black edge $wx$ for each $x\in V(G)$ such that $xu$, $xv\in E(G)$, and (3) adding a red edge $wy$ for each $y\in V(G)$ such that $yu\in R(G)$, or $yv\in R(G)$, or $y$ contains only a single black edge to either $v$ or $u$.
A sequence $C = (G = G_1,\ldots,G_n)$ is a \emph{partial contraction sequence of $G$} if it is a sequence of trigraphs such that for all $i\in [n-1]$, $G_{i+1}$ is obtained from $G_i$ by contracting two vertices.
A \emph{contraction sequence} is a partial contraction sequence which ends with a single-vertex graph.
The \emph{width} of a (partial) contraction sequence $C$, denoted $w(C)$, is the maximum red degree over all vertices in all trigraphs in $C$; we also use \emph{$\alpha$-contraction sequence} as a shorthand for a contraction sequence of width at most $\alpha$. 
The \emph{twin-width} of $G$, denoted $\tww(G)$, is the minimum width of any contraction sequence of $G$, and a contraction sequence of width $\tww(G)$ is called \emph{optimal}. An example of a contraction sequence is provided in Figure~\ref{fig:seq}.

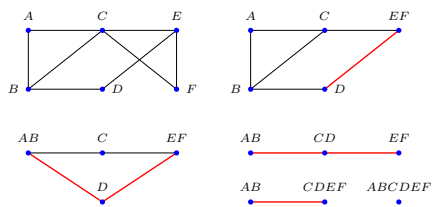
\begin{figure}[ht]
\floatbox[{\capbeside\thisfloatsetup{capbesideposition={left},capbesidewidth=8cm}}]{figure}[\FBwidth]
{
\vspace{-0.4cm}
\scalebox{0.65}{\begin{tikzpicture}[line cap=round,line join=round, x=1.0cm,y=1.0cm]
\draw (2,5)-- (2,3.8);
\draw (3.5,5)-- (2,5);
\draw (3.5,5)-- (2,3.8);
\draw (3.5,3.8)-- (2,3.8);
\draw (5,5)-- (3.5,5);
\draw (5,5)-- (3.5,3.8);
\draw (5,3.8)-- (3.5,5);
\draw (5,3.8)-- (5,5);
\draw (6.5,5)-- (6.5,3.8);
\draw (8,5)-- (6.5,5);
\draw (8,5)-- (6.5,3.8);
\draw (8,3.8)-- (6.5,3.8);
\draw (9.5,5)-- (8,5);
\draw [color=red, thick] (9.5,5)-- (8,3.8);
\draw (3.5,2.5)-- (2,2.5);
\draw (5,2.5)-- (3.5,2.5);
\draw [color=red, thick] (5,2.5)-- (3.5,1.5);
\draw [color=red, thick] (2,2.5)-- (3.5,1.5);
\draw [color=red, thick] (9.5,2.5)-- (8,2.5);
\draw [color=red, thick] (6.5,2.5)-- (8,2.5);
\draw [color=red, thick] (6.5,1.5)-- (8,1.5);
\begin{scriptsize}
\fill [color=blue] (2,5) circle (1.5pt);
\draw[color=blue] (2,5.3) node {\color{black}$A$};
\fill [color=blue] (3.5,5) circle (1.5pt);
\draw[color=blue] (3.5,5.3) node {\color{black}$C$};
\fill [color=blue] (5,5) circle (1.5pt);
\draw[color=blue] (5,5.3) node {\color{black}$E$};

\fill [color=blue] (2,3.8) circle (1.5pt);
\draw[color=blue] (1.7,3.8) node {\color{black}$B$};
\fill [color=blue] (3.5,3.8) circle (1.5pt);
\draw[color=blue] (3.8,3.8) node {\color{black}$D$};
\fill [color=blue] (5,3.8) circle (1.5pt);
\draw[color=blue] (5.3,3.8) node {\color{black}$F$};

\fill [color=blue] (8,3.8) circle (1.5pt);
\draw[color=blue] (8.3,3.8) node {\color{black}$D$};
\fill [color=blue] (6.5,3.8) circle (1.5pt);
\draw[color=blue] (6.2,3.8) node {\color{black}$B$};

\fill [color=blue] (6.5,5) circle (1.5pt);
\draw[color=blue] (6.5,5.3) node {\color{black}$A$};

\fill [color=blue] (8,5) circle (1.5pt);
\draw[color=blue] (8,5.3) node {\color{black}$C$};

\fill [color=blue] (9.5,5) circle (1.5pt);
\draw[color=blue] (9.5,5.3) node {\color{black}$EF$};

\fill [color=blue] (2,2.5) circle (1.5pt);
\draw[color=blue] (2.,2.8) node {\color{black}$AB$};
\fill [color=blue] (3.5,2.5) circle (1.5pt);
\draw[color=blue] (3.5,2.8) node {\color{black}$C$};
\fill [color=blue] (3.5,1.5) circle (1.5pt);
\draw[color=blue] (3.5,1.8) node {\color{black}$D$};
\fill [color=blue] (5,2.5) circle (1.5pt);
\draw[color=blue] (5.,2.8) node {\color{black}$EF$};

\fill [color=blue] (6.5,2.5) circle (1.5pt);
\draw[color=blue] (6.5,2.8) node {\color{black}$AB$};
\fill [color=blue] (8,2.5) circle (1.5pt);
\draw[color=blue] (8.,2.8) node {\color{black}$CD$};
\fill [color=blue] (9.5,2.5) circle (1.5pt);
\draw[color=blue] (9.5,2.8) node {\color{black}$EF$};

\fill [color=blue] (6.5,1.5) circle (1.5pt);
\draw[color=blue] (6.5,1.8) node {\color{black}$AB$};
\fill [color=blue] (8,1.5) circle (1.5pt);
\draw[color=blue] (8.,1.8) node {\color{black}$CDEF$};

\fill [color=blue] (9.5,1.5) circle (1.5pt);
\draw[color=blue] (9.5,1.8) node {\color{black}$ABCDEF$};
\fill [color=white] (9.5,1) circle (1.5pt);
\end{scriptsize}
\end{tikzpicture}
}}
{\caption{A $2$-contraction sequence of the graph depicted on the top left; the sequence consists of $6$ trigraphs and ends at the bottom right. \vspace{1cm}~
\label{fig:seq}
}}
\end{figure}

Let us now fix a contraction sequence $C = (G = G_1,\ldots,G_n)$.
For each $i \in [n]$, we associate each vertex $u \in V(G_i)$ with a set $\beta(u, i) \seq V(G)$, called the \emph{bag} of $u$, which contains all vertices contracted into $u$. \ifshort ($\clubsuit$)\fi
\iflong
Formally, we define the bags as follows:

\begin{itemize}
\item for each $u \in V(G)$, $\beta(u, 1) := \{u\}$;
\item for $i \in [n-1]$, if $w$ is the new vertex in $G_{i+1}$ obtained by contracting $u$ and $v$, then $\beta(w, i+1) := \beta(u, i) \cup \beta(v,i)$; otherwise, $\beta(w, i+1) := \beta(w, i)$. 
\end{itemize}
\fi

Note that if a vertex $u$ appears in multiple trigraphs in $C$, then its bag is the same in all of them, and so we may denote the bag of $u$ simply by $\beta(u)$.
Let us fix $i, j \in [n]$, $i \le j$.
If $u \in V(G_i)$, $v \in V(G_j)$, and $\beta(u) \seq \beta(v)$, then we say that $u$ is \emph{an ancestor} of $v$ in $G_i$ and $v$ is \emph{the descendant} of $u$ in $G_j$ (clearly, the descendant is unique).
If $H$ is an induced subtrigraph of $G_i$, then $u \in V(G_j)$ is a \emph{descendant} of $H$ if it is a descendant of at least one vertex of $H$, and we say that $u \in V(G_i)$ is \emph{contracted to $H$ in $G_j$} if $u$ is an ancestor of a descendent of $H$ in $G_j$.
A contraction of $u, v \in V(G_j)$ into $uv \in V(G_{j+1})$ \emph{involves} $w \in V(G_i)$ if $w$ is an ancestor of $uv$. 

\iflong
The following definition provides terminology that allows us to partition a contraction sequence into ``steps'' based on contractions between certain vertices in the original graph.

\begin{definition}\label{def:restriction}
Let $C$ be a contraction sequence of a trigraph $G$, and let $H$ be an induced subgraph of $G$ with $|V(H)| = m$.
For $i \in [m-1]$, let $C\ang i_H$ be the trigraph in $C$ obtained by the $i$-th contraction between two descendants of $H$, and let $C\ang 0_H = G$. For $i\in [m-1]$, let $U_i$ and $W_i$ be the bags of the vertices which are contracted into the new vertex of $C\ang i_H$.

A contraction sequence $C[H] = (H = H_1, \ldots, H_m)$ is the \emph{restriction of $C$ to $H$} if for each $i \in [m-1]$, $H_{i+1}$ is obtained from $H_{i}$ by contracting the two vertices $u, w \in V(H_{i})$ such that $\beta(u) = U_i \cap V(H)$ and $\beta(w) = W_i \cap V(H)$.
\end{definition}
\fi

Next, we introduce a notion that will be useful when dealing with reduction rules in the context of computing contraction sequences.

\begin{definition}\label{def:effective}
Let $G, G'$ be trigraphs. We say that the twin-width of $G'$ is \emph{effectively} at most the twin-width of $G$, denoted $\tww(G') \le_e \tww(G)$, if (1) $\tww(G') \le \tww(G)$ and (2) given a contraction sequence $C$ of $G$, a contraction sequence $C'$ of $G'$ of width at most $w(C)$ can be constructed in polynomial time.
If $\tww(G') \le_e \tww(G)$ and $\tww(G) \le_e \tww(G')$, then we say that the two graphs have \emph{effectively} the same twin-width, $\tww(G') =_e \tww(G)$.
\end{definition}

We say that $G'$ is a \emph{pseudoinduced} subtrigraph of $G$ if $G'$ is obtained from an induced subtrigraph of $G$ by the removal of red edges or their replacement with black edges.

\iflong
\begin{observation}
\fi
\ifshort
\begin{observation}[$\clubsuit$]
\fi
\label{obs:induced}
If $G'$ is a pseudoinduced subtrigraph of $G$, then $\tww(G') \le_e \tww(G)$.
\end{observation}
\iflong \begin{proof}
Given a contraction sequence $C$ of $G$, simply take the restriction of $C$ to $G[V(G')]$.
\end{proof} \fi

\smallskip
\noindent \textbf{Preliminary Observations and Remarks.}\quad
We begin by stating a simple brute-force algorithm for computing twin-width.

\iflong
\begin{observation}
\fi
\ifshort
\begin{observation}[$\clubsuit$]
\fi
\label{obs:brute-force}
An optimal contraction sequence of an $n$-vertex graph can be computed 
in time $2^{\bigoh(n\cdot \log n)}$.
\end{observation}
\iflong \begin{proof}
Each contraction sequence is defined by $n-1$ choices of a pair of vertices, and so the number of contraction sequences is $\bigoh((n^2)^n) = \bigoh(2^{2n\cdot \log n}) \le 2^{\bigoh(n\cdot \log n)}$. Moreover, computing the width of a contraction sequence can clearly be done in polynomial time.
\end{proof} \fi

The following observation not only establishes the twin-width of trees---which form a baseline case for our algorithms---but also notes that the necessary contractions are almost entirely independent of $r$.

\begin{observation}[{\cite[Section~3]{BonnetKTW22}}]\label{obs:contract-trees}
For any rooted tree $T$ with root $r$, there is a contraction sequence $C$ of $T$ of width at most $2$ such that the only contraction involving $r$ is the very last contraction in $C$.
\end{observation}

Next, we recall that a $1$-contraction sequence can be computed in polynomial time not only on graphs, but also on trigraphs with at most one red edge.

\begin{theorem}[{\cite[Section 7]{DBLP:journals/algorithmica/BonnetKRTW22}}]
\label{thm:tww1-deciding}
If $G$ is a trigraph with at most one red edge, then it can be decided in polynomial time whether the twin-width of $G$ is at most 1. In the positive case, the algorithm also returns an optimal contraction sequence of $G$.
\end{theorem}

Finally, we formalize the claim made in Section~\ref{sec:intro} that there 
can be no ``simple'' twin-width preserving reduction rule for handling long dangling paths; in particular, any such rule for simplifying dangling paths cannot depend purely on the length of the path itself.

\iflong
\begin{proposition}
\fi
\ifshort
\begin{proposition}[$\clubsuit$]
\fi
\label{prop:dangling}
For every integer $c\geq 1$, there exists a graph $G^c$ with the following properties: (1) $G^c$ contains a dangling path $P$ of length $c$, (2) $\tww(G^c)\geq 5$, and (3) the graph obtained by subdividing one edge in $P$ (i.e., replacing $P$ with a path $P'$ whose length is $c+1$) has twin-width at most $4$.
\end{proposition}

\iflong \begin{proof}
By setting $d:=4$ and $k:=c$ in Proposition~28 stated in the second paper on twin-width~\cite{DBLP:conf/soda/BonnetGKTW21}, we obtain that the $c$-th subdivision $G_1$ of $K_n$ for some sufficient large $n$ (where $n$ is exponential in $c$) has twin-width at least $5$. At the same time, it is known that the $2\log n$-th subdivision $G_2$ of $K_n$ has twin-width at most $4$~\cite[Theorem 2]{DBLP:conf/icalp/BergeBD22}. Thus, if we look at the sequence of graphs obtained from $G_1$ by progressively subdividing the dangling paths between the initial vertices of the $K_n$ until we obtain $G_2$, there must exist a final graph in this sequence which has twin-width $5$. We take this graph as $G^c$, and observe that it clearly satisfies all the stated requirements.
  \end{proof} \fi
  
  Finally, we remark that throughout the paper, we assume the input graph $G$ to be connected; this is without loss of generality, since otherwise one can handle each of the graph's connected components separately.

\section{Cutting Down a Forest}\label{sec:cutting-forest}

After computing a minimum feedback edge set of an input graph $G$, the first task on our route towards Theorems~\ref{thm:tww2} and~\ref{thm:tww-3+} is to devise reduction rules which can safely deal with dangling trees. While it is not possible to delete such trees entirely while preserving twin-width, we show that they can be safely ``cut down''; in particular, depending on the structure of the tree it can be replaced by one of the two kinds of stumps defined below.

\begin{definition}
Let $G$ be a trigraph, and let $u,v,w \in V(G)$. We say that $u$ has a \emph{half stump} $v$ if $uv \in E(G)$, and the degree of $v$ in $G$ is 1. We say that $u$ has a \emph{red} (resp. \emph{black}) \emph{stump} $vw$ if $uv \in E(G)$ and $vw \in R(G)$ (resp. $vw \in E(G)$), and the degrees of $v$ and $w$ are 2 and 1, respectively. Half, red, and black stumps are collectively called \emph{stumps}. The stump $vw$ (or $v$) then \emph{belongs} to $u$.
\end{definition}

We begin by observing that special kinds of subtrees---specifically, stars---can be safely replaced with just a black stump.

\iflong
\begin{observation}
\fi
\ifshort
\begin{observation}[$\clubsuit$]
\fi
\label{obs:onestar}
Let $G$ be a trigraph with a dangling tree $T$ connected to the rest of the graph via the bridge $e=uv$ where $v\in V(T)$. Assume that $T$ is a black star consisting of at least one vertex other than $v$. 
Then $G$ has effectively the same twin-width as the trigraph $G'$ obtained from $G-T$ by adding a black stump to $u$.
\end{observation}

\iflong \begin{proof}
Notice that $\tww(G') \le_e \tww(G)$ by Observation~\ref{obs:induced}. For the converse, let us consider a contraction sequence $C'$ of $G'$. We define a contraction sequence $C$ of $G$ of width at most $w(C')$. First, $C$ contracts all neighbors of $v$ in $T$ (which are all twins) into a single vertex. The obtained trigraph is isomorphic to $G'$, and so from now on, $C$ follows $C'$. This establishes that $\tww(G) \le_e \tww(G')$, as desired.
  \end{proof} \fi

Next, we show that all dangling trees not covered by Observation~\ref{obs:onestar} can be safely cut down to a red stump, as long as the trigraph obtained by the cutting has twin-width at least $2$ (a condition that will be handled later on). The proof of this case is significantly more difficult than the previous one.

\iflong \begin{lemma} \fi \ifshort \begin{lemma}[$\clubsuit$]\fi
\label{lem:onebigtree}
Let $G$ be a trigraph with a bridge $e=uv$ such that the connected component $T$ of $G-\{e\}$ containing $v$ is a black dangling tree which contains a vertex at distance $2$ from $v$ and let $G'$ be the trigraph obtained from $G - T$ by adding a red stump to $u$. If $\tww(G) \ge 2$ and $\tww(G') \ge 2$, then $G$ and $G'$ have effectively the same twin-width.
\end{lemma}
\ifshort
\begin{proof}[Proof Sketch]
We begin by establishing $\tww(G) \le_e \tww(G')$. Let $C'$ be a contraction sequence of $G'$, and we will show how to construct a contraction sequence $C = (G = G_0, G_1,\ldots)$ of $G$ of width at most $w(C')$. %
$C$ starts with contracting $T$, following the contraction sequence $C_T$ given by Observation~\ref{obs:contract-trees} (with $v$ as the root), but stops before the first contraction involving $v$ (which is the very last contraction in $C_T$).
Let $G_i$ be the obtained trigraph. By definition of $C_T$, no vertex of $G_j$, $j \in [i]$, exceeds the bound on its red degree since $w(C_T) \le 2 \le \tww(G')$.
$G_i$ is isomorphic to $G'$ and so from here on, $C$ follows~$C'$. %

Our task in the remainder of the proof will be to establish $\tww(G') \le_e \tww(G)$.
Let $C = (G, G_1,\ldots)$ be a contraction sequence of $G$, and let $w, x \in V(T) \setminus\{v\}$ be such that $vw, wx \in E(T)$.
We need to construct a contraction sequence of $G'$ of width at most $w(C)$; note that there can be vertices with red degree up to 2 in this desired sequence since $w(C) \ge \tww(G)\ge 2$.
Let $G^- := G[V(H)\cup\{v, w\}]$.
Observe that the only difference between $G'$ and $G^-$ is the color of the edge $vw$ (it is red in $G'$, but black in $G^-$) and let $C^- = (G^- = G^-_0, G_1^-, \ldots, G_{m}^-)$ be the restriction of $C$ to $G^-$; hence $w(C^-)\leq w(C)$.
Let us consider the contraction sequence $C' = (G' = G'_0, G_1', \ldots, G_m')$ of $G'$ which follows $C^-$ in each step (i.e., $V(G_i^-) = V(G_i')$ for all $i \in [m]$). To avoid any confusion, we explicitly note that it is not possible to rule out $w(C')>w(C^-)$. We complete the proof by performing a case distinction that will allow us to either guarantee $w(C') \le w(C^-)$, or---in the most difficult case---construct a new contraction sequence $C''$ or $G'$ such that $w(C'') \le w(C)$.
\end{proof}
\fi
\iflong \begin{proof}
We begin by establishing $\tww(G) \le_e \tww(G')$. Let $C'$ be a contraction sequence of $G'$, and we will show how to construct a contraction sequence $C = (G = G_0, G_1,\ldots)$ of $G$ of width at most $w(C')$. %
$C$ starts with contracting $T$, following the contraction sequence $C_T$ given by Observation~\ref{obs:contract-trees} (with $v$ as the root), but stops before the first contraction involving $v$ (which is the very last contraction in $C_T$).
Let $G_i$ be the obtained trigraph. By definition of $C_T$, no vertex of $G_j$, $j \in [i]$, has too high red degree since $w(C_T) \le 2 \le \tww(G')$.
Observe that $G_i$ is isomorphic to $G'$, and so from now on, $C$ follows $C'$. %

Our task in the remainder of the proof will be to establish $\tww(G') \le_e \tww(G)$.
Let $C = (G, G_1,\ldots)$ be a contraction sequence of $G$, and let $w, x \in V(T) \setminus\{v\}$ be such that $vw, wx \in E(T)$.
We need to construct a contraction sequence of $G'$ of width at most $w(C)$; note that there can be vertices with red degree up to 2 in this desired sequence since $w(C) \ge \tww(G)\ge 2$.
Let $G^- := G[V(H)\cup\{v, w\}]$.
Observe that the only difference between $G'$ and $G^-$ is the color of the edge $vw$ (it is red in $G'$, but black in $G^-$) and let $C^- = (G^- = G^-_0, G_1^-, \ldots, G_{m}^-)$ be the restriction of $C$ to $G^-$; hence $w(C^-)\leq w(C)$.
Let us consider the contraction sequence $C' = (G' = G'_0, G_1', \ldots, G_m')$ of $G'$ which follows $C^-$ in each step (i.e., $V(G_i^-) = V(G_i')$ for all $i \in [m]$). To avoid any confusion, we explicitly remark that it could happen that $w(C')>w(C^-)$. We complete the proof by performing a case distinction that will allow us to either guarantee $w(C') \le w(C^-)$, or---in the most difficult case---construct a new contraction sequence $C''$ or $G'$ such that $w(C'') \le w(C)$.

Let $i \in [m]$ be the smallest index s.t.\ $G_i'$ (or, equivalently, $G_i^-$) is obtained by a contraction involving at least one element of $\{u,v,w\}$.
Observe that for all $j \in [i-1]$, all descendants of $H$ have the same red degree in $G_j'$ as in $G_j^-$, and $v$ and $w$ have red degree exactly 1 in $G_j'$.
Let $y \in V(G_{i-1}')\setminus\{u,v,w\}$.
If $G_i'$ is obtained by contracting $v$ and $w$, $u$ and $v$, $v$ and $y$, or $w$ and $y$, then $G_i' = G_i^-$ because the edge $vw$ either disappears (in the first case) or becomes red (in the latter three cases). By induction, $G_j' = G_j^-$ for all $j \in [i, m]$. Thus, in these four cases, we have $w(C') \le w(C^-) \le w(C)$.

Suppose $G_i'$ is obtained by contracting $u$ and $y$. In this case, all descendants of $H$ have the same red degree in $G_i'$ as in $G_i^-$, $v$ has red degree $2$, and $w$ has red degree 1 in $G_i'$ (which, as argued above, does not affect $w(C')$).
Now let $j \in [i+1, m]$ be the smallest index such that $G_j'$ is obtained by a contraction involving $v$ or $w$. Observe that $G_j'$ is obtained by contracting $v$ and $w$, $v$ and $z$, or $w$ and $z$ for some $z \in V(G_{j-1}')\setminus\{v,w\}$.
In all these case, $G_j' = G_j^-$; indeed, similarly as in the cases described in the previous paragraph, the edge $vw$ either becomes red in $G_j^-$ or disappears in both trigraphs (notice that in the case of contracting $w$ with the descendant of $u$, this holds only thanks to the previous contraction of $u$ and $y$). Thus, in this case we also conclude that $w(C') \le w(C^-)$, as desired.

Finally, suppose $G_i'$ is obtained by contracting $u$ and $w$ into a new vertex which we call $u^*$. This case is more involved because $u^*v$ is a red edge in $G_i'$ but black in $G_i^-$, and the red degree of $u^*$ in $G_i^-$ could be as large as $\tww(G_i^-)$.
Let $d$ be the red degree of $u^*$ in $G_{i}^-$, and observe that the red degree of $u^*$ in $G_{i}'$ is $d+1$. %
We now recall the notation from Definition~\ref{def:restriction}, and turn our attention back to the original contraction sequence $C$: let $C\ang j := C\ang j _{H'}$ for all $j \in [0, m]$, where $H' := G[V(H) \cup \{v, w\}]$. Moreover, recalling that $w$ has a neighbor $x$ distinct from $v$, we will now consider the descendant $x^*$ of $x$ in $C\ang i$.

First, suppose $x^*$ is not contracted to any vertex of $H$ in $C \ang i$. This means that the red degree of $u^*$ in $C\ang i$ is $d+1$ because the contraction of $u$ and $w$ creates a red edge $u^*x$ in $C\ang i$ which is not present in $G_{i}^-$. Thus, $d+1 \le w(C) = \tww(G)$.
We prove that $w(C') \le \max\{w(C'), d+1)\} \le w(C)$.
Let $j \in [i+1, m]$ be the smallest index such that $G_j'$ is obtained by a contraction involving $u^*$ or $v$. Notice that $G_j' = G_j^-$ because the edge $u^*v$, which is black in all trigraphs of $C^-$ before $G_j^-$, becomes red in $G_j^-$ or disappears.
Thus, it remains only to deal with trigraphs $G_k'$ for $k \in [i, j-1]$. However, the only vertices which have higher red degree in $G_k'$ than in $G_k^-$ are $v$, which has red degree 1 in $G_k'$, and $u^*$, which has red degree at most $d+1$ since $u^*$ has no black neighbors in $G_i'$, and so the red degree of $u^*$ may not increase by contractions not involving $u^*$. 

Second, suppose $x^*$ is contracted to some vertex of $H$ in $C \ang i$. Let $j \in [i-1]$ be the largest index such that $x^*$ is not contracted to $H$ in $C \ang j$. 
Let $y \in V(C \ang j)$ be the descendant of $H$ with which $x^*$ is contracted in $C\ang {j+1}$. We define a new contraction sequence $C'' = (G' = G''_0, G''_1, \ldots)$ of $G'$ as follows: 

\begin{enumerate}
\item for the first $j$ contractions, $C''$ follows $C'$ (or, equivalently, $C^-$). 
\item Then, $G''_{j+1}$ is obtained by contracting $w$ with $y$ (or, more precisely, with the unique vertex which is the descendant of the same vertices of $H$ in $G''_j$ as $y$ in $C \ang j$; we ignore this subtle difference in the rest of the proof as it does not impact any of the arguments). Informally, we view this contraction as being ``inserted'' into $C'$, and now $C''$ is shifted by one relative to $C'$. Formally, for $k \in [j+1, i-1]$ (whereas this set may be empty), $G''_{k+1}$ is obtained from $G''_k$ by performing the same contraction (between descendants of $H$ by definition of $i$) as the one performed to obtain $G'_k$ from $G'_{k-1}$. 
\item Next, $G''_{i+1}$ is obtained by contracting $v$ to $u$. Observe that $G''_{i+1}$ is a pseudoinduced subtrigraph of $C \ang i$ (informally, $w$ played the role originally played by $x^*$, and $v$ played the role originally played by $w$).
\item Finally, we invoke Observation~\ref{obs:induced} and complete $C''$ by following the restriction of $C$ from $C \ang i$ onward.
\end{enumerate}
To argue correctness, it suffices to consider the red degrees in $G''_k$ for $k \in [j+1, i]$. In $G''_k$, $v$ has red degree 1 (its only red neighbor is the descendant of $w$), %
and all descendants of $H$ have a red degree which is upper-bounded by the red degree of their counterparts in $C \ang {k-1}$ by construction (we note that the red degrees may be higher on the latter because in $C \ang {k-1}$, many vertices of $T$ may be contracted to $H$, whereas in $G''_k$, only $y$ has received such a contraction---albeit from $w$ instead of $x^*$). 
Thus, $w(C'') \le w(C)$, as desired.

The six cases above are exhaustive and in each case we have identified a contraction sequence (either $C'$ or $C''$) of width at most $w(C)$. Moreover, for each direction of establishing that the effective twin-width of the two trigraphs is the same, the proof guarantees the existence of a polynomial-time algorithm for constructing the relevant contraction sequences (assuming that the sequence for the other graph is provided). Therefore, $G$ has effectively the same twin-width as the trigraph $G'$.
  \end{proof} \fi

Let us now provide some intuition on how we aim to apply Lemma~\ref{lem:onebigtree} (a formalization is provided in the proof of Theorem~\ref{thm:pruning} at the end of this section). Assume without loss of generality that the input graph $G$ has twin-width at least $2$. The first time we want to apply Lemma~\ref{lem:onebigtree} to go from $G$ to $G'$, we can verify that $\tww(G')\geq 2$ by Theorem~\ref{thm:tww1-deciding}; if this check fails then we show how to construct a $2$-contraction sequence of $G$, and otherwise we replace $T$ with a red stump as per the lemma statement.
For every subsequent application of Lemma~\ref{lem:onebigtree}, $G'$ will have two red edges and hence we cannot rely on Theorem~\ref{thm:tww1-deciding} anymore---but instead, we can guarantee that the condition on $G'$ holds by establishing the following lemma.

\iflong \begin{lemma} \fi \ifshort \begin{lemma}[$\clubsuit$]\fi
\label{lem:twostumps}
Let $G$ be a connected trigraph with two red stumps. Then $\tww(G)\geq 2$.
\end{lemma}
\iflong \begin{proof}
Let us first remark that, due to $G$ being connected, it contains an induced path (possibly on one vertex) between the two vertices the stumps belong to, and thus contains an induced path whose first and last edge is red and which consists of at least $5$ vertices (we add the stumps at both ends of the initial path). This means that we only need to lower bound the twin-width of such a path to prove the claimed result.

We now prove by induction that any path on $p\geq 3$ vertices whose first and last edge is red has twin-width at least $2$. 
The base case with $p=3$ is trivial, since the initial trigraph has a vertex of red degree $2$. 
Let $p\geq 4$, consider the first contraction in an optimal contraction sequence. It is easy to check that the only contractions not creating a vertex with red degree at least $2$ are contracting one end-point of the path with its unique neighbor. However, the trigraph obtained after this first contraction is exactly the path on $p-1$ vertices with the first and last edge being red, and the statement thus holds by the induction hypothesis.
  \end{proof} \fi

With Observation~\ref{obs:onestar} and Lemmas~\ref{lem:onebigtree}-\ref{lem:twostumps}, we can effectively preprocess (tri)graphs by ``cutting down'' all dangling trees (i.e., replacing them with stumps). However, we will also need to deal with the fact that a vertex could now be connected to many distinct stumps. For half stumps this is not an issue, as multiple half stumps can be assumed to be contracted into a single half stump due to the vertices in them being twins. For all pairs of other kinds of stumps (except for a pair consisting of a half and a black stump), we show that it is sufficient to replace these with a single stump instead.

\iflong
\begin{observation}
\fi
\ifshort
\begin{observation}[$\clubsuit$]
\fi
\label{obs:redstumps}
Let $G'$ be a trigraph such that $\tww(G') \ge 2$, let $u \in V(G')$ be a vertex with a red stump $S$ in $G'$, and let $G$ be a trigraph obtained from $G'$ by adding an additional stump to $u$. Then $G'$ and $G$ have effectively the same twin-width.
\end{observation}

\iflong \begin{proof}
The fact that $\tww(G')\le_e \tww(G)$ follows immediately from the fact that $G'$ is an induced subtrigraph of $G$. For the converse, let us consider an arbitrary contraction sequence $C'$ of $G'$, and let the red stump $S$ consist of a vertex $v$ adjacent to $u$ and a vertex $w$ adjacent to $v$. We construct a contraction sequence $C$ of $G$ such that $w(C')=w(C)$ as follows. If $G$ was created by adding a half stump to $u$, we contract it into $v$. Otherwise, $G$ was created from $G'$ by adding a pendant vertex $v_0$ to $u$ and then another pendant vertex $w_0$ to $v_0$. In this case, we contract $w_0$ into $w$ and then immediately contract $v_0$ into $v$. None of these contractions result in a red degree greater than $2$. At this point, we obtain a trigraph isomorphic to $G'$ and hence simply follow $C'$.
  \end{proof} \fi

\iflong \begin{lemma} \fi \ifshort \begin{lemma}[$\clubsuit$]\fi\label{lem:two-black-one-red}
Let $G'$ be a trigraph, let $u \in V(G')$ be a vertex with a red stump $S$ in $G'$, and let $G$ be a trigraph obtained from $G'$ by removing $S$ and adding two black stumps to $u$. If $\tww(G) \ge 2$ and $\tww(G') \ge 2$, then $G'$ and $G$ have effectively the same twin-width.
\end{lemma}

\iflong \begin{proof}
We begin by establishing $\tww(G)\le_e \tww(G')$, which follows by the same argument as Observation~\ref{obs:redstumps}. Indeed, let us consider an arbitrary contraction sequence $C'$ of $G'$, and let the red stump $S$ consist of a vertex $v$ adjacent to $u$ and a vertex $w$ adjacent to $v$. Let the black stumps of $u$ in $G$ consist of vertices $v_0, w_0, v_1, w_1$ and edges $uv_0,v_0w_0,uv_1,v_1w_1$. We construct a contraction sequence $C$ of $G$ such that $w(C')=w(C)$ by first contracting $w_0$ with $w_1$, then $v_0$ with $v_1$. Neither of these contractions creates a vertex of red degree exceeding $2$. At this point, we notice that the trigraph is isomorphic to $G'$ and hence complete the construction of $C$ by following the contraction sequence $C'$.

We now turn to establishing $\tww(G') \le_e \tww(G)$.
Let $C = (G, G_1,\ldots)$ be a contraction sequence of $G$.
We need to construct a contraction sequence of $G'$ of width at most $w(C)$; note that there can be vertices with red degree up to 2 in this desired sequence since $w(C) \ge \tww(G)\ge 2$.
Our approach now will roughly follow the one we used in the proof of Lemma~\ref{lem:onebigtree}, albeit the details will differ. 
Without loss of generality, we assume that the first contraction in $C$ of a vertex from $\{v_0,w_0,v_1,w_0\}$ contracts at least one vertex from $\{v_0,w_0\}$; indeed, if this is not the case then we can simply swap the names of the two black stumps to ensure the property holds. 
Let $G^- := G[V(H)\cup\{v_0, w_0\}]$, and let us hereinafter identify the vertex $v_0$ (and $w_0$) in $G^-$ with the vertex $v$ (and $w$) in $G'$. Observe that now the only difference between $G'$ and $G^-$ is the color of the edge $vw$ (it is red in $G'$, but black in $G^-$). Let $C^- = (G^- = G^-_0, G_1^-, \ldots, G_{m}^-)$ be the restriction of $C$ to $G^-$.
We now show that either the contraction sequence $C' = (G' = G'_0, G_1', \ldots, G_m')$ of $G'$ which follows $C^-$ in each step (i.e., $V(G_i^-) = V(G_i')$ for all $i \in [m]$) satisfies $w(C') \le w(C^-)\le w(C)$, or we can use $C'$ to construct a new contraction sequence $C''$ of $G'$ such that $w(C'') \le w(C)$.

Let $i \in [m]$ be the smallest index s.t.\ $G_i'$ (or, equivalently, $G_i^-$) is obtained by a contraction involving at least one element of $\{u,v,w\}$.
Observe that for all $j \in [i-1]$, all descendants of $H$ have the same red degree in $G_j'$ as in $G_j^-$, and $v$ and $w$ have red degree exactly 1 in $G_j'$. We now perform a case distinction analogous to the one used in the proof of Lemma~\ref{lem:onebigtree}, whereas the first five cases can be handled in exactly the same way as in that proof (we provide the arguments below only for completeness).

Let $y \in V(G_{i-1}')\setminus\{u,v,w\}$.
If $G_i'$ is obtained by contracting $v$ and $w$, $u$ and $v$, $v$ and $y$, or $w$ and $y$, then $G_i' = G_i^-$ because the edge $vw$ either disappears (in the first case) or becomes red (in the latter three cases). By induction, $G_j' = G_j^-$ for all $j \in [i, m]$. Thus, in these four cases, we have $w(C') \le w(C^-) \le w(C)$. For the fifth case, suppose $G_i'$ is obtained by contracting $u$ and $y$. Here, all descendants of $H$ have the same red degree in $G_i'$ as in $G_i^-$, $v$ has red degree $2$, and $w$ has red degree 1 in $G_i'$ (which, as argued above, does not affect $w(C')$).
Now let $j \in [i+1, m]$ be the smallest index such that $G_j'$ is obtained by a contraction involving $v$ or $w$, and let $z \in V(G_{j-1}')\setminus\{v,w\}$. Observe that $G_j'$ is obtained by contracting $v$ and $w$, $v$ and $z$, or $w$ and $z$.
In all these case, $G_j' = G_j^-$; indeed, similarly as in the cases described in the previous paragraph, the edge $vw$ either becomes red in $G_j^-$ or disappears in both trigraphs (notice that in the case of contracting $w$ with the descendant of $u$, this holds only thanks to the previous contraction of $u$ and $y$). Thus, in this case we also conclude that $w(C') \le w(C^-)$, as desired.

For the final sixth case, suppose $G_i'$ is obtained by contracting $u$ and $w$ into a new vertex which we call $u^*$. This case again requires us to change $C'$, since $u^*v$ is a red edge in $G_i'$ but black in $G_i^-$ and the red degree of $u^*$ in $G_i^-$ could be as large as $\tww(G_i^-)$. To this end, we define a new contraction sequence $C''=(G'=G''_0,G''_1,\dots,G''_m)$ of $G'$ as follows:

\begin{enumerate}
\item for the first $i-1$ contractions, $C''$ follows $C$ (or, equivalently, $C^-$). 
\item Then, instead of contracting $u$ and $w$, the trigraph $G''_{i}$ is obtained by contracting $u$ and $v$ (let us call the obtained vertex $u^*$). Crucially, we now observe that $G''_{i}$ is isomorphic to a pseudoinduced subtrigraph of $G_i$. Indeed, since $G''_{i-1}=G_{i-1}$, we can map descendants of vertices in $V(H)$ in a one-to-one fashion between $G''_i$ and $G_i$, while the vertex $w$ in $G''_i$ is mapped to the vertex $v_1$ in $G_i$. Then the edges of $G''_{i}$ not incident to $u^*$ or $w$ must also exist in $G_{i}[V(H)\cup \{v_1\}]$ due to the fact that $G''_{i-1}=G_{i-1}$, the red edge $u^*w$ in $G''_i$ corresponds to the red edge $u^*v_1$ in $G_{i}[V(H)\cup \{v_1\}]$, and every other edge incident to $u^*$ in $G''_i$ must be red due to the contraction of $u$ with $v$ in $C''$ but is also red in $G_{i}[V(H)\cup \{v_1\}]$ due to the contraction of $u$ with $w$ in $C$.
\item Since $G''_{i}$ is isomorphic to the pseudoinduced subtrigraph of $G_i$ described above, we can complete $C''$ by following a restriction of $C$, specifically to $V(H)\cup \{v_1\}$ by treating $w\in V(G''_i)$ as $v_1\in V(G_i)$.
\end{enumerate}

The above construction ensures that $w(C'')\leq w(C)$. Hence, in all cases we have $\tww(G') \le_e \tww(G)$ and the lemma follows.
  \end{proof} \fi

We conclude this section by formalizing the trigraph that can be obtained through the exhaustive application of the reduction rules arising from Observations~\ref{obs:onestar},~\ref{obs:redstumps} and Lemmas~\ref{lem:onebigtree},~\ref{lem:two-black-one-red}. We say that an induced subtrigraph in $G$ is a \emph{dangling pseudo-path} if it can be obtained from a dangling path $P$ in $G$ by adding, to each of the vertices in $P$, either (a) one red stump or (b) at most one black stump and at most one half-stump.

\begin{definition}\label{def:more-paths}
A connected trigraph $G$ with $\tww(G) \ge 2$ is an $(H,\ca P)$-graph if:

\begin{itemize}
\item the vertices of $G$ are partitioned into the induced subtrigraphs $H$ and $\sump\ca P$, and
\item $\ca P$ is a set of dangling pseudo-paths $\{P_1, \ldots, P_n\}$.
\end{itemize}
\end{definition}

We proceed with some related terminology that will be used extensively in the subsequent sections.
Let $\sump \ca P$ denote the disjoint union of the paths in $\ca P$. A vertex 
$u \in V(H)$ is a \emph{connector in $G$} if $u$ is adjacent to a vertex of $\sump\ca P$ in $G$.
We say that $P \in \ca P$ is \emph{original} if all edges in $P$ that do not belong to a stump are black and the edges connecting the endpoints of $P$ to $H$ are also black.
Later, we will also deal with \emph{tidy} paths in $\ca P$, where $P \in \ca P$ is tidy if $P$ is a dangling red path (i.e., contains no stumps) which additionally satisfies the following three technical conditions for each connector $u \in V(H)$ adjacent to an endpoint $v$ of $P$:\begin{enumerate}
\item $u$ has black degree 0;
\item $v$ is the only neighbor of $u$ in $\sump\ca P$; and
\item $u$ has a unique neighbor $u'$ in $V(H)$ and $u'$ has positive black degree.
\end{enumerate}

We say that an $(H, \ca P)$-graph is \emph{original} (\emph{tidy}) if all paths $P \in \ca P$ are original (tidy, respectively). An illustration of these notions is provided in Figure~\ref{fig:preprocessing}, which also showcases the outcome of applying the culmination of this section---Theorem~\ref{thm:pruning}---on an input graph.

\begin{figure}
\begin{subfigure}{0.3\linewidth}
    \centering
    \rotatebox{90}{
\begin{tikzpicture}[line cap=round,line join=round,>=triangle 45,x=1.0cm,y=1.0cm,scale=0.9]
\clip(2.4,0.5) rectangle (10,4);
\draw [line width=1pt,color=ffxfqq] (4.,1.)-- (5.5,1.);
\draw [line width=1pt,color=ffxfqq] (6.5,1.)-- (8.,1.);
\draw  (4.,1.)-- (3.5,3.);
\draw  (8.,1.)-- (8.5,3.);

\draw  (3.5,3.)-- (8.5,3.);

\draw  (3.5,3.)-- (5.5,1.);
\draw  (6.5,1.)-- (8.5,3.);
\draw  (4.5,3.)-- (5.,3.);
\draw  (7.5,3.)-- (7.,3.);
\draw  (4.,1.)-- (3.,1.5);
\draw  (4.,3.)-- (4.,3.5);
\draw  (8.,1.)-- (9.,1.5);
\draw  (8.,3.)-- (8.,3.5);
\draw  (3.,1.5)-- (3.,2.);
\draw  (3.,1.5)-- (2.5,2.);
\draw  (3.,1.5)-- (2.5,1.5);
\draw  (3.,1.5)-- (2.5,1.);
\draw  (3.,1.5)-- (3.,1.);
\draw  (4.,1.)-- (4.5,1.5);
\draw  (9.,1.5)-- (9.,2.);
\draw  (9.,2.)-- (9.5,2.);
\draw  (9.,2.)-- (9.,2.5);
\draw  (9.5,2.)-- (9.5,1.5);
\draw  (9.5,1.5)-- (9.5,1.);
\draw  (9.,1.5)-- (9.,1.);
\draw  (4.,3.5)-- (3.5,3.5);
\draw  (4.,3.)-- (4.5,2.5);
\draw  (5.,3.)-- (5.,2.5);
\draw  (5.,2.5)-- (5.5,2.5);
\draw  (5.5,2.5)-- (5.5,2.);
\draw  (5.5,2.)-- (6.,2.);

\draw  (4.,3.5)-- (4.5,3.5);
\draw  (5.,3.)-- (5.,3.5);
\draw  (8.,3.5)-- (7.5,3.5);
\draw  (7.5,3.5)-- (7.,3.5);
\draw  (6.5,3.)-- (6.5,2.5);
\draw  (7.,3.)-- (7.,2.5);
\draw  (6.5,2.5)-- (6.5,2.);

\draw  (6.5,3.)-- (6.5,3.5);
\draw  (6,3.5)-- (6.5,3.5);

\foreach \i in {0.5,0.75,...,4.5}
{
    \draw [fill=black] (3.5+\i ,3.) circle (1.75pt);
        
}

\begin{scriptsize}
\draw [fill=black] (4.,1.) circle (2.5pt);
\draw [fill=black] (5.5,1.) circle (2.5pt);
\draw [fill=black] (6.5,1.) circle (2.5pt);
\draw [fill=black] (8.,1.) circle (2.5pt);
\draw [fill=black] (3.5,3.) circle (2.5pt);
\draw [fill=black] (8.5,3.) circle (2.5pt);

\draw [fill=ududff] (3.,1.5) circle (2pt);
\draw [fill=ududff] (4.,3.5) circle (2pt);
\draw [fill=ududff] (9.,1.5) circle (2pt);
\draw [fill=ududff] (8.,3.5) circle (2pt);
\draw [fill=ududff] (3.,2.) circle (2pt);
\draw [fill=ududff] (2.5,2.) circle (2pt);
\draw [fill=ududff] (2.5,1.5) circle (2pt);
\draw [fill=ududff] (2.5,1.) circle (2pt);
\draw [fill=ududff] (3.,1.) circle (2pt);
\draw [fill=ududff] (4.5,1.5) circle (2pt);
\draw [fill=ududff] (9.,2.) circle (2pt);
\draw [fill=ududff] (9.5,2.) circle (2pt);
\draw [fill=ududff] (9.,2.5) circle (2pt);
\draw [fill=ududff] (9.5,1.5) circle (2pt);
\draw [fill=ududff] (9.5,1.) circle (2pt);
\draw [fill=ududff] (9.,1.) circle (2pt);
\draw [fill=ududff] (3.5,3.5) circle (2pt);
\draw [fill=ududff] (4.5,2.5) circle (2pt);
\draw [fill=ududff] (5.,2.5) circle (2pt);
\draw [fill=ududff] (5.5,2.5) circle (2pt);
\draw [fill=ududff] (5.5,2.) circle (2pt);
\draw [fill=ududff] (6.,2.) circle (2pt);

\draw [fill=ududff] (4.5,3.5) circle (2pt);
\draw [fill=ududff] (5.,3.5) circle (2pt);
\draw [fill=ududff] (7.5,3.5) circle (2pt);
\draw [fill=ududff] (7.,3.5) circle (2pt);
\draw [fill=ududff] (6.5,2.5) circle (2pt);
\draw [fill=ududff] (7.,2.5) circle (2pt);
\draw [fill=ududff] (6.5,2.) circle (2pt);

\draw [fill=ududff] (6.5,3.5) circle (2pt);
\draw [fill=ududff] (6,3.5) circle (2pt);
\end{scriptsize}

\end{tikzpicture}
}
    \label{fig:subfig1}
  \end{subfigure}
  \hfill
  \begin{subfigure}{0.3\linewidth}
    \centering
\rotatebox{90}{
\begin{tikzpicture}[line cap=round,line join=round,>=triangle 45,x=1.0cm,y=1.0cm,scale=0.9]
\clip(2.4,0.5) rectangle (10.,4.);

\draw (4.,1.)-- (5.5,1.);
\draw  (6.5,1.)-- (8.,1.);
\draw  (4.,1.)-- (3.5,3.);
\draw  (8.,1.)-- (8.5,3.);
\draw  (3.5,3.)-- (5.5,1.);
\draw  (6.5,1.)-- (8.5,3.);
\draw  (4.,1.)-- (3.5,1.5);
\draw  (4.,3.)-- (4.,3.5);
\draw  (8.,1.)-- (8.5,1.5);
\draw  (8.,3.)-- (8.,3.5);
\draw  (3.5,1.5)-- (3.,2.);
\draw  (4.,1.)-- (4.5,1.5);
\draw [line width=1.pt,color=red] (8.5,1.5)-- (9.,2.);
\draw  (4.,3.)-- (4.5,2.5);
\draw  (5.,3.)-- (5.,2.5);
\draw [line width=1.pt,color=red] (5.,2.5)-- (5.5,2.5);
\draw  (4.,3.5)-- (4.5,3.5);
\draw [line width=1.pt,color=red] (8.,3.5)-- (7.5,3.5);
\draw  (6.5,3.)-- (6.5,2.5);
\draw  (7.,3.)-- (7.,2.5);
\draw[line width=1.pt,color=red]  (6.5,2.5)-- (6.5,2.);

\draw  (3.5,3.)-- (8.5,3.);

\foreach \i in {0.5,0.75,...,4.5}
{
    \draw [fill=gray] (3.5+\i ,3.) circle (1.75pt);
        
}

\begin{scriptsize}
\draw [fill=black] (4.,1.) circle (2.5pt);
\draw [fill=black] (6.5,1.) circle (2.5pt);
\draw [fill=black] (5.5,1.) circle (2.5pt);
\draw [fill=black] (8.,1.) circle (2.5pt);
\draw [fill=black] (3.5,3.) circle (2.5pt);
\draw [fill=black] (8.5,3.) circle (2.5pt);

\draw [fill=ududff] (3.5,1.5) circle (2pt);
\draw [fill=ududff] (4.,3.5) circle (2pt);
\draw [fill=ududff] (8.5,1.5) circle (2pt);
\draw [fill=ududff] (8.,3.5) circle (2pt);
\draw [fill=ududff] (3.,2.) circle (2pt);
\draw [fill=ududff] (4.5,1.5) circle (2pt);
\draw [fill=ududff] (9.,2.) circle (2pt);
\draw [fill=ududff] (4.5,2.5) circle (2pt);
\draw [fill=ududff] (5.,2.5) circle (2pt);
\draw [fill=ududff] (5.5,2.5) circle (2pt);
\draw [fill=ududff] (4.5,3.5) circle (2pt);
\draw [fill=ududff] (7.5,3.5) circle (2pt);
\draw [fill=ududff] (6.5,2.5) circle (2pt);
\draw [fill=ududff] (7.,2.5) circle (2pt);
\draw [fill=ududff] (6.5,2.) circle (2pt);
\end{scriptsize}
\end{tikzpicture}
    }
    \label{fig:subfig2}
  \end{subfigure}
  \hfill
  \begin{subfigure}{0.3\linewidth}
    \centering\rotatebox{90}{
    \begin{tikzpicture}[line cap=round,line join=round,>=triangle 45,x=1.0cm,y=1.0cm,scale=0.9]
\clip(2.4,0.5) rectangle (10.,4.);

\draw (4.,1.)-- (5.5,1.);
\draw  (6.5,1.)-- (8.,1.);
\draw  (4.,1.)-- (3.5,3.);
\draw  (8.,1.)-- (8.5,3.);
\draw  (3.5,3.)-- (5.5,1.);
\draw  (6.5,1.)-- (8.5,3.);
\draw  (4.,1.)-- (3.5,1.5);
\draw  (4.,3.)-- (4.,3.5);
\draw  (8.,1.)-- (8.5,1.5);
\draw  (8.,3.)-- (8.,3.5);
\draw  (3.5,1.5)-- (3.,2.);
\draw  (4.,1.)-- (4.5,1.5);
\draw [line width=1.pt,color=red] (8.5,1.5)-- (9.,2.);
\draw  (4.,3.5)-- (4.5,3.5);
\draw [line width=1.pt,color=red] (8.,3.5)-- (7.5,3.5);
\draw  (4,3.)-- (4.5,2.5);

\draw (3.5,3)-- (4.25,3);
\draw (8.5,3)-- (7.75,3);
\draw [line width=1.pt,color=red] (4.25,3)-- (7.75,3);

\foreach \i in {1.25,1.5,...,3.75}
{
    \draw [fill=green] (3.5+\i ,3.) circle (1.75pt);   
}

\foreach \i in {0,0.25,...,0.5}
{
    \draw [fill=black] (4+\i,3.) circle (1.75pt);
      \draw [fill=black](8-\i,3.) circle (1.75pt); 
}

\begin{scriptsize}
\draw [fill=black] (4.,1.) circle (2.5pt);
\draw [fill=black] (6.5,1.) circle (2.5pt);
\draw [fill=black] (5.5,1.) circle (2.5pt);
\draw [fill=black] (8.,1.) circle (2.5pt);
\draw [fill=black] (3.5,3.) circle (2.5pt);
\draw [fill=black] (8.5,3.) circle (2.5pt);

\draw [fill=black] (4.5,3.5) circle (2pt);
\draw [fill=black] (7.5,3.5) circle (2pt);
\draw [fill=black] (3.5,1.5) circle (2pt);
\draw [fill=black] (4.,3.5) circle (2pt);
\draw [fill=black] (8.5,1.5) circle (2pt);
\draw [fill=black] (8.,3.5) circle (2pt);
\draw [fill=black] (3.,2.) circle (2pt);
\draw [fill=black] (4.5,1.5) circle (2pt);
\draw [fill=black] (9.,2.) circle (2pt);
\draw [fill=black] (4.5,2.5) circle (2pt);

\end{scriptsize}
\end{tikzpicture}
    }
    \label{fig:subfig3}
  \end{subfigure}
\caption{\textbf{Left:} A graph $G$ with feedback edge number two. Feedback edges are orange, vertices in dangling trees are blue. \textbf{Middle:} The original $(H,\ca P)$-graph obtained from $G$ after all dangling trees have been cut down (i.e., the outcome of Theorem~\ref{thm:pruning}). A dangling pseudo-path is depicted via the grey vertices. \textbf{Right:} The tidy $(H,\ca P)$-graph that will later be obtained from the original $(H,\ca P)$-graph by applying Corollary~\ref{cor:no-stumps} at the end of Section~\ref{sec:preprocessing}. Here, $\ca P$ contains a single tidy dangling path which is colored green, and all other vertices lie in $H$.}
\label{fig:preprocessing}
\end{figure}
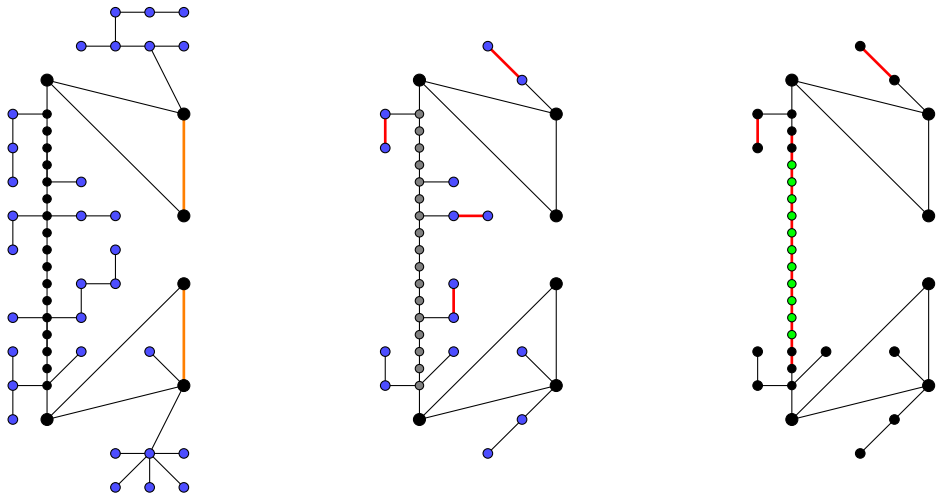

\iflong
\begin{theorem}
\fi
\ifshort
\begin{theorem}[$\clubsuit$]
\fi
\label{thm:pruning}
There is a polynomial-time procedure which takes as input a graph $G$ with feedback edge number $k$ and either outputs an optimal contraction sequence of $G$ of width at most $2$, or an original $(H,\ca P)$-graph $G'$ with effectively the same twin-width as $G$ such that $|V(H)|\leq 16k$ and $|\ca P|\leq 4k$.
\end{theorem}

\iflong \begin{proof}
As an initial step, we apply Theorem~\ref{thm:tww1-deciding} to check whether $\tww(G)\ge 2$; if not, we output an optimal contraction sequence for $G$ and terminate. We proceed by computing a feedback edge set $F$ in $G$ in polynomial time, and denoting by $Q$ the set of at most $2k$ many endpoints of the computed feedback edges. Next, we exhaustively check for dangling stars in the graph, and for each such star we apply Observation~\ref{obs:onestar} to replace it with a black stump; the graph $G^\alpha$ obtained in this way has effectively the same twin-width as $G$.

Next, we check whether $G^\alpha$ contains a dangling tree which is not a stump; if we detect such a tree, we expand it into a maximal dangling tree $T$ in $G^\alpha$. We would like replace $T$ with a red stump using Lemma~\ref{lem:onebigtree} to obtain a trigraph $G'$ but this can be done only if $\tww(G') \ge 2$. However, since $G'$ is a trigraph with a single red edge, we can apply Theorem~\ref{thm:tww1-deciding} again to decide whether $\tww(G') \ge 2$.
If $\tww(G') < 2$, we obtain a $2$-contraction sequence $C$ of $G^\alpha$ as follows.
Let $r \in V(T)$ be the endpoint of the bridge connecting $T$ to the rest of $G^\alpha$.
Similarly as in the proof of Lemma~\ref{lem:onebigtree}, $C$ starts with contracting $T$, following the contraction sequence $C_T$ given by Observation~\ref{obs:contract-trees} (with $r$ as the root), but stops before the first contraction involving $r$ (which is the very last contraction in $C_T$). Observe that the obtained trigraph is $G'$. Thus, $C$ can from now on follow the $1$-contraction sequence of $G'$ provided by Theorem~\ref{thm:tww1-deciding}. Since $w(C_T) \le 2$, we indeed have $w(C) \le 2$.

So, we can now safely assume that $\tww(G')\ge 2$ which in turn allows us to invoke Lemma~\ref{lem:onebigtree} and go from $G^\alpha$ to $G'$.
At this point, we exhaustively check for the presence of maximal dangling black trees, and for each tree we replace it with a red stump and proceed; the safeness of this step is guaranteed by Lemma~\ref{lem:twostumps} (which guarantees a lower-bound on the twin-width of the obtained trigraph since it contains at least two red stumps) in combination with Lemma~\ref{lem:onebigtree}. 

Next, we check if there is any vertex $u$ in the resulting trigraph $G^\beta$ which is adjacent to at least $2$ stumps.
If $u$ has a red stump and another stump $S$, we first use Theorem~\ref{thm:tww1-deciding} for our due-diligence check that $G' := G^\beta - S$ still has twin-width at least $2$.
If yes, we use Observation~\ref{obs:redstumps} to remove $S$. 
If not, we construct a $2$-contraction sequence of $G^\beta$ by starting with the $1$-contraction sequence of $G'$ (the descendant of $u$ may now have red degree 2 because of the extra edge to $S$) and then arbitrarily contracting the remaining three vertices.
If none of the stumps belonging to $u$ is red, we contract all half stumps into each other to reduce their number to at most $1$. Finally, if there are at least two black stumps $v_1w_1$ and $v_2w_2$, we use Theorem~\ref{thm:tww1-deciding} for the due-diligence check to ensure that the trigraph $G'$ with these two stumps replaced by a red stump $S$ still has twin-width at least $2$. If yes, we apply Lemma~\ref{lem:two-black-one-red} to replace the two stumps with $S$.
If not, we construct a $2$-sequence of $G$ by first contracting $v_1$ with $v_2$ and $w_1$ with $w_2$; the obtained trigraph is $G'$ and we finish by following a $1$-contraction sequence of $G'$.

After the above exhaustive procedures terminate, we arrive at a trigraph $G^\gamma$ whose vertices can be partitioned into a set of stumps and the remaining vertices, say $V'$. Observe that the trigraph $G^\gamma[V']$ induced on $V'$ satisfies that $G^\gamma[V']-F$ is a tree; in particular, none of the operations carried out above resulted in any contraction involving $V'$, and every leaf in $G^\gamma[V']-F$ must be an element of $Q$. Since $G^\gamma[V']-F$ contains at most $2k$ leaves, it must also contain at most $2k$ many vertices of degree greater than $2$; let $Q'$ denote the union of $Q$ and all of these vertices of degree greater than $2$ in $G^\gamma[V']-F$. Clearly, $|Q'|\leq 4k$. Finally, let $H$ be the union of $Q'$ and the vertices of all stumps that are attached to any vertex in $Q'$. Since a vertex could have at most three stump-vertices (one half stump and one black stump in the worst case) attached to it, we obtain $|H|\leq 16k$. Moreover, all of the at most $4k$ paths between vertices in $H$ are black dangling pseudo-paths and hence satisfy the condition imposed on $\ca P$, concluding the proof.
\end{proof} \fi

\section{Cleanup}
\label{sec:preprocessing}
In the second phase of our proof, our aim is to simplify the instance even further after Theorem~\ref{thm:pruning} in order to reach a trigraph which is ``clean'' enough to support the path reduction rules developed in the next sections. In particular, we will show that all the dangling pseudo-paths arising from Theorem~\ref{thm:pruning} can be safely transformed into red dangling paths, resulting in an $(H,\ca P)$-graph that is tidy as per the definition in the previous section.

We first show how to deal with stumps belonging to a single vertex.

\iflong
\begin{observation}
\fi
\ifshort
\begin{observation}[$\clubsuit$]
\fi
\label{obs:kill-stumps}
Let $G$ be a trigraph, let $u \in V(G)$ be a vertex with a single stump or a black and a half stump, and let $G'$ be the trigraph obtained from $G$ by deleting the stumps belonging to $u$ and making all edges incident to $u$ red. There is a partial contraction sequence from $G$ to $G'$ of width $\max\{d_1 + 1, d_2\}$, where $d_1$ is the red degree of $u$ in $G$ and $d_2$ is the maximum red degree of any vertex in $G'$.
\end{observation}
\iflong \begin{proof}
If $u$ has a half stump $v$ and a black stump $wx$, start with contracting $v$ and $w$, creating a red stump.
If $u$ has a black or red stump $vw$, contract $v$ and $w$ into $x$: this adds a new red neighbor to $u$ (that is why the width may be as high as $d_1 + 1$).
Now $u$ has an extra neighbor $x$ not present in $G'$ ($x$ is either the vertex obtained in the previous step or a half stump), and so we conclude by contracting $u$ and $x$, obtaining $G'$. 
\end{proof} \fi

Next, we establish that if a dangling pseudo-path has two consecutive vertices without stumps, the edge between them can be turned red. We remark that this lemma will be applied to subtrigraphs of the considered trigraph, and hence we cannot assume that the trigraph has twin-width at least $2$.

\iflong \begin{lemma} \fi \ifshort \begin{lemma}[$\clubsuit$]\fi\label{lem:change-edge-to-red}
Let $G$ be a trigraph,
let $(u_1, u_2, u_3, u_4)$ be an induced path in $G$ in this order, and assume that the degree of $u_2$ and $u_3$ is $2$ in $G$ and that the degree of $u \in \{u_1,u_4\}$ in $G$ is $3$ if $u$ has a single stump, $4$ if $u$ has a half stump and a black stump, and $2$ otherwise.
Let $G'$ be the trigraph obtained from $G$ by changing the color of the edge $u_2u_3$ to red. 
Given a contraction sequence $C$ of $G$, a contraction sequence $C'$ of $G'$ s.t.\ $w(C') = \max\{2, w(C)\}$ can be constructed in polynomial time.
\end{lemma}

\iflong \begin{proof}
Let $C = (G = G_0, G_1, \ldots, G_n)$ be a contraction sequence of $G$, and let $U := \{u_i \sep i\in [4]\}$.
We construct a contraction sequence $C' = (G' = G'_0, G'_1, \ldots, G'_n)$ of $G'$, and show that $w(C') = \max\{2, w(C)\}$.
Let $G_i$ be the first trigraph in $C$ which is obtained by a contraction involving $u_2$ or $u_3$.
Without loss of generality, we may assume that $G_i$ is obtained by a contraction involving $u_2$ and some $v \in V(G_{i-1})$.
Observe that only $u_2$ and $u_4$ may be black neighbors of $u_3$ in $G$. %
Thus, after the contraction of $u_2$ and $v$, the edge $u_2u_3$ either disappears (if $v = u_3$) or turns red unless $v = u_4$.

First, suppose $v \ne u_4$.
In this case, we construct $C'$ by following contractions in $C$ (recall $V(G) = V(G')$).
Observe that for all $j \in [0, i-1]$ all vertices of $G'_j$ have red degree at most $\max\{2, w(C)\}$
because $u_2$ and $u_3$ have red degree 1 or 2 in $G_j'$, and all other vertices in $G_j'$ have the same red degree as in $G_j$.
Moreover, $G_i = G_i'$ because the color of $u_2u_3$ is now the same, as argued above.
By induction, $G_j = G_j'$ for all $j \in [i, n]$, which shows $w(C') = \max\{2, w(C)\}$.

Second suppose $G_i$ is obtained by a contraction of $u_2$ and $u_4$.
For $u \in\{u_1,u_4\}$, let:
\begin{itemize}
\item $u'$ be the unique neighbor of $u$ in $U$ (in the original trigraph $G$ as well as in $G_{i-1}$), 
\item $u^*$ be the descendant of $u$ in $G_{i-1}$, and 
\item $\hat u$ be the descendant (in $G_{i-1}$) of the unique neighbor of $u$ which is neither in $U$ nor in a stump of $u$ (in $G$).
\end{itemize}
Note that $u'$ is also a neighbor of $u^*$ in $G_{i-1}$ since $u' \in \{u_2, u_3\}$ and no contraction up to $G_i$ involved $u'$ by choice of $i$.
If $u^*u' \in R(G_{i-1})$ or the red degree of $u^*$ is at most 1 in $G_{i-1}$, then we say that $u$ is \emph{safe}.
Suppose $u$ is not safe. Observe that $u = u^*$ because $u$ could not have been contracted with $u_2$ nor with $u_3$ by choice of $i$, and any other contraction involving $u$ would make the edge $u^*u'$ red in $G_{i-1}$, which would make $u$ safe.
If $\hat u$ is a black neighbor of $u$ in $G_{i-1}$, then we say that $u$ is \emph{half-safe}. Otherwise, $u$ is \emph{unsafe}.

Let $S_u$ be the set of stumps belonging to $u$ in $G$.
Observe that if $u$ is half-safe, then $|S_u| = 2$: indeed, $u$ has two black neighbors ($u'$ and $\hat u$) and two red neighbors (otherwise it would be safe) in $G_{i-1}$, and it can have degree 4 only if $|S_u| = 2$ (here we use $u = u^*$).

We define $C'$ as follows:
\begin{enumerate}
\item Perform the first $i-1$ contractions of $C$ but skip all contractions involving stumps of unsafe vertices and half-stumps of half-safe vertices. Let $G_{i_1}'$ be the obtained trigraph.
\item For all (at most 2) unsafe vertices $u$, contract $u$ with its stumps as per Observation~\ref{obs:kill-stumps}.
\item For all (at most 2) half-safe vertices $u$, if $S_u = \{v, wx\}$, then contract $v$ and the descendant of $w$.
\item Let $G_{i_2}'$ be the current trigraph. Contract $u_2$ and $u_3$ into a new vertex $u_{23}$. Let $G_{i_3}'$ be the obtained trigraph.
\item Finally, contract the new vertex $u_{23}$ and $u_4$. Let $G_{i_4}'$ be the resulting trigraph.
Finish as a restriction of $C$ because $G_{i_4}'$ is a pseudoinduced subtrigraph of $G_i$ (as we argue below).
\end{enumerate}

Naturally, an observant reader might by now be wondering what would go wrong if one were to ignore $u_1$ and $u_4$ and simply contract $u_2$ and $u_3$ immediately after the first $i-1$ contractions (circumventing the whole discussion about safe, half-safe and unsafe vertices). At first glance, it seems that the presence of a new red neighbor for the descendants of $u_1$ and $u_4$ would be offset by the fact that in $G_i$, the descendant of $u_1$ is also new red neighbor of the descendant of $u_4$ (and vice-versa for $u_4$). However, that is not necessarily the case:  the descendants of $u_1$ and $u_4$ may in fact already have been red neighbors in $G_{i-1}$---necessitating the case analysis performed above.

Now we show that no trigraph $G_j'$ in $C'$ contains a vertex with high red degree.
Let $j \in [i_1]$ and let $G_j$ be the trigraph corresponding to $G_j'$ in $C$ (this is a slight abuse of notation, since $C$ and $C'$ are not ``synchronized'' because of the skipped contractions). In particular $G_{i-1}$ corresponds to $G_{i_1}'$.
By definition of $i$, $u_2$ and $u_3$ have red degree at most 2 in $G_j'$ (no contraction involved them).
Let $u$ be an unsafe or half-safe vertex.
As observed above, $u$ has not been contracted with any vertex in $G_{i-1}$, and so the same thing holds also for $G_j'$.
If $u$ is unsafe, then its only possible red neighbor is $\hat u$, and if $u$ is half-safe, then $S_u = \{v, wx\}$ and the only possible red neighbor of $u$ is the descendant of $w$.
Moreover, the red degree of any vertex contained in a stump of $u$ is at most 1 in $G_j'$.
Any other vertex of $G_j'$ has at most as many red neighbors in $G_j'$ as in $G_j$ (some vertices may have smaller red degree in $G_j'$ thanks to the skipping of some contractions). %

Let $u$ be an unsafe vertex. We need to show that contracting the stumps of $u$ with $u$ is safe. Crucially, the edge $u\hat{u}$ is red in $G_{i-1}$, and so it can cause no problem if contracting the stumps with $u$ makes it red (it can be black in $G_{i_1}'$ only thanks to the skipping of some contractions). Thus, after the stumps of $u$ are gone, the descendant of $u$ has two red neighbors, namely $u'$ and $\hat u$, and $u' = u_k$ has also two red neighbors, namely the descendant of $u$ and $u_{5-k}$.

Let $u$ be a half-safe vertex such that $S_u = \{v, wx\}$. Observe that the descendant $w^*$ of $w$ is a red neighbor of $u$ in $G_{i_1}'$ (otherwise $u$ would be safe since $u\hat{u}$ is a black edge in $G_{i-1}$). Since no contraction so far involved $u$, a contraction must have involved $w$, which means that $w^*$ has black degree 0 (either $x$ is in the bag of $w^*$ or $w^*x$ is a red edge). Thus the contraction of $v$ and $w^*$ does not create any red edges. We have now argued safeness for $j \in [i_2]$.

Let us now consider the trigraph $G'_{i_3} = G'_{i_2+1}$, obtained by contracting $u_2$ and $u_3$ into a new vertex $u_{23}$, which has red degree 2 in $G'_{i_3}$. Clearly, the red degree of vertices not adjacent to $u_{23}$ is the same in $G'_{i_3}$ as in $G'_{i_2}$. Let $u$ be a neighbor of $u_{23}$. The red degree of $u$ may increase by the contraction of $u_2$ and $u_3$ only if $uu'$ is a black edge in $G'_{i_2}$ (recall $u' \in\{u_2, u_3\}$), which means that $u$ is safe or half-safe (for unsafe vertices, the edge has turned red).
If $u$ is safe, then, by definition, the red degree of $u$ is at most 1 in $G_{i-1}$, and thus also in $G'_{i_2}$.
If $u$ is half-safe, then its red degree in $G'_{i_2}$ is 1 (its unique red neighbor is the new vertex obtained from $w^*$ and $v$).
In all cases, the red degree of $u$ is at most $2$ in $G'_{i_3}$.

Finally, $u_{23}$ and $u_4$ are contracted. Now we show that the resulting trigraph, $G_{i_4}'$, is a pseudoinduced subtrigraph of $G_i$. Indeed, $G_{i_4}'$ can be obtained from $G_i$ by first removing $u_3$,
second removing the vertices of stumps of unsafe and half-safe vertices which are not present in $G_{i_4}'$ because of contractions in steps 2 and 3,
and third replacing each red edge that is present in $G_i$ because of contractions skipped by $C'$ (in step 1) by a black edge or a non-edge.
Thus, we may safely continue $C'$ as a restriction of $C$, which concludes the proof.
\end{proof} \fi

With Lemma~\ref{lem:change-edge-to-red}, we show that a dangling pseudo-path can either be safely transformed into a ``real'' dangling path, or (if the part is too short) absorbed into $H$\ifshort, which in turn allows us to prove:\fi\iflong:

\iflong \begin{lemma} \fi \ifshort \begin{lemma}[$\clubsuit$]\fi\label{lem:one-path-preprocess}
Let $G$ be an $(H, \ca P)$-graph such that each path in $\ca P$ is original or tidy, and let $P\in \ca P$ be an original path. There is a polynomial time procedure which constructs an $(H', \ca P')$-graph $G'$ such that:
\begin{itemize}
\item $G'$ has effectively the same twin-width as $G$;
\item all paths in $\ca P'$ are either original or tidy;
\item $|V(H')| \le |V(H)| + 24$;
\item $|\ca P'|\leq |\ca P|$ and $G'$ contains one fewer original paths than $G$.
\end{itemize}
\end{lemma}
\begin{proof}
Let $(u_1, \ldots, u_n)$ be the sequence of vertices obtained by traversing $P$.
If $n \le 6$, then let $G' := G$, $\ca P' := \ca P \setminus \{P\}$, and let $H'$ be the induced subtrigraph of $G$ containing $H$, $P$, and all stumps belonging to $P$.
Each vertex of $P$ has either a single stump or a black and a half stump (by Definition~\ref{def:more-paths}), which means up to three stump-vertices per vertex of $P$, and so indeed $|V(H')| \le |V(H)| + 24$.
From now on, suppose $n > 6$.

Let $P' = (u_4, \ldots, u_{n-3})$ be a red path, let $U = \{u_1, u_2, u_3, u_{n-2}, u_{n-1}, u_n\}$, %
let $H'$ be the subtrigraph of $G$ induced by $V(H) \cup U \cup \{s \sep s$ is a part of a stump belonging to $u_1 \text{ or }u_n\}$,
and let $\ca P' := \ca P \setminus \{P\} \cup \{P'\}$.
Let $G^-$ be the trigraph obtained from $G$ by deleting all stumps belonging to each $u_i$ for $i \in [2, n-1]$,
and let $G'$ be the trigraph obtained from $G^-$ by changing the color of all edges $u_iu_{i+1}$ for $i \in [2, n-2]$ to red.
Observe that:
\begin{itemize}
\item $G'$ is an $(H', \ca P')$-graph;
\item $|V(H')| \le |V(H)| + 12$ because $|U| = 6$ and stumps belonging to $u_1$ and $u_n$ consist of at most six vertices in total;
\item $P'$ is tidy because each $u \in \{u_3, u_{n-2}\}$, i.e., each connector adjacent to an endpoint of $P'$, has black degree 0, is adjacent to a single vertex of $\sump\ca P'$ (namely $u_4$ or $u_{n-3}$), and has a unique neighbor in $H'$ (namely $u_2$ or $u_{n-1}$), which has positive black degree;
\item all paths in $G'$ are tidy or original (since no path $Q \in \ca P \cap \ca P'$ has been affected by the change).
\end{itemize}
Thus it only remains to be shown that $G'$ has effectively the same twin-width as $G$.
First we show $\tww(G') \le_e \tww(G)$.
Let $C$ be a contraction sequence of $G$, and let $C^-$ be its restriction to $G^-$.
By Observation~\ref{obs:induced}, $w(C^-) \le w(C)$ and $C^-$ can be constructed in polytime, given $C$.
Notice that the degree of all vertices of $P$ is 2 in $G^-$, except for $u \in \{u_1, u_n\}$, which may have degree 3 if it has a single stump or 4 if it has a black and a half stump.
This allows us to repeatedly apply Lemma~\ref{lem:change-edge-to-red}
to change the color of all edges $u_iu_{i+1}$, for $i\in[2, n-2]$, to red, without increasing the twin-width, i.e.,
we obtain a contraction sequence $C'$ of $G'$ s.t. $w(C') = \max\{2, w(C^-)\}$, and $C'$ can be constructed in polynomial time given $C^-$. 
If $w(C^-) \ge 2$, we have $w(C') = w(C^-) \le w(C)$, and if $w(C^-) < 2$, then $w(C') = 2 \le \tww(G) \le w(C)$ by Definition~\ref{def:more-paths}.
Hence we have $w(C') \le w(C)$.

Second we show $\tww(G) \le_e \tww(G')$.
Let $C'$ be a contraction sequence of $G'$ (recall $w(C') \ge 2$ since $G'$ contains a red path).
We construct a contraction sequence $C$ of $G$ of width at most $w(C')$.
Let $i \in [3, n-2]$ be arbitrary such that $u_i$ has a stump.
We use Observation~\ref{obs:kill-stumps} to safely contract the stumps of $u_i$ with $u_i$: now the descendant of $u_i$ has red degree 2, and its two neighbors have red degree 1 each.
Now $C$ repeats this process for all $j \in [3, n-2]$ s.t. $u_j$ has a stump in $G$, from the $u_j$ closest to $u_i$ to the farthest (i.e., if $|i - j| > |i - k|$, then the stumps of $u_k$ must be contracted before the stumps of $u_j$).
Throughout this process, no vertex has red degree higher than 2 because each $u_j$ has at most one red neighbor (namely, its neighbor in $P$ closer to $u_i$) before $C$ starts contracting its stumps (recall the $d_1 + 1$ bound in the statement of Observation~\ref{obs:kill-stumps}). Let $G^*$ be the obtained trigraph.

Now the only vertices of $P$ which may have stumps in $G^*$ but not in $G'$ are $u_2$ and $u_{n-1}$.
Suppose $u \in \{u_2, u_{n-1}\}$ has a stump.
Now we again start contracting the stumps of $u$ as in the proof of Observation~\ref{obs:kill-stumps} but we stop before the contraction which would contract the last remaining stump-vertex $x$ and $u$ and instead we contract $x$ with the neighbor of $u$ in $\{u_3, u_{n-2}\}$.
Now the obtained trigraph is a pseudoinduced subtrigraph of $G'$ (the only difference between them is that some edges $u_iu_{i+1}$ for $i \in [2, n-2]$, red in $G'$, may be black instead), and so $C$ concludes by following $C'$. We have shown that $w(C) \le w(C')$. %
  \end{proof}

By exhaustively applying Lemma~\ref{lem:one-path-preprocess}, we obtain:
\fi

\begin{corollary}\label{cor:no-stumps}
There is a polynomial-time algorihm which transforms an original $(H, \ca P)$-graph into a tidy $(H', \ca P')$-graph with effectively the same twin-width 
such that $|V(H')| \le |V(H)| + 24\cdot |\ca P|$ and $|\ca P'|\leq |\ca P|$.
\end{corollary}

Before we proceed towards establishing our main algorithmic theorems, we remark that Theorem~\ref{thm:pruning} and Corollary~\ref{cor:no-stumps} allow us to bound the twin-width of graphs with feedback edge number $1$, generalizing the earlier result of Bonnet et al.~\cite[Section~3]{BonnetKTW22} for trees.

\iflong
\begin{theorem}
\fi
\ifshort
\begin{theorem}[$\clubsuit$]
\fi
\label{thm:fenone}
Every graph with feedback edge number $1$ has twin-width at most $2$.
\end{theorem}

\iflong \begin{proof}
We apply Theorem~\ref{thm:pruning} and Corollary~\ref{cor:no-stumps} to obtain either a $2$-contraction sequence of $G$ or a tidy $(H,\ca P)$-graph $G'$ with effectively the same twin-width as $G$. Moreover, by following the construction we observe that $G'$ will consist of a single cycle $C=\{u_1,\dots,u_m\}$, each of whose vertices may be attached to a set of stumps. We complete the proof by providing a $2$-contraction sequence for $G'$:

\begin{enumerate}
\item for each vertex $u_i\in V(C)$, we contract all of its stumps into a single vertex $u_i'$ which has red degree $1$ and whose only neighbor is $u_i$;
\item we contract $u_1'$ with $u_2'$, and then contract the resulting vertex into $u_2$. At this point, $u_2$ is the only vertex in $C$ with red degree $2$;
\item repeat the previous step for each $u_{2j-1}'$ and $u_{2j}'$ for $2\leq j\leq \lfloor\frac{m}{2}\rfloor$. If $m$ is odd, observe that $u_m$ has red degree $2$ at this stage, and finish by contracting $u_m'$ into $u_m$;
\item at this point, the trigraph simply consists of a red cycle which is well-known (and easily observed) to have twin-width $2$. \qedhere
\end{enumerate}
\end{proof} \fi

As an immediate corollary, we can obtain an even more general statement:
\iflong
\begin{corollary}
\fi
\ifshort
\begin{corollary}[$\clubsuit$]
\fi
\label{cor:fenmore}
Every graph with feedback edge number $\ell\ge 1$ has twin-width at most $1+\ell$.
\end{corollary}

\iflong
\begin{proof}
The statement holds by directly combining Theorem~\ref{thm:fenone} with the fact that adding $\ell-1$ edges to a graph $G$ can only increase the twin-width of $G$ by at most $\ell-1$. Indeed, an arbitrary contraction sequence of $G$ of width $\tww(G)$ will have width at most $\tww(G)+\ell-1$ after adding an arbitrary set of $\ell-1$ edges.
\end{proof}
\fi

\section{Establishing Theorem~\ref{thm:tww2}}\label{sec:tww2}

Our aim now is to make the step from Corollary~\ref{cor:no-stumps} towards a proof of Theorem~\ref{thm:tww2}. Towards this, let us fix a tidy $n$-vertex $(H,\ca P)$-graph $G$. %
When dealing with a contraction sequence, we will use $G_i$ to denote the $i$-th trigraph obtained from $G$,
and let $H_i$ be the subtrigraph of $G_i$ induced by the descendants of $H$.
We say that $u \in V(G_i)$ is \emph{an outer vertex in $G_i$} if $u \notin V(H_i)$, and we lift the previous definition of connectors by saying that $u$ is a \emph{connector in $G_i$} if $u \in V(H_i)$ and $u$ is adjacent to an outer vertex in $G_i$.

We begin with a simple observation which will be useful throughout the rest of the section.

\iflong
\begin{observation}
\fi
\ifshort
\begin{observation}[$\clubsuit$]
\fi
\label{obs:bd0}
If $G_i$ is a trigraph obtained from a tidy $(H,\ca P)$-graph $G$ by a sequence of contractions, then all the neighbors of outer vertices in $G_i$ have black degree 0 in $G_i$.
\end{observation}
\iflong \begin{proof}
The fact that outer vertices have black degree 0 is immediate ($\ca P$ is a collection of red paths).
Let $u_1 \in V(H_i)$ be a connector adjacent to an outer vertex $u_2 \in V(G_i)$.
Clearly, there must exist vertices $u_1', u_2' \in V(G)$ which are adjacent in $G$ and $u_1'$ ($u_2'$) is contained in the bag of $u_1$ ($u_2$, respectively).
Since $u_2$ is outer, we have $u_2' \in V(\sump\ca P)$.
By the definition of tidy $(H,\ca P)$-graphs, all neighbors of $u_2'$ have black degree 0 in $G$.
In particular, $u_1'$ has black degree 0 in $G$, and so clearly $u_1$ has black degree 0 in $G_i$.
\end{proof} \fi

Our proof of Theorem~\ref{thm:tww2} relies on establishing that if a tidy $(H,\ca P)$-graph $G$ has twin-width $2$, then it also admits a contraction sequence which is, in a sense, ``well-behaved''. The proof of this fact is based on induction, and hence being ``well-behaved'' (formalized under the notion of \emph{regularity} below) is defined not only for entire sequences but also for prefixes.

\begin{definition}\label{def:regular-prefix}
Let $C = (G_1 = G, G_2,\ldots, G_n)$ be a contraction sequence.
For $P \in \ca P$, let us denote by $P_i$ the subtrigraph of $G_i$ induced by the descendants of $P$ which are not in $H_i$.
We say that a prefix $(G_1, \ldots, G_i)$ of $C$ is \emph{regular} if: 
\begin{compactitem}
\item for all $j \in [i]$, $\{P_j \sep P \in \ca P\}$ is a set of disjoint red paths, and the endpoints of these paths are adjacent to connectors; and
\item for all $j \in [i-1]$:
\begin{compactenum}
\item $G_{j+1}$ is obtained by a contraction inside $H_j$, or \label{loc:inH}
\item there is $P \in \ca P$ s.t. if you shorten $P_j$ by one vertex in $G_j$, you obtain $G_{j+1}$, or \label{loc:shorten}
\item there is $P \in \ca P$ s.t. $|V(P_j)| = 1$ and $|V(P_{j+1})| = 0$.\label{loc:delete}
\end{compactenum}
\end{compactitem}
\end{definition}

Let $reg(C)$ denote the length of the longest regular prefix of $C$.
Below, we show that unless we reach a degenerate trigraph (a case which is handled in the proof of Proposition~\ref{prop:tww2} later), every contraction sequence of width $2$ can be made ``more regular'' until it is entirely regular.

\iflong \begin{lemma} \fi \ifshort \begin{lemma}[$\clubsuit$]\fi\label{lem:regularize}
Let $C = (G_1 = G, G_2,\ldots, G_n)$ be an optimal contraction sequence of $G$ of width $2$, and let $i := reg(C)$.
If $i < n$ and $G_i$ is not a red cycle of length 4, then there is an optimal contraction sequence $C'$ of $G$ s.t. $reg(C') > i$.
\end{lemma}
\iflong \begin{proof}
Suppose that $i < n$ and that $G_i$ is not a red cycle of length 4.
We will divide the proof into several cases but in all of them,
we will construct $C'$ by following $C$ until $G_i$, then $C'$ performs a contraction different from $C$, creating a trigraph $K$, and then we observe that $K$ is an induced subtrigraph of $G_i$, which allows us to finish $C'$ as a restriction of $C$.
Crucially, the prefix $(G_1, \ldots, G_i, K)$ of $C'$ will be regular, and so we will indeed have $reg(C') > reg(C)$.

Suppose that there is an outer vertex $u \in V(G_i)$ with a unique neighbor $u' \in V(G_i)$.
By Definition~\ref{def:regular-prefix},
$\{u\} = V(P_i)$ for some $P \in \ca P$ since otherwise $u$ would be an endpoint of $P_i$ not adjacent to a connector.
In this case, we let $K$ be the trigraph obtained from $G_i$ by contracting $u$ and $u'$ (the prefix of $C'$ ending with $K$ is now regular by item~\ref{loc:delete} of Definition~\ref{def:regular-prefix}).
This contraction does not increase the red degree of any vertex because $u'$, which is a connector, has black degree 0 by Observation~\ref{obs:bd0}.

Otherwise,
let $u,v \in V(G_{i})$ be the two vertices contracted in $G_{i+1}$, into a new vertex $uv$.
If $u,v \in V(H_i)$, then $G_{i+1}$ would satisfy item~\ref{loc:inH}, contradiction with $reg(C) = i$.
Thus suppose $u \in V(P_i)$ for some $P \in \ca P$.
First suppose $v \in V(Q_i)$ for some $Q \in \ca P$.
We may assume that both $u$ and $v$ have (red) degree 2 in $G_i$ (see the previous case).
However, $uv$ has red degree at most 2 in $G_{i+1}$ by optimality of $C$, and so either $u$ and $v$ are neighbors or they have the same neighborhood in $G_i$.
If they were neighbors, then $Q$ would be equal to $P$ and $(G_1, \ldots, G_{i+1})$ would be a regular prefix by item~\ref{loc:shorten}, contradiction.
On the other hand, suppose $u$ and $v$ have the same neighborhood in $G_i$, say $\{w, x\}$.
By Observation~\ref{obs:bd0}, $w$ and $x$ have black degree 0, and so $G_i$ is just a red cycle $uwvx$ (since it is connected), which is a contradiction with our initial assumption.

Second suppose $v \in V(H_i)$.
If $uv \in R(G_i)$, then $(G_1, \ldots, G_{i+1})$ would be a regular prefix by item~\ref{loc:shorten} (or~\ref{loc:delete} if $V(P_i) = \{u\}$), contradiction with $i = reg(C)$.
Thus let $N := \{w, x\}$ be the neighborhood of $u$; we may assume $|\{v, w, x\}| = 3$.
Observe that all neighbors of $v$ must be in $N$ since $uv$ has red degree at most 2 in $G_{i+1}$.
If $v$ has degree 1 in $G_i$, then its unique neighbor $v' \in N$ has black degree 0 by Observation~\ref{obs:bd0}, %
and we let $K$ be the trigraph obtained from $G_i$ by contracting $v$ and $v'$ (the prefix $(G_1, \ldots, G_i, K)$ is regular by item~\ref{loc:inH} or~\ref{loc:shorten}, depending on whether $v' \in V(H_i)$ or $v' \in V(P_i)$).
Thus suppose $v$ has degree 2 in $G_i$, i.e., its neighbors are $w$ and $x$.
Same way as before, $G_i$ is just a red cycle $uwvx$, again a contradiction.
  \end{proof} \fi

We can now use Lemma~\ref{lem:regularize} to show that contracting all the tidy dangling paths in $G$ into singletons cannot increase the twin-width of $G$.

\iflong
\begin{proposition}
\fi
\ifshort
\begin{proposition}[$\clubsuit$]
\fi
\label{prop:tww2}
Let $G'$ be the trigraph obtained from a tidy $(H,\ca P)$-graph $G$ of twin-width $2$ by shortening each path in $\ca P$ to a single vertex. Then $\tww(G') = 2$.
\end{proposition}

\iflong \begin{proof}
Let $C_0$ be an optimal contraction sequence of $G$.
We repeatedly apply Lemma~\ref{lem:regularize} to $C_0$ until we obtain an optimal contraction sequence $C_1 = (G_1 = G,\ldots, G_n)$ of $G$ such that either $G_{n-3}$ is a red cycle or $reg(C_1) =n$.
Let $C_2$ be the longest regular prefix of $C_1$, i.e., either $C_2 = C_1$ or $C_2 =(G_1,\ldots, G_{n-3})$.
Thus, each contraction in $C_2$ is one of the three types listed in Definition~\ref{def:regular-prefix}.
We define a (possibly partial) contraction sequence $C'_2$ of $G'$ by removing all contractions of type~\ref{loc:shorten} from $C_2$, i.e., those contractions shortening a path in $\ca P$.

Let $G_i'$ be any trigraph in $C'_2$ and let $G_i$ be the trigraph in $C_2$ which was obtained by the same contraction which produced $G_i'$
(to avoid any confusion, we are not claiming that $G_i'$ is the trigraph on position $i$ in $C'_2$).
First observe that any outer vertex in $G_i'$ has red degree at most 2 because it has not been contracted with any other vertex.
Second observe that the subtrigraph of $G_i'$ induced by the descendants of $H$ is isomorphic to $H_i$.
Moreover, if $u \in V(H_i)$ is adjacent to $m$ outer vertices
in $G_i$, then $u$ is adjacent to $m' \le m$ outer vertices in $G_i'$ (if $m = 2$ and both outer vertices adjacent to $u$ in $G_i$ are in $P_i$ for some $P \in \ca P$, then $m' = 1$; otherwise, $m' = m$).
Thus, no vertex in $G_i'$ has red degree higher than 2.

Now we are ready to define an optimal contraction sequence $C'$ of $G'$ of width 2. If $C_2 = C_1$, then $C'_2$ ends with a single-vertex graph and we simply set $C' := C'_2$. Otherwise $C'_2$ ends with a trigraph $G_i'$ that is isomorphic to a red cycle of length $\ell \le 4$ ($\ell < 4$ only if $|V(P_{n-3})| \ge 2$ for some $P \in \ca P$).
Observe that all contraction sequences of $G_i'$ have width at most 2, and so $C'$ may be any non-partial contraction sequence extending $C'_2$.
  \end{proof} \fi
  
With Proposition~\ref{prop:tww2} in hand, we can complete the proof of Theorem~\ref{thm:tww2}.

\thmone*

\begin{proof}
First we use Theorem~\ref{thm:pruning}: if it returns an optimal contraction sequence of the input graph $G_0$, we immediately know its twin-width. Otherwise, we obtain an original $(H, \ca P)$-graph $G$ with effectively the same twin-width as $G_0$ such that $|V(H)| \le 16k$ and $|\ca P| \le 4k$. Now we use Corollary~\ref{cor:no-stumps} to transform $G$ into a tidy $(H', \ca P')$-graph $G'$ that has effectively the same twin-width as $G$ and that satisfies $|V(H')| \le 112k$ and $|\ca P'| \le 4k$. By transitivity of $=_e$, we obtain $\tww(G_0) =_e \tww(G')$. Finally, let $G''$ be the trigraph obtained from $G'$ by shortening each path in $\ca P'$ to a single vertex.

By Proposition~\ref{prop:tww2}, $\tww(G') = 2$ implies $\tww(G'') = 2$. Conversely, given a contraction sequence $C''$ of $G''$, we can construct a contraction sequence $C'$ of $G'$ of width at most $w(C'')$ by first shortening each path in $\ca P'$ to a single vertex via progressive contractions of consecutive vertices, and then following $C''$; thus, $\tww(G') \le_e \tww(G'')$. Since $\tww(G') \ge 2$, we obtain that $\tww(G'') = 2$ implies $\tww(G') = 2$.
By combining these two implications with $\tww(G_0) = \tww(G')$, we obtain that $\tww(G_0) = 2$ if and only if $\tww(G'') = 2$. Since all operations required to construct $G''$ from $G_0$ can be performed in polynomial time and $|V(G'')| \le 116k$, $G''$ is indeed a linear bikernel for the considered problem. 

Finally, if $\tww(G_0) \le 2$, an optimal contraction sequence of $G_0$ can be computed in the desired time: either it is given by Theorem~\ref{thm:pruning}, or we construct $G''$ in polynomial time and compute an optimal contraction sequence of $G''$ in time $2^{\bigoh(k\cdot\log k)}$ as per Observation~\ref{obs:brute-force}, and then the result follows by the effectiveness in $\tww(G_0) =_e \tww(G') \le_e \tww(G'')$.
\end{proof}

\section{Establishing Theorem~\ref{thm:tww-3+}}\label{sec:tww-3+}

We now move on to the most involved part of the paper: the final step towards proving Theorem~\ref{thm:tww-3+}, which we will outline in the next few paragraphs. Recall that after applying Corollary~\ref{cor:no-stumps}, we obtain a tidy $(H,\ca P)$-graph $G$ with effectively the same twin-width as the input graph, and we ``only'' need to show that the dangling paths in $\ca P$ can be shortened to length bounded by the input parameter without increasing the twin-width too much. As we noted earlier, there is no ``local'' way of shortening a dangling path (see Proposition~\ref{prop:dangling}).

Instead, our approach is based on establishing the existence of a $(\tww(G)+1)$-contraction sequence $C^*$ for the trigraph $G^*$ obtained from $G$ by shortening its long paths; $C^*$ is obtained by  non-trivially repurposing a hypothetical optimal contraction sequence $C = (G_1 = G, G_2,\ldots, G_n)$ of $G$. Once we do that, we will have proven that our algorithm can produce a near-optimal contraction sequence for $G$ by first shortening the paths (by contracting neighbors) and then, when the paths are as short as in $G^*$, by applying Observation~\ref{obs:brute-force}. In other words, the complex machinery devised in this section is required for the correctness proof, while the algorithm itself is fairly simple.

Intuitively, the reason why it is difficult to go from $C$ to $C^*$ is that $C$ has too much ``freedom'': it can perform arbitrary contractions between vertices of the dangling paths, whereas the only limitation is that the red degrees cannot grow too high (this was not an issue in Section~\ref{sec:tww2}, since having twin-width 2 places strong restrictions on how the dangling paths may interact). 
We circumvent this issue by not following $C$ too closely when constructing $C^*$. Instead, we only look at a bounded number of special trigraphs in $C$, called \emph{checkpoints}---forming the \emph{big-step} contraction sequence defined later---and show that we can completely ignore what happened in $C$ between two checkpoints when constructing $C^*$.
Moreover, while checkpoints may be large and complicated, we identify for each checkpoint a small set of characteristics which will be sufficient to carry out our construction; we call this set the \emph{blueprint} of the checkpoint, and it also includes the induced subtrigraph $H_i$ of all descendants of $H$ in a checkpoint $G_i$ (the so-called \emph{core}). 

Our aim is to simulate the transition from one checkpoint to another via a ``controlled'' contraction sequence. For this purpose, we define \emph{representatives}---trigraphs in $C^*$ which match the blueprints of checkpoints in $C$%
---and show how to construct a partial contraction sequence from one representative to another in Subsection~\ref{sub:repr}. In the end, these partial contraction sequences will be concatenated to create $C^*$. 

A crucial gadget needed to define these representatives are \emph{centipedes}; these are well-defined and precisely structured objects which simulate the (possibly highly opaque) connections between red components in the core, and are illustrated in Figure~\ref{fig:sec6} later.
Subsection~\ref{sub:centipedes} is dedicated to establishing the operations required to alter the structure and placement of centipedes between individual representatives.
These operations rely on the fact that a vertex is allowed to have red degree 4; this is one reason why Theorem~\ref{thm:tww-3+} produces sequences whose width may be one larger than the optimum (the other reason is that the possibility of supporting an additional red edge provides more flexibility when moving and altering the centipedes between checkpoints).

One final issue we need to deal with is that the paths in $\ca P$ must be sufficiently long in order to support the creation of the centipedes at the beginning of $C^*$. Fortunately, there is a simple way of resolving this: paths in $\ca P$ which are not long enough can be moved into $H$. However, the cost of this is that each time we add such a path into $H$, the size of $H$---and hence also the bound on the length of the paths in $G^*$---can increase by an exponential factor. For this reason, unlike in the previous section, the bikernel we obtain is not polynomial and not even elementary; its size will be bounded by a tower of exponents whose height is linear in the parameter.

\subsection{Initial Setup}\label{sub:init}
Recall that at this point, we are dealing with a tidy $(H,\ca P)$-graph $G$. Let $C = (G_1 = G, G_2,\ldots, G_n)$ be a contraction sequence of $G$, and recall that $H_i$ denotes the subtrigraph of $G_i$ induced on the descendants of $H$ and that we call vertices not in $H_i$ \emph{outer}.

We say that $G_i$ is \emph{decisive} if $H_i \ne H_{i-1}$ or $i = 1$.
We define \emph{the big-step contraction sequence $C_{BS}$} as a subsequence of $C$ which contains $G_i$ if and only if $G_i$ or $G_{i+1}$ is decisive. We call the trigraphs in $C_{BS}$ \emph{checkpoints}.
Furthermore, we define $f_H:\mathbb{N}\rightarrow \mathbb{N}$ as follows: $f_H(\ell)=(3^{|V(H)|+4}\cdot|V(H)|^2)^{\ell}$. 
Informally, this function describes how big the centipedes will need to be in the $\ell$-th trigraph in $C_{BS}$, counted from the end (the centipedes need to be the largest at the beginning, as they shrink during each transition between checkpoints).

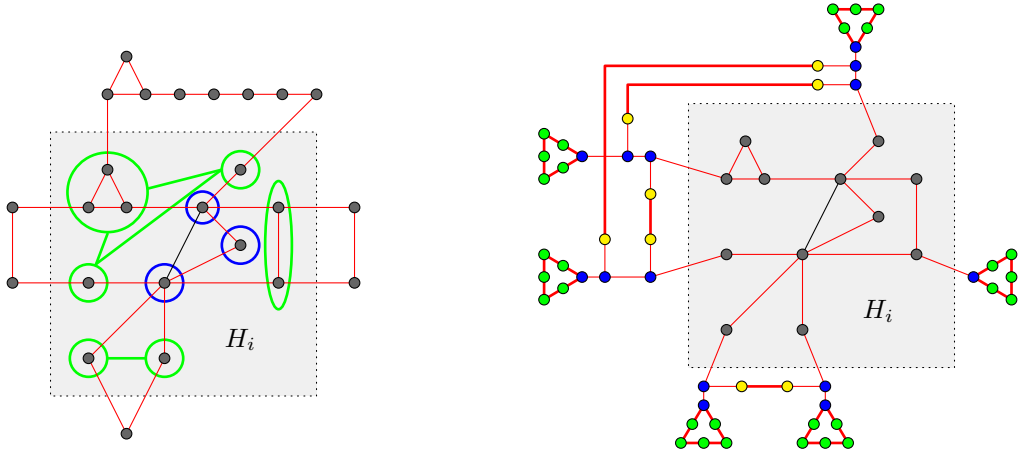
\begin{figure}
\begin{tikzpicture}[line cap=round,line join=round,>=triangle 45,x=1.0cm,y=1.0cm]
\clip(2.3,1.8) rectangle (9.3,7.2);
\fill[fill=black!60,fill opacity=0.1] (6.5,2.5) -- (6.5,6.) -- (3.,6.) -- (3.,2.5) -- cycle;
\draw[dotted] (6.5,2.5) -- (6.5,6.) -- (3.,6.) -- (3.,2.5) -- cycle;

\draw [line width = 1pt, color=blue] (5.,5.) circle (6pt);
\draw [line width = 1pt, color=blue] (4.5,4.) circle (7pt);
\draw [line width = 1pt, color=blue] (5.5,4.5) circle (7pt);
\draw [line width = 1pt, color=green] (3.5,3) circle (7pt);
\draw [line width = 1pt, color=green] (3.5,4.) circle (7pt);
\draw [line width = 1pt, color=green] (4.5,3.) circle (7pt);
\draw [line width = 1pt, color=green] (5.5,5.5) circle (7pt);

\draw [line width = 1pt, color=green] (3.75,5.2) circle (15pt);

\draw[line width = 1pt, color=green] (6,4.5) ellipse (5pt and 0.85cm);

\draw[line width = 1pt, color=green] (3.75,3)-- (4.25,3);
\draw[line width = 1pt, color=green] (3.6,4.25)-- (3.75,4.65);
\draw[line width = 1pt, color=green] (4.28,5.25)-- (5.25,5.5);
\draw[line width = 1pt, color=green] (3.6,4.25)-- (5.25,5.5);

\draw (5.,5.)-- (4.5,4.);
\draw [color=ffqqqq] (5.,5.)-- (4.,5.);
\draw [color=ffqqqq] (5.,5.)-- (6.,5.);
\draw [color=ffqqqq] (6.,5.)-- (6.,4.);
\draw [color=ffqqqq] (6.,4.)-- (4.5,4.);
\draw [color=ffqqqq] (5.,5.)-- (5.5,4.5);
\draw [color=ffqqqq] (5.5,4.5)-- (4.5,4.);
\draw [color=ffqqqq] (4.,5.)-- (3.75,5.5);
\draw [color=ffqqqq] (3.75,5.5)-- (3.5,5.);
\draw [color=ffqqqq] (3.5,5.)-- (4.,5.);
\draw [color=ffqqqq] (4.5,4.)-- (3.5,4.);
\draw [color=ffqqqq] (4.5,4.)-- (3.5,3.);
\draw [color=ffqqqq] (4.5,4.)-- (4.5,3.);
\draw [color=ffqqqq] (5.,5.)-- (5.5,5.5);

\draw [color=ffqqqq] (3.5,3.)-- (4.,2);
\draw [color=ffqqqq] (4.,2)-- (4.5,3.);
\draw [color=ffqqqq] (6.,4.)-- (7,4.);
\draw [color=ffqqqq] (6.,5.)-- (7,5.);
\draw [color=ffqqqq] (7,5.)-- (7,4.);
\draw [color=ffqqqq] (5.5,5.5)-- (6.5,6.5);
\draw [color=ffqqqq] (6.5,6.5)-- (4.25,6.5);
\draw [color=ffqqqq] (4.25,6.5)-- (3.75,6.5);
\draw [color=ffqqqq] (4.,7.)-- (4.25,6.5);
\draw [color=ffqqqq] (4.,7.)-- (3.75,6.5);
\draw [color=ffqqqq] (3.75,6.5)-- (3.75,5.5);
\draw [color=ffqqqq] (3.5,5)-- (2.5,5.);
\draw [color=ffqqqq] (2.5,5.)-- (2.5,4.);
\draw [color=ffqqqq] (2.5,4.)-- (3.5,4.);
\begin{scriptsize}
\draw [fill=black!60] (5.,5.) circle (2pt);
\draw [fill=black!60] (4.5,4.) circle (2pt);
\draw [fill=black!60] (4.,5.) circle (2pt);
\draw [fill=black!60] (6.,5.) circle (2pt);
\draw [fill=black!60] (6.,4.) circle (2pt);

\draw [fill=black!60] (5.5,4.5) circle (2pt);
\draw [fill=black!60] (3.75,5.5) circle (2pt);
\draw [fill=black!60] (3.5,5.) circle (2pt);

\draw [fill=black!60] (3.5,4.) circle (2pt);

\draw [fill=black!60] (3.5,3.) circle (2pt);
\draw [fill=black!60] (4.5,3.) circle (2pt);
\draw [fill=black!60] (5.5,5.5) circle (2pt);

\draw [fill=black!60] (4.,2) circle (2pt);
\draw [fill=black!60] (7,4.) circle (2pt);
\draw [fill=black!60] (7,5.) circle (2pt);
\draw [fill=black!60] (6.5,6.5) circle (2pt);
\draw [fill=black!60] (4.25,6.5) circle (2pt);
\draw [fill=black!60] (4.,7.) circle (2pt);
\draw [fill=black!60] (3.75,6.5) circle (2pt);
\draw [fill=black!60] (2.5,5.) circle (2pt);
\draw [fill=black!60] (2.5,4.) circle (2pt);

\foreach \i in {0.45,0.9,...,2}
{\draw [fill=black!60] (4.25+\i,6.5) circle (2pt);}
\end{scriptsize}

\draw (5.5,3.25) node {$H_i$};
\end{tikzpicture}
\begin{tikzpicture}[line cap=round,line join=round,>=triangle 45,x=1.0cm,y=1.0cm]
\fill[fill=black!60,fill opacity=0.1] (6.5,2.5) -- (6.5,6.) -- (3.,6.) -- (3.,2.5) -- cycle;
\draw[dotted] (6.5,2.5) -- (6.5,6.) -- (3.,6.) -- (3.,2.5) -- cycle;

\draw [color=ffqqqq] (6,4)-- (6.75,3.7);
\draw [color=ffqqqq, line width=1pt] (6.75,3.7)-- (7.25,4)-- (7.25, 3.4)-- cycle;
\draw [fill=blue] (6.75,3.7) circle (2pt);
\draw [fill=green] (7.25, 3.4) circle (2pt);
\draw [fill=green] (7.25,3.7) circle (2pt);
\draw [fill=green] (7.25,4) circle (2pt);
\draw [fill=green] (7.,3.85) circle (2pt);
\draw [fill=green] (7.,3.55) circle (2pt);

\draw (5.,5.)-- (4.5,4.);
\draw [color=ffqqqq] (5.,5.)-- (4.,5.);
\draw [color=ffqqqq] (5.,5.)-- (6.,5.);
\draw [color=ffqqqq] (6.,5.)-- (6.,4.);
\draw [color=ffqqqq] (6.,4.)-- (4.5,4.);
\draw [color=ffqqqq] (5.,5.)-- (5.5,4.5);
\draw [color=ffqqqq] (5.5,4.5)-- (4.5,4.);
\draw [color=ffqqqq] (4.,5.)-- (3.75,5.5);
\draw [color=ffqqqq] (3.75,5.5)-- (3.5,5.);
\draw [color=ffqqqq] (3.5,5.)-- (4.,5.);
\draw [color=ffqqqq] (4.5,4.)-- (3.5,4.);
\draw [color=ffqqqq] (4.5,4.)-- (3.5,3.);
\draw [color=ffqqqq] (4.5,4.)-- (4.5,3.);
\draw [color=ffqqqq] (5.,5.)-- (5.5,5.5);

\draw [color=ffqqqq] (3.5,5)-- (2.5,5.3);
\draw [color=ffqqqq] (2.5,3.7)-- (3.5,4.);

\draw [color=ffqqqq] (3.2,2.25)-- (3.5,3.);
\draw [color=ffqqqq] (3.2,2.25)-- (3.2,2);
\draw [color=ffqqqq] (4.8,2.25)-- (4.5,3.);
\draw [color=ffqqqq] (4.8,2.25)-- (4.8,2);
\draw [color=ffqqqq] (4.8,2.25)-- (4.3,2.25);
\draw [color=ffqqqq] (3.2,2.25)-- (3.7,2.25);
\draw [color=ffqqqq, line width=1pt] (3.7,2.25)-- (4.3,2.25);
\draw [color=ffqqqq, line width=1pt] (3.2,2)-- (2.9,1.5)-- (3.5, 1.5)-- cycle;
\draw [color=ffqqqq, line width=1pt] (4.8,2)-- (4.5,1.5)-- (5.1, 1.5)-- cycle;

\draw [color=ffqqqq] (1.6,5.3)-- (2.5,5.3);
\draw [color=ffqqqq] (1.6,3.7)-- (2.5,3.7);

\draw [color=ffqqqq] (5.5,5.5)-- (5.2,6.25);
\draw [color=ffqqqq] (5.2,6.25)-- (5.2,6.75);
\draw [color=ffqqqq] (5.2,6.25)-- (4.7,6.25);
\draw [color=ffqqqq] (5.2,6.5)-- (4.7,6.5);

\draw [color=ffqqqq] (2.5,4.8)-- (2.5, 5.3);
\draw [color=ffqqqq] (2.5,3.7)-- (2.5, 4.2);
\draw [color=ffqqqq](2.2, 5.3)-- (2.2, 5.8);
\draw [color=ffqqqq](1.9, 3.7)-- (1.9, 4.2);
\draw [color=ffqqqq, line width=1pt] (2.5, 4.8)-- (2.5, 4.2);
\draw [color=ffqqqq, line width=1pt] (2.2, 5.8)-- (2.2, 6.25)-- (4.7, 6.25);
\draw [color=ffqqqq, line width=1pt] (1.9, 4.2)-- (1.9, 6.5)-- (4.7, 6.5);

\draw [color=ffqqqq, line width=1pt] (5.2,6.75)-- (4.9, 7.25)-- (5.5, 7.25)-- cycle;

\draw [color=ffqqqq, line width=1pt] (1.1,5)--(1.1,5.6)--(1.6,5.3)-- cycle;
\draw [color=ffqqqq, line width=1pt] (1.1,3.4)--(1.1,4)--(1.6,3.7)-- cycle;

\begin{scriptsize}
\draw [fill=black!60] (5.,5.) circle (2pt);
\draw [fill=black!60] (4.5,4.) circle (2pt);
\draw [fill=black!60] (4.,5.) circle (2pt);
\draw [fill=black!60] (6.,5.) circle (2pt);
\draw [fill=black!60] (6.,4.) circle (2pt);

\draw [fill=black!60] (5.5,4.5) circle (2pt);
\draw [fill=black!60] (3.75,5.5) circle (2pt);
\draw [fill=black!60] (3.5,5.) circle (2pt);
\draw [fill=black!60] (3.5,4.) circle (2pt);

\draw [fill=black!60] (3.5,3.) circle (2pt);
\draw [fill=black!60] (4.5,3.) circle (2pt);
\draw [fill=black!60] (5.5,5.5) circle (2pt);

\draw [fill=blue] (5.2,6.25) circle (2pt);
\draw [fill=blue] (5.2,6.5) circle (2pt);
\draw [fill=blue] (5.2,6.75) circle (2pt);

\draw [fill=yellow] (4.7,6.25) circle (2pt);
\draw [fill=yellow] (4.7,6.5) circle (2pt);

\draw [fill=blue] (2.5,5.3) circle (2pt);
\draw [fill=blue] (2.5,3.7) circle (2pt);
\draw [fill=blue] (2.2,5.3) circle (2pt);
\draw [fill=blue] (1.9,3.7) circle (2pt);
\draw [fill=blue] (1.6,5.3) circle (2pt);
\draw [fill=blue] (1.6,3.7) circle (2pt);

\draw [fill=green] (1.1, 3.4) circle (2pt);
\draw [fill=green] (1.1, 4) circle (2pt);
\draw [fill=green] (1.1, 3.7) circle (2pt);
\draw [fill=green]  (1.35, 3.85) circle (2pt);
\draw [fill=green] (1.35, 3.55) circle (2pt);

\draw [fill=green] (1.1, 5) circle (2pt);
\draw [fill=green] (1.1, 5.6) circle (2pt);
\draw [fill=green] (1.1, 5.3) circle (2pt);
\draw [fill=green]  (1.35, 5.45) circle (2pt);
\draw [fill=green] (1.35, 5.15) circle (2pt);

\draw [fill=yellow] (2.5,4.8) circle (2pt);
\draw [fill=yellow] (1.9,4.2) circle (2pt);
\draw [fill=yellow] (2.5,4.2) circle (2pt);
\draw [fill=yellow] (2.2,5.8) circle (2pt);

\draw [fill=blue] (3.2,2.25) circle (2pt);
\draw [fill=blue] (4.8,2.25) circle (2pt);
\draw [fill=blue] (3.2,2) circle (2pt);
\draw [fill=blue] (4.8,2) circle (2pt);

\draw [fill=yellow] (3.7,2.25) circle (2pt);
\draw [fill=yellow] (4.3,2.25) circle (2pt);

\draw [fill=green] (2.9, 1.5) circle (2pt);
\draw [fill=green] (3.5, 1.5) circle (2pt);
\draw [fill=green] (3.2, 1.5) circle (2pt);
\draw [fill=green] (3.05, 1.75) circle (2pt);
\draw [fill=green] (3.35, 1.75) circle (2pt);

\draw [fill=green] (4.5, 1.5) circle (2pt);
\draw [fill=green] (5.1, 1.5) circle (2pt);
\draw [fill=green] (4.8, 1.5) circle (2pt);
\draw [fill=green] (4.65, 1.75) circle (2pt);
\draw [fill=green] (4.95, 1.75) circle (2pt);

\draw [fill=green] (4.9,7.25) circle (2pt);
\draw [fill=green] (5.5,7.25) circle (2pt);
\draw [fill=green] (5.2,7.25) circle (2pt);
\draw [fill=green] (5.35,7) circle (2pt);
\draw [fill=green] (5.05,7) circle (2pt);
\end{scriptsize}
\draw (5.5,3.25) node {$H_i$};
\end{tikzpicture}
\vspace{0.5cm}
\caption{\textbf{Left: }A possible trigraph $G_i$. The components of $H_i^R$ (i.e., vertices of $B_i$) are marked by green or blue circles: blue means isolated and green means non-isolated. The edges of $B_i$ are depicted as green lines. \textbf{Right:} a representative for $G_i$. Blue vertices lie in the body of a centipede, yellow vertices form legs of centipedes, and green vertices lie in the tail of a centipede. Note that the tails and the leg paths, which are drawn as thick red lines, contain many vertices.}\label{fig:sec6}
\end{figure}

Now we define the blueprints; these capture the information we need about the checkpoints.

\begin{definition}\label{def:bluep}
The \emph{blueprint} of a trigraph $G_i$ in $C$, denoted $\ca B_i$, is the tuple $(H_i, B_i)$ where $B_i$ is the vertex-labeled graph constructed in the following way:
\begin{itemize}
\item let $H^R_i$ be the subgraph obtained from $H_i$ by removing all black edges and the red edges incident to them;
\item $V(B_i)$ is the set of connected components of $H^R_i$ (the \emph{fully-red components});
\item $U\in V(B_i)$ is labeled \emph{isolated} if for all $u\in U$, we have $N_{G_i}(u)\seq V(H_i)$, and otherwise it is \emph{non-isolated};
\item $UU' \in E$ if there is a (red) path between some vertices $u\in U$ and $u'\in U'$ in $G_i$ which contains no vertex of $V(H_i)$ except for $u$ and $u'$. %
$E(B_i)$ is then the transitive closure of $E$. Observe that isolated vertices have degree 0 in $B_i$ but a non-isolated vertex may have degree 0, too; see Figure~\ref{fig:sec6} for an illustration.
\end{itemize}
\end{definition}

\iflong
A \emph{blueprint} is then a tuple which is the blueprint for some trigraph in $C$.
\fi
A reader may wonder why we define the edge relation of $B_i$ to be transitive. The reason is that with this definition, a connection can disappear only when $H_i$ changes, see Lemma~\ref{lem:non-dec_reachability-increase} below. We are now ready to define centipedes---the technical gadget underlying our entire construction.

\begin{definition}\label{def:centipede}
Let $d \ge 0$ and $\ell \ge 1$. The \emph{centipede} $\text{cen}(d, \ell)$ is the following graph:
\begin{itemize}
\item the vertex set consists of three disjoint sets: \emph{the body} $\{u_i \sep i \in [d+1]\}$, \emph{the legs} $\{v_i \sep i \in [d]\}$, and \emph{the tail} $\{w_i \sep i \in [\ell]\}$; %
\item the edge set is $\{u_iu_{i+1}, u_iv_i \sep i \in [d]\} \cup \{w_iw_{i+1} \sep i \in [\ell - 1] \} \cup \{u_{d+1}w_1, u_{d+1}w_\ell\}$.
\end{itemize}
We say that $u_1$ is the \emph{head of $\text{cen}(d, \ell)$}. %
Let $G'$ be a supergraph of $\text{cen}(d, \ell)$. We say that a leg $v$ of $\text{cen}(d, \ell)$ is \emph{free in $G'$} if its degree is 1 in $G'$.
For $x \in V(G')$, we say that $\text{cen}(d, \ell)$ is \emph{attached to $x$} if $x$ is adjacent to the head of $\text{cen}(d, \ell)$.
\end{definition}

Now we define the representatives; trigraphs representing the blueprints in $C'$. Note that $\ell$ specifies how big the centipedes in the representative need to be.

\begin{definition}\label{def:representative}
Let $\ca B_i = (H_i, B_i)$ be a blueprint and $\ell$ be an integer.
Now a \emph{representative for $\ca B_i$ of order $\ell$}, denoted $R_i^\ell$, is a trigraph that can be built as follows:

\begin{enumerate}
\item start with $H_i$; let $\ca U \seq V(B_i)$ be the set of non-isolated vertices of $B_i$;
\item for each $U \in \ca U$, add a red centipede $\cen_U \cong \text{cen}(\deg_{B_i}(U), f_H(\ell))$;
\item for all $U \in \ca U$, add a red edge between the head of $\cen_U$ and %
some vertex of $U$ which has at least one neighbor outside of $H_i$ in $G_i$ (there must be at least one such vertex because $U$ is non-isolated);%
\item for each edge $UU' \in E(B_i)$, do the following:
\begin{itemize}
\item let $w$ (resp. $w'$) be a leg of $\cen_U$ (resp. $\cen_{U'}$) free in the current trigraph. Add a red path of length $f_H(\ell)$ connecting $w$ and $w'$. We call this path, including the two legs, \emph{the leg path connecting $\cen_U$ and $\cen_{U'}$}.
\end{itemize}
\end{enumerate}
\end{definition}

Notice that the construction in Definition~\ref{def:representative} works by a simple inductive argument because the number of legs of $\cen_U$ is $\deg_{B_i}(U)$ by construction; an illustration is provided in Figure~\ref{fig:sec6}.
Moreover, observe that a representative is not uniquely determined by $i$ and $\ell$; the order of leg paths, as well as the vertices which the centipedes are attached to, can differ. \ifshort ($\clubsuit$)\fi
\iflong However, this will not be an issue in our proof since the legs can be reordered via Observation~\ref{obs:reorder-legs} and centipedes can be moved via Lemma~\ref{lem:move-cat}. \fi

\subsection{Moving Centipedes Around}\label{sub:centipedes}

In this subsection, we describe several operations, i.e., partial contraction sequences, which will later be used to obtain a partial contraction sequence which transitions from a representative of one checkpoint in $C_{BS}$ to a representative of the next checkpoint.
These operations will focus on the centipedes introduced in the previous subsection, and may result in trigraphs which are not necessarily a representative for any graph in $C$; we refer to these obtained trigraphs \emph{intermediate graphs}\iflong.~\fi\ifshort ~and note that their structure can be precisely formalized. ($\clubsuit$)\fi
\iflong
Note that by \emph{the centipedes}, we mean the centipedes added in step 2 of Definition~\ref{def:representative}, i.e., we ignore subtrigraphs of $H_i$ which happen to be centipedes by chance. \fi

\iflong
In the following definition, we state which properties these intermediate graphs have.

\begin{definition}\label{def:intermediate-graph}
Let $G'$ be a trigraph that has an induced subtrigraph $H'$ isomorphic to $H_i$ for some $i \in [n]$ and the maximum red degree in $G'$ is at most $\tww(G) + 1$. %
Moreover, the trigraph $G' - V(H')$ consists of disjoint red centipedes and some red paths connecting their legs (which we call \emph{leg paths}, as in Definition~\ref{def:representative}).
Each of these centipedes is attached to a vertex of $H'$ without black neighbors in $G'$ and none of them has a free leg in $G'$. %
If for any two centipedes in $G'$, there is at most one leg path connecting them and there is no leg path connecting a centipede to itself, then we say that $G'$ is \emph{an intermediate graph}; without this requirement, we say that $G'$ is \emph{a generalized intermediate graph}.

For a centipede $\cen$ in $G'$, let $N(\cen)$ be the set of all centipedes $\cen'$ s.t. there is a leg path connecting $\cen$ and $\cen'$ (possibly, $\cen \in N(\cen)$).
\end{definition}
\fi

We start with a few simple operations on generalized intermediate graphs:

\begin{observation}[Shortening a centipede]
\label{obs:shortening}
We can \emph{shorten the tail} of a centipede $\cen$, by contracting two neighboring vertices belonging to the tail, to any length.  %
Similarly, we can \emph{shorten a leg path} to any non-zero length.
\end{observation}

\iflong
\begin{observation}[Destroying a centipede]
\fi
\ifshort
\begin{observation}[Destroying a centipede, $\clubsuit$]
\fi
\label{obs:destroy-cat}
Let $G'$ be an intermediate graph, let $u \in V(H')$, and let $\cen = \text{cen}(0, \ell)$ be a centipede attached to $u$.
We can \emph{destroy $\cen$}, i.e., there is a  partial contraction sequence of width at most $\tww(G) + 1$ from $G'$ to $G' - V(\cen)$.
\end{observation}
\iflong \begin{proof}
First observe that $N(\cen) = \emptyset$. We start by shortening the tail of $\cen$ to a single vertex $v$, as per Observation~\ref{obs:shortening}. After that, the centipede is isomorphic to $\text{cen}(0, 1)$, i.e., it consists of only two vertices, $v$ and $w$, and $w$ is the only red neighbor of $v$. Thus we simply contract $v$ and $w$, and then the obtained vertex and $u$.
  \end{proof} \fi

\ifshort
The remaining operations allow us to 
\begin{enumerate}
\item reorder the legs of a centipede ($\clubsuit$),
\item move a centipede to a different vertex of the core ($\clubsuit$),
\item connect two centipedes by a new leg path ($\clubsuit$),
\item merge two centipedes into a single centipede ($\clubsuit$), and
\item split a centipede into two new centipedes ($\clubsuit$).
\end{enumerate}

Each of these operations can be formalized by carefully prescribing the initial and final intermediate graph, the impact on the length of the involved centipedes, and a proof ensuring that the contraction subsequence between the initial and final intermediate graph has width at most $\tww(G)+1$. ($\clubsuit$)
\fi

\iflong %
\iflong
\begin{observation}[Reordering legs]
\fi
\ifshort
\begin{observation}[Reordering legs, $\clubsuit$]
\fi

\label{obs:reorder-legs} 
Let $G'$ be a generalized intermediate graph with a centipede $\cen = \text{cen}(d, \ell)$.
Let $v_1$ and $v_2$ be two consecutive legs of $\cen$, i.e., $\dist_{\cen}(v_1, v_2) = 3$.
For $i \in [2]$, let $N(v_i) = \{u_i, w_i\}$, let $u_i$ be in the body of $\cen$, and suppose that $w_i$ is not a leg of a centipede.
We can \emph{swap} $v_1$ and $v_2$, i.e., there is a partial contraction sequence of width at most $\tww(G) + 1$ from $G'$ to a generalized intermediate graph $G''$, where $G''$ is obtained by deleting the vertices $v_1$ and $v_2$ and adding the (red) edges $u_1w_2$ and $u_2w_1$.
\end{observation}
\iflong \begin{proof}
First contract $v_1$ and $u_2$; the obtained vertex has red degree 4. Then contract $v_2$ and $u_1$, which yields a trigraph isomorphic to $G''$. Observe that $w_1$ and $w_2$ are legs of (a centipede isomorphic to) $\cen$ in $G''$. Since we assumed that these two vertices are not legs in $G'$, the centipedes remain disjoint in $G''$, as required by Definition~\ref{def:intermediate-graph}.
  \end{proof} \fi

Now we continue with slightly more complicated operations. 
\iflong \begin{lemma}[Moving a centipede] \fi \ifshort \begin{lemma}[Moving a centipede, $\clubsuit$]\fi
\label{lem:move-cat}

Let $G'$ be an intermediate graph, let $u,v$ be two vertices in $H'$, and let $\cen = \text{cen}(d, \ell)$ be a centipede attached to $u$, where $d \ge 0$ and $\ell \ge 3$. If all leg paths going from $\cen$ have length at least 2, we can \emph{move $\cen$ from $u$ to $v$}, i.e., there is a partial contraction sequence of width at most $\tww(G) + 1$ from $G'$ to an intermediate graph $G''$, where $G''$ is obtained from $G'$ by:
\begin{compactitem}
\item deleting $\cen$ and all its leg paths, and adding a centipede $\cen' = \text{cen}(d, \ell -2)$, attached to $v$;
\item for each deleted leg path from $\cen$ to some $\cen''$ of length $m$, we add a leg path from $\cen'$ to $\cen''$ of length $m-1$;
\item turning red all edges incident to $v$, and adding a red edge between $u$ and $v$.
\end{compactitem}
\end{lemma}
\iflong \begin{proof}
Let $w_1$ be the head of $\cen$, let $(w_1, \ldots, w_{d+1})$ be the body of $\cen$, let $(x_1, \ldots, x_d)$ be the legs of $\cen$, and suppose $w_ix_i \in R(G')$ for all $i \in [d]$.
Let $y_1$ and $y_2$ be the only two vertices of the tail of $\cen$ adjacent to the body, namely to $w_{d+1}$.
First, we contract $y_1$ and $y_2$ into a new vertex called $y$.
Then, we gradually contract $x_d$ and $w_{d+1}$, followed by contracting $x_{d-1}$ and $w_d$, and so on until we contract $x_1$ and $w_2$ (if $d = 0$, then none of these contractions are carried out).
Finally, we contract $w_1$ and $v$. Since here $w_1$ has a red edge to $u$ and no incident black edges, this contraction turns all edges incident to $v$ red, and creates the red edge $uv$ if it did not exist already.
Observe that now there is a centipede with body $(w_2, \ldots, w_{d+1}, y)$ attached to $v$, as required by the statement.
  \end{proof} \fi

Next we show how to prolong the body and create new leg paths.

\iflong \begin{lemma}[Connecting centipedes] \fi \ifshort \begin{lemma}[Connecting centipedes, $\clubsuit$]\fi\label{lem:connect-cats}
Let $G'$ be an intermediate graph, let $\ell \ge 2$, and let $\cen_1 = \text{cen}(d_1, t_1)$ and $\cen_2 = \text{cen}(d_2, t_2)$ be two centipedes in $G'$,
where $d_1, d_2 \ge 0$, $t_1 \ge 2\ell +5$, and $t_2 \ge 3$.
We can \emph{connect $\cen_1$ and $\cen_2$ by a leg path of length $\ell-1$}, i.e., there is partial contraction sequence of width at most $\tww(G) + 1$ from $G'$ to an intermediate graph $G''$, where $G''$ is obtained from $G'$ as follows:
\begin{compactitem}
\item replace $\cen_1$ by $\cen_1' = \text{cen}(d_1+1, t_1 - 2\ell - 4)$ and $\cen_2$ by $\cen_2' = \text{cen}(d_2+1, t_2 - 2)$, without deleting the leg paths going from $\cen_1$ and $\cen_2$ (each of the two new centipedes has a new leg; we assume it is the one closest to the tail);
\item add a leg path of length $\ell-1$ from $\cen_1'$ to $\cen_2'$, using the new legs.
\end{compactitem}
\end{lemma}
\iflong \begin{proof}
Recall that $\tww(G) \ge 3$, and so it suffices to avoid vertices of red degree 5 or higher.
For $i \in [2]$, let $(u_1^i, u_2^i, \ldots, u^i_{t_i} = u^i_0 )$ be the tail of $\cen_i$ and let $v_i$ be the vertex of the body of $\cen_i$ adjacent to $u_1^i$ and $u_0^i$. Let $v_i'$ be the neighbor of $v_i$ in the body of $\cen_i$ (or in $H'$ if $d_i = 0$). 
We start by contracting $u_1^2$ and $u_0^2$, into $w_2$.
Then we contract $v_2$ and $u^1_{\ell+1}$, into $x_2$. Observe that $x_2$ has four red neighbors: $u^1_\ell$, $u^1_{\ell+2}$, $w_2$, and $v_2'$.
Then we contract $u^1_\ell$ and $u^1_{\ell+2}$, into $y_1$, then $u^1_{\ell-1}$ and $u^1_{\ell+3}$, into $y_2$, and so on until we contract $u^1_1$ and $u^1_{2\ell+1}$ into $y_\ell$.
Observe that for all $i \in [\ell]$, $y_i$ has red degree 3 at the moment of its creation (and 2 later if $i < \ell$).

Now we contract $v_1$ with $u^1_{2\ell+2}$, into $x_1$, which has four red neighbors: $y_\ell$, $u^1_0$, $u^1_{2\ell+3}$ and $v_1'$.
Finally, we contract $u^1_0$ and $u^1_{2\ell + 3}$, into $w_1$.
Observe that the requirements of the statement are now satisfied: for $i \in [2]$, $\cen_i'$ has the same body and legs as $\cen_i$ up to $v_i'$ (this is void if $d_i = 0$), then the body continues with $x_i$ (which is adjacent to a leg, namely $y_1$ or $y_\ell$), and then the body ends with $w_i$ (which does not have a leg and is adjacent to the tail of $\cen_i'$).
  \end{proof} \fi

\iflong \begin{lemma}[Merging centipedes] \fi \ifshort \begin{lemma}[Merging centipedes, $\clubsuit$]\fi\label{lem:merge-cats}
Let $G'$ be an intermediate graph and assume that all leg paths in $G'$ have length at least $4m+2$, where $m$ is the number of centipedes in $G'$.
Let $\cen_1 = \text{cen}(d_1, t_1)$ and $\cen_2 = \text{cen}(d_2, t_2)$ be two centipedes in $G'$, both attached to $u \in V(H')$,
where $t_2 \ge 2d_1 + 1$.
For $i \in [2]$, let $N_i := N(\cen_i)$, %
and let $N := N_1 \cup N_2 \setminus \{\cen_1, \cen_2\}$.
We can \emph{merge $\cen_1$ and $\cen_2$}, i.e., there is partial contraction sequence of width at most $\tww(G) + 1$ from $G'$ to an intermediate graph $G''$, where $G''$ is obtained from $G'$ as follows:
\begin{compactitem}
\item remove $\cen_1$ and $\cen_2$ and their leg paths, and add a centipede $\cen = \text{cen}(|N|, t_2-2d_1)$, attached to $u$;
\item for $\cen' \in N_1 \cap N_2$, remove one of its two free legs (which were freed in the previous step) and its neighbor $v$, and then add a red edge between the two remaining neighbors of $v$ (which are both in the body of $\cen'$);
\item let $\cen' \in N$ be a centipede connected to $\cen_i$ by a leg path of length $\ell_i$ in $G'$ for some $i \in [2]$; if $\cen'$ is not connected to $\cen_{3-i}$, then we set $\ell_{3-i} = \infty$; add a leg path from $\cen$ to $\cen'$ of length $\min(\ell_1, \ell_2)$;
\item shorten all leg paths (not only those going from $\cen$) by $4m$ vertices.
\end{compactitem}
\end{lemma}
\iflong \begin{proof}
For $i \in [2]$, let $(v^i_1, \ldots, v^i_{d_i+1})$ be the body of $\cen_i$, $(w^i_1, \ldots, w^i_{d_i})$ be the legs of $\cen_i$, and suppose that for all $j \in [d_i]$, $v^i_jw^i_j$ is an edge.
Let $(x^i_1, x^i_2, \ldots, x^i_{-2}, x^i_{-1})$ be the tail of $\cen_i$, i.e., each vertex of the tail has two indices, namely $j $ and $j - t_i - 1$ where $j > 0$.
Suppose $x^i_1$ and $x^i_{-1}$ are neighbors of $v^i_{d_i+1}$.
First, for both $i \in [2]$, we prolong the body of $\cen_i$ by the number of legs of the other centipede $\cen_{3-i}$, i.e., we contract $x^i_j$ and $x^i_{-j}$ into a new ``body vertex'', called $v^i_{d_i+1+j}$, for all $j \in [d_{3-i}]$ in increasing order.

Second, we move all legs of $\cen_1$ up by $d_2$ positions, i.e., for all $i \in [d_1]$ in decreasing order, we do the following: if $(y_1  = w^1_i, \ldots, y_\ell)$ is a leg path, then for all $j \in [d_2]$ in increasing order, contract $y_j$ and $v^1_{i+j}$.
Let $d = d_1 + d_2$.
Observe that now $v^1_i$ for $i \in [d_2]$ and $v^2_i$ for $i \in [d_2 + 1, d]$ are ``body vertices without legs'', and $v^1_i$ for $i \in [d_2+1, d]$ and $v^2_i$ for $i \in [d_2]$ are ``body vertices with legs''. %
Let us relabel the legs: for $i \in [2]$ and $j \in [d]$, let $w_j$ be the leg adjacent to $v^i_j$ (we do not need to specify the superscript since $v^1_j$ has a leg if and only if $v^2_j$ does not, as argued in the previous sentence).

Third, for $i \in [d]$ in increasing order, contract $v^1_i$ and $v^2_i$, into a new vertex called $v_i$. Observe that $v_i$ has four red neighbors when it is created: $v_{i-1}$ (or $u$ if $i = 1$), $v^1_{i+1}$, $v^2_{i+1}$, and $w_i$. %
Fourth, we destroy the tail of $\cen_1$ by applying Observation~\ref{obs:shortening} until it is reduced to a single vertex $x$,
then we contract $x$ with $v^1_{d+1}$, and finally we contract the obtained vertex with its only neighbor $v_d$.

Now $\cen_1$ and $\cen_2$ are merged into a single centipede $\cen$.
Observe that the obtained trigraph is a generalized intermediate graph, see Definition~\ref{def:intermediate-graph}.
However, $G''$ is required to be non-generalized, and so we must deal with pairs of leg paths going from $\cen$ to the same centipede and with a leg path connecting two legs of $\cen$.

We repeatedly ``swap legs'', i.e., we use Observation~\ref{obs:reorder-legs}, to ensure that: 1) if there is a leg path connecting two legs of $\cen$, then these legs are $w_i$ and $w_{i+1}$ for some $i \in [d-1]$; 2) if there are two leg paths going from $\cen$ to some $\cen'$, then a) their legs in $\cen$ are $w_i$ and $w_{i+1}$ for some $i \in [d-1]$, and b) their legs in $\cen'$ are consecutive as well (their distance is $3$ in $\cen'$).
Observe that 1) and 2a) can be performed together in such a way that each leg path going from $\cen$ is shortened by at most $d$ vertices (this can be achieved by using, e.g., bubble sort).
Similarly, 2b) can be performed in such a way that each leg path going from $\cen'$ (but not necessarily going to $\cen$) is shortened by at most $m$ vertices (recall that $m$ is the number of centipedes in $G'$) since the length of the body of $\cen'$ is at most $m$.
Also recall that leg paths originally going from $\cen_1$ were already shortened by $d_2$ vertices (when they were moved up).
Now simply observe that $d + d_2 + m \le 4m$, as follows from the last sentence of the statement.
We now shorten all leg paths (not only those going from $\cen$) so that each will have been shortened by exactly $4m$ vertices by the entire process.

Suppose there is a leg path $P = (w_i, \ldots, w_{i+1})$ for some $i \in [d-1]$.
We repeatedly use Observation~\ref{obs:shortening} to shorten $P$ until it consists of a single vertex $y$.
At this point there is a triangle $v_iv_{i+1}y$, which we contract into a single vertex.
This vertex is now a ``body vertex without legs'', and we contract it with one of its neighbors.
This way we have dealt with a leg path connecting two legs of $\cen$.
Now let $d := d-2$, and relabel the body and legs of $\cen$ accordingly.

Now suppose there is a centipede $\cen' \ne \cen$ and two leg paths $Y = (y_1 = w_i, \ldots, y_\ell)$ and $Z = (z_1 = w_{i+1}, \ldots, z_{\ell'})$ for some $i \in [d-1]$ s.t. $y_\ell$ and $z_{\ell'}$ are two distinct legs of $\cen'$.
Without loss of generality, suppose $\ell' \ge \ell$, and shorten $Z$ so that its length is $\ell - 1$.
Contract $v_i$ and $v_{i+1}$ and then for $j \in [\ell]$ in increasing order, contract $y_j$ and $z_j$ into $a_j$.
Finally, contract the two vertices in the body of $\cen'$ which are both adjacent to $a_\ell$ (recall that these are adjacent by point 2b) above).
To conclude the proof, we repeat the process described in this paragraph for all other $\cen'$ satisfying these properties.
  \end{proof} \fi

\iflong \begin{lemma}[Creating a new centipede] \fi \ifshort \begin{lemma}[Creating a new centipede, $\clubsuit$]\fi\label{lem:new-cat}
Let $G'$ be an intermediate graph, let $t' \ge 1, \ell \ge 2$,
let $\cen_0 = \text{cen}(d, t)$, where $t \ge 2\ell + 2t' + 15$, and let $u \in V(H')$ be a vertex with no black neighbors of degree at most $\tww(G) - 1$.
We can \emph{create a new centipede attached to $u$ using $\cen$}, i.e., there is partial contraction sequence of width at most $\tww(G) + 1$ from $G'$ to an intermediate graph $G''$, where $G''$ is obtained from $G'$ as follows:
\begin{compactitem}
\item replace $\cen_0$ with $\cen_1 = \text{cen}(d+1, t - 2\ell - 2t' - 14)$ while keeping all leg paths going from it;
\item add a centipede $\cen_2 = \text{cen}(1, t')$, attached to $u$;
\item add a leg path of length $\ell-1$ connecting the free legs of $\cen_1$ and $\cen_2$. 
\end{compactitem}
\end{lemma}

\iflong \begin{proof}
Let $(v_1,\ldots, v_{d+1})$ be the body of $\cen_0$ and $(w_1, \ldots, w_t)$ be the tail of $\cen_0$. Suppose $w_1$ and $w_t$ are neighbors of $v_{d+1}$.
Let $m = \ell + t' + 6$.
Start by contracting $w_{m+1}$ with $u$ into a new vertex $u'$, which has red degree at most $\tww(G) + 1$.
Then, for $i \in [m]$ in increasing order, contract $w_{m+1-i}$ and $w_{m+1+i}$, into a new vertex called $x_i$.
Observe that the neighbors of $x_m$, after it is created, are $x_{m-1}$, $w_{2m+2}$, and $v_{d+1}$.
Contract $x_m$ and $w_t$ into $v_{d+2}$ (its red neighbors are $x_{m-1}$, $w_{2m+2}$, $v_{d+1}$, and $w_{t-1}$),
and then $x_{m-1}$ and $v_{d+1}$ into $v_{d+1}'$ (its red neighbors are $v_d$, $x_{m-2}$ and $v_{d+2}$).
Observe that now we have replaced $\cen_0$ with a centipede $\text{cen}(d+1, t - 2m - 2)$: its body is $(v_1, \ldots, v_d, v_{d+1}', v_{d+2}$).

Now contract $x_2$ with $x_{t' + 3}$ into $y_2$ (its neighbors are $x_1, x_3, x_{t'+2}$, and $x_{t'+4}$),
and $x_{t'+4}$ with $x_1$ into $y_1$ (its red neighbors are $u'$, $x_{t'+5}$ and $y_2$).
Now observe that we have created a centipede $\text{cen}(1, t')$: its body is $y_1$ and $y_2$, and its tail is $(x_3, \ldots, x_{t'+2})$.
Finally, observe that $(x_{t'+5}, \ldots, x_{m-2})$ is a leg path of length $m - 2 - t' - 5 = m - t' - 7 = \ell - 1$, which concludes the proof.
  \end{proof} \fi

\fi

\subsection{Contraction Sequences for Representatives}\label{sub:repr}

In this subsection, we will construct partial contraction sequences between the representatives of consecutive checkpoints by making use of the operations with the centipedes defined in Subsection~\ref{sub:centipedes}. We distinguish between two cases depending on whether the core changes between the checkpoints or not, see Lemmas~\ref{lem:non-dec} and~\ref{lem:dec} below.

We start by observing that if the core does not change between two consecutive checkpoints, then the blueprint of the latter checkpoint contains all edges present in the blueprint of the former checkpoint.

\iflong \begin{lemma} \fi \ifshort \begin{lemma}[$\clubsuit$]\fi\label{lem:non-dec_reachability-increase}
Let $G_i$ and $G_j$ be two consecutive trigraphs in $C_{BS}$ s.t. $H_i = H_j$. This means that $V(B_i) = V(B_j)$, see Definition~\ref{def:bluep}.
It holds that $E(B_i) \seq E(B_j)$. 
\end{lemma}
\iflong \begin{proof}
Let $U_1, U_2 \in V(B_i)$ be such that %
$U_1U_2 \in E(B_i)$.
Suppose that there are vertices $u_1 \in U_1$, $u_2 \in U_2$ s.t. there is a path between $u_1$ and $u_2$ in $G_i$ which contains no vertices of $V(H_i)$ except for $u_1$ and $u_2$.
Let us call this path $P = (p_1 = u_1, p_2, \ldots, p_\ell = u_2)$.
For $a \in [\ell]$, let $p_a'$ be the descendant of $p_a$ in $G_j$, and let $P' = (p_1', \ldots, p_\ell')$.
Observe that if $p_a'$ is a descendant of $H$, then $p_a' \in U$ for some $U \in V(B_j)$ because $H_i = H_j$. %

We say that a subsequence $(p_a', \ldots, p_b')$ of $P'$ is a \emph{segment} if $p_a'$ and $p_b'$ are descendants of $H$, $b > a + 1$, and for all $a < c < b$, $p_c'$ is not a descendant of $H$.
By Definition~\ref{def:bluep}, each segment adds an edge to $B_j$, namely between the two vertices of $B_j$ which contain $p_a'$ and $p_b'$.
Observe that the first segment starts in $U_1$ and the last one ends in $U_2$, and so the segments induce a walk from $U_1$ to $U_2$ in $B_j$.
Since $E(B_j)$ is transitive, we obtain that $U_1U_2 \in E(B_j)$.

Finally, suppose that there is no such path $P$. By Definition~\ref{def:bluep}, $U_1U_2 \in E(B_i)$ only because of the transitive closure, and so we may conclude by induction.
  \end{proof} \fi
  
We are now ready to define the partial contraction sequence between two representatives in the case when the core does not change.

\iflong \begin{lemma} \fi \ifshort \begin{lemma}[$\clubsuit$]\fi\label{lem:non-dec}
Let $G_i$, $G_j$ ($i<j$) be two consecutive trigraphs in $C_{BS}$, such that $H_i= H_j$. 
For any $\ell\in \mathbb{N}$, there is a partial contraction sequence from $R_i^{\ell+1}$ to $R_j^\ell$ whose width is at most $\tww(G)+1$.
\end{lemma}

\ifshort
\begin{proof}[Proof Sketch]
Even though the number of contraction happening between $G_i$ and $G_j$ in $C$ can be huge, the effect on the blueprints (and thus the representatives) is somewhat limited. After checking what could differ between the two representatives, we present a sequence of operations on the centipedes---namely moving, creating, connecting, shortening and destroying centipedes---which is sufficient to obtain $R_j^\ell$ from $R_i^{\ell+1}$. 

To prove that this sequence of operation on centipedes is feasible, we verify that the tails and leg paths are sufficiently longer in $R_i^{\ell+1}$ than in $R_j^\ell$ to sustain the operations without becoming too short. Moreover, the control we have over the operations enables us to make sure that no vertex has a red degree higher than $tww(G)+1$ at any point in the created contraction sequence.
\end{proof}
\fi

\iflong \begin{proof}
First, we recall that $H_i=H_j$ implies $V(B_i)=V(B_j)$ and that Lemma~\ref{lem:non-dec_reachability-increase} gives us $E(B_i)\subseteq E(B_j)$. We start by showing that if $U \in V(B_j)$ is isolated in $B_i$ and non-isolated in $B_j$, then the degree of $U$ in $B_j$ is at least 1. By Definition~\ref{def:bluep}, there is a vertex $u \in U$ that is adjacent to an outer vertex $v$ in $G_j$. Since $U$ is isolated in $B_i$, there must be two outer vertices $u', v' \in V(G_i)$ such that $u'v' \in R(G_i)$ and $u', v'$ are ancestors of $u, v$. Since $G_i$ is connected, there must be a path $P$ from $u'$ to a vertex $w \in V(H_i)$ in $G_i$ such that $w$ is the only vertex of $H_i$ in $P$. However, $U$ is isolated in $G_i$, which means $w \notin U$. Now look at the descendants of $P$ in $G_j$ and observe that they form a walk from $u$ to $w \in V(H_j) \setminus U$. The presence of this walk in $G_j$ implies $UU' \in E(B_j)$, where $U'$ is the fully-red component containing $w$.

Now we explain how to create a partial contraction sequence from $R_i^{\ell+1}$ to $R_j^\ell$. First, for each component having a centipede in $R_i^{\ell+1}$ but not in $R_j^\ell$, we use Observation~\ref{obs:destroy-cat} to destroy the centipede in $R_i^{\ell+1}$.
Additionally, for each $U\in V(B_j)$, if the centipede $\cen_U$ belonging to $U$ is attached to a vertex $u \in U$ such that $N_{G_j}(u) \seq V(H_j)$, we move $\cen_U$ to a vertex $v \in U$ such that $N_{G_j}(v) \nsubseteq V(H_j)$ using Lemma~\ref{lem:move-cat}, by following a path from $u$ to $v$ contained in $U$ (such path exists by connectivity of $U$). Note that such vertex $v$ must exist because $U$ is non-isolated since it has a centipede.
Let us call the obtained trigraph $R^*$.

Let $S = \{UU' \sep UU' \in E(B_j) \setminus E(B_i)\}$, i.e., $S$ contains pairs of components which we need to connect.
Suppose there is $UU' \in S$ such that $U$ has a centipede $\cen_0=\text{cen}(d,t)$ in $R^*$ but $U'$ does not. %
Using Lemma~\ref{lem:new-cat}, we create a centipede $\cen_1=\text{cen}(1,t')$, where $t'=(t-2f_H(\ell)-16)/3$, attached to any vertex of $U'$ which is adjacent to an outer vertex in $G_j$ and with a leg path of length $f_H(\ell)$ to $\cen_0$. Note that the centipede remaining in place of $\cen_0$ after this operation is $\text{cen}(d+1, t')$ since $t - 2(f_H(\ell)+1) - 2t' - 14 = t'$.

We now %
remove $UU'$ from $S$ and repeat this step as long as we can find such a pair in $S$. After this, it is easy to see that every component $U$ appearing in a pair in $S$ has a centipede $\cen_U$ attached to it.
For each remaining pair $UU'$ in $S$, we now use Lemma~\ref{lem:connect-cats} to create a leg path of length $f_H(\ell)$ from $\cen_U$ to $\cen_{U'}$, and this shortens the tail of one centipede by $2\cdot f_H(\ell)+6$ and the tail of the other one by $2$.
We repeat this step until $S$ is empty.

In the last step, we use Observation~\ref{obs:shortening} to shorten each of the leg paths and tails until they are of length $f_H(\ell)$, and we obtain the representative $R_j^\ell$ that we were seeking.
Note that at any point during the presented steps, all outer vertices always have red degree at most $4\leq \tww(G)+1$ (since this property is guaranteed by each operation defined in Subsection~\ref{sub:centipedes}).
As for the vertices of $H_i$, they can have at most $1$ neighbor outside of $H_i$ at once---since at most one centipede is attached to each fully-red component---and hence their red degree is also never higher than $\tww(G)+1$.

To complete the proof, we only need to check that we can always use the described operations with centipedes, namely that the tails of the centipedes are always long enough. Notice that we move each centipede a distance of at most $|V(H_i)|$, the number of times a new centipede is created is upper-bounded by $|V(B_i)|$, and the number of new leg paths created from any given centipede is upper-bounded by $|V(B_i)|$; thus it is enough to check that the tail of every centipede in $R_i^{\ell+1}$ is long enough to sustain all of these operations. 
We remark that the estimates provided below are overestimations of the upper bounds which focus on simplicity of presentation rather than tightness. We begin by recalling how much each operation shortens the tail. If the tail of a centipede $\cen$ has length $x$, then:

\begin{itemize}
\item after creating a connection to another centipede as per Lemma~\ref{lem:connect-cats}, the tail of $\cen$ has length $\alpha(x) := x - 2\cdot f_H(\ell) - 6$;
\item after creating a new centipede connected to $\cen$ as per Lemma~\ref{lem:new-cat}, the tail of $\cen$ has length $\beta(x) := (x - 2\cdot f_H(\ell)-16)/3$, see also above;
\item after moving $\cen$ by one vertex, its tail has length $\gamma(x) := x - 2$.
\end{itemize}

We shall use standard notation for function composition; namely, for $g \in\{\alpha,\beta,\gamma\}$ and $i\in \mathbb{N}$, we denote by $g^i(x)$ the identity function if $i=0$, and $g(g^{i-1}(x))$ otherwise. 

Recall that we first move the centipedes, then we create the new centipedes, and finally we add the connections between existing centipedes. Thus, what we want to ensure is: 

$$\alpha^{|V(B_i)|}\left(\beta^{|V(B_i)|}\left(\gamma^{|V(H_i)|}(f_H(\ell+1)\right)\right) \geq f_H(\ell) $$

Clearly, $f_H(\ell) \geq 1$, which immediately yields $\alpha(x)\geq x - 8 f_H(\ell)$.
Let us denote $t := |V(H)|$ and observe that $1 \le |V(B_i)|\leq |V(H_i)| \leq t$.
We approximate $\beta$ by observing:
\[\beta^t(x) \ge \frac{x}{3^t} - t\cdot(2\cdot f_H(\ell)+16) \]

Using these two approximations, we progressively strengthen the desired inequality to make it cleaner.

\begin{align*}
\left( \dfrac{f_H(\ell+1)-2t}{3^t} - t\cdot (2f_H(\ell)+16) \right) - 8tf_H(\ell) &\geq f_H(\ell) \\
\dfrac{f_H(\ell+1)-2t}{3^t} &\ge f_H(\ell) \cdot (1 + 8t + 2t) + 16t \\
f_H(\ell+1) &\ge f_H(\ell) \cdot 27t \cdot 3^t + 2t \\
f_H(\ell+1) &\ge 3^t \cdot 29t \cdot f_H(\ell)
\end{align*}

Which is true for any positive integer $\ell$ according to the definition of $f_H$; let us recall that $f_H(\ell)=(3^{t+4}t^2)^{\ell}$.

We have shown that the tails are long enough to allow for the required operations.
As for the leg paths, we will show that each leg path that existed already in $R_i^{\ell+1}$ has length at least $f_H(\ell+1)-t$ after all the operations (and hence it can be shortened to $f_H(\ell) \le f_H(\ell+1)-t$ in $R_j^\ell$). Indeed, the leg paths are shortened only by moving centipedes (not by creating new centipedes or connections) and each such path underwent at most $|V(H_i)|\le t$ move operations.
Finally, the newly-created leg paths have the desired length $f_H(\ell)$ and they are never shortened.
\end{proof} \fi

Next, we define the partial contraction sequence between two representatives in the second and final case, namely the case when the core does change.
Note that in this case, we allow the sequence to terminate in a slightly different trigraph (which will be handled by Lemma~\ref{lem:suffix_contract}).

\iflong \begin{lemma} \fi \ifshort \begin{lemma}[$\clubsuit$]\fi\label{lem:dec}
Let $G_i$, $G_j$ ($i<j$) be two consecutive trigraphs in $C_{BS}$ such that $H_i\neq H_j$. 
For any $\ell\in \mathbb{N}$, there is a partial contraction sequence $C^p$ from $R_i^{\ell+1}$ to a trigraph that is a pseudoinduced subtrigraph of $R_j^\ell$ such that $w(C^p) \le \tww(G)+1$.
\end{lemma}

\ifshort
\begin{proof}[Proof Sketch]
By construction of $C_{BS}$, since $H_i \neq H_j$, we know that $j=i+1$, i.e., that there are vertices $u,v \in V(G_i)$ such that $G_j$ is the trigraph obtained from $G_i$ by contracting $u$ and $v$. However, even though one contraction was sufficient to transfer from $G_i$ to $G_j$, more operations might be necessary to move from $R_i^{\ell+1}$ to $R_j^\ell$. Indeed, the construction of our representatives makes it so that the contractions occurring on the outer vertices are not discernible in the blueprints, until some decisive contraction makes all of these changes appear at once. 

To establish the lemma, we prove that a careful sequence of operations on the centipedes---moving, creating, connecting, shortening and destroying---can be used to transform our initial representative into the target one, without creating vertices whose red degree is too high. Similarly as in the previous lemma, here we once again make use of the fact that the leg paths and tails of the centipedes in the initial representative are much longer than in the target representative.
\end{proof}
\fi

\iflong \begin{proof} 
By construction of $C_{BS}$, since $H_i \neq H_j$, we know that $j=i+1$, i.e., that there are vertices $u,v \in V(G_i)$ such that $G_j$ is the trigraph obtained from $G_i$ by contracting $u$ and $v$. Let $u'$ be the new vertex in $G_j$.

Let us first consider the case that $u,v \in V(H_i)$.
Suppose that $u$ and $v$ are in the same component of $H_i^R$.
Observe that $w \in \{u, v\}$ has black degree 0 in $H_i$; otherwise all edges incident to $w$ would be deleted in $H_i^R$ and $\{w\}$ would be a fully-red component.
For this reason, the contraction of $u$ and $v$ affects the rest of $G_i$ only by decreasing the red degree of some vertices by 1. In particular, $B_i = B_j$.
Thus if we contract $u$ and $v$ in $R_i^{\ell+1}$ and then shorten the tails and the leg paths of all centipedes from length $f_H(\ell+1)$ to length $f_H(\ell)$ using Observation~\ref{obs:shortening}, we obtain $R_j^\ell$ (it is clearly a pseudoinduced subtrigraph of itself).

Now suppose that $u,v$ are in fully-red components $U, V$ and $U \ne V$. 
In this case, the blueprint changes because $U$ and $V$ are merged together in $B_j$, into a new component called $U'$.
Let $\cen_U$ and $\cen_V$ be the centipedes attached to $U$ and $V$ in $R_i^{\ell+1}$ (possibly, one or both of them does not exist).
Now to get from $R_i^{\ell+1}$ to $R_j^\ell$, we first contract $u$ with $v$.
Then, if both $\cen_U$ and $\cen_V$ exist, we move $\cen_U$ to the vertex which $\cen_V$ is attached to using Lemma~\ref{lem:move-cat} by following a path in $U'$ and we merge the two centipedes using Lemma~\ref{lem:merge-cats}. Let $R^*$ be the obtained trigraph.

Observe that contracting $u$ and $v$ has also other effects, namely a black neighbor $w$ of $u$ or $v$ may lose its last black edge because of the contraction, which means that the component containing $w$ (necessarily the singleton $\{w\}$), can be merged with other components. Thus, for each such vertex $w$, let $W \in V(B_j)$ be the component containing $w$ and suppose at least one centipede $\cen$ is attached to a vertex $w' \in W$ in $R^*$. %
Now we move one by one all centipedes attached to $W$ to $w'$ using Lemma~\ref{lem:move-cat} by following a path contained in $W$, and merge them there (as soon as there are two together) into one using Lemma~\ref{lem:merge-cats}.
Note that it is possible that $W = U'$ but it is not necessary: $u'$ may have black neighbors in $H_j$.
Finally we shorten the paths and tails of all centipedes to length $f_H(\ell)$ using Observation~\ref{obs:shortening}.

\medskip
Let us now look at the other case: without loss of generality, $u\in V(H_i)$, and $v\notin V(H_i)$. %
The effects of contracting $u$ and $v$ on $G_i$ are the following: turning red (previously black) edges incident to $u$ and adding red edges between $u$ and $N_{G_i}(v)$.
This translates to the blueprint in this way: $H_j$ is equal to $H_i$ after turning red all edges incident to $u$ and adding new red edges between $u$ and $N_{H_i}(v)$. Thus, components of $H_j^R$ can be different from the ones of $H_i^R$, but only the component $U'$ containing $u'$ (initially $U$ containing $u$) can get larger (possibly merging with other components). Indeed, all new red edges in $H_j$ are incident to $u'$, thus the modifications in $H_j^R$ are also limited to $u'$ and its neighborhood, and the only component of $H_j^R$ which can be different from those in $H_i^R$ contains $u'$. 
Moreover, $U'$ might also be adjacent to a component $U'' \in V(B_j)$ such that no component that got merged into $U'$ is adjacent to $U''$ in $B_i$.
Indeed, $u$ has been connected to the elements of $N_{G-V(H_i)}(v)$, and that might have created new paths from $U$ to some other components (and these changes can propagate because of the transitivity of $E(B_j)$). Note that the label \emph{(non)-isolated} can also possibly change for $U$.

Note that since our target trigraph can be a pseudoinduced subtrigraph of $R_j^\ell$, we do not have to add red edges between descendants of $H$. Let us consider the following course of action:
\begin{enumerate}
    \item If $U'$ is isolated in $B_j$, destroy all centipedes attached to vertices of $U'$.
    \item Move one by one all centipedes (if any) attached to $U'$ to some vertex $u'$ of $U'$ such that $N_{G_j}(u')\nsubseteq V(H_j)$, using Lemma~\ref{lem:move-cat} by following a path contained in $U'$, and merge them there (as soon as there are two together) into one with Lemma~\ref{lem:merge-cats}.
    \item Suppose $U'$ is non-isolated in $B_j$ and there is no centipede attached to $U'$ yet. Observe that this is possible only if $u'$ has a neighbor $v'$ outside of $H_j$ such that there is a path from $v'$ to a component $U'' \in V(B_j)$, $U' \ne U''$ (in $G_i$, $v'$ is a neighbor of $v$). Create a centipede $\text{cen}(1,t')$---where $t'=(t-2f_H(\ell)-16)/3$---attached to $u'$ using the centipede of $U''$ (with Lemma~\ref{lem:new-cat}).
    \item Create all other missing edges of $B_j$, always with a leg path of length $f_H(\ell)$, by repeated use of Lemma~\ref{lem:connect-cats}.
    \item Shorten the paths and tails of all centipedes to reach $f_H(\ell)$ using Observation~\ref{obs:shortening}.
\end{enumerate}

Now let us check whether the tails and leg paths in $R_i^{\ell + 1}$ are long enough to follow the described procedure. Let us denote $t := |V(H)|$ and observe $t \ge |V(H_j)| \ge |V(B_j)|$. Let us now look at all the operations that we perform.
\begin{itemize}
    \item Each centipede is moved at most $t$ steps. This shortens the tail by at most $2t$ vertices and the leg paths by at most $t$ vertices, see Lemma~\ref{lem:move-cat}.
    \item We have performed at most $t$ merges of two centipedes. This shortens the tail by at most $2t^2$ vertices and the leg paths by at most $4t^2$ vertices, see Lemma~\ref{lem:merge-cats}. Note that for the tails, we obtained this by upper-bounding the degree of a centipede by $t$.
    \item We create at most $1$ new centipede, namely in step 3 (in the case $v \in V(H_i)$, no new centipede is created). This shortens the tail from length $x$ down to $(x-2f_H(\ell)-16)/3$ but it does not shorten leg paths, see Lemma~\ref{lem:new-cat}.
    \item We create at most $t$ leg paths for each centipede, namely in step 4 (in the case $v \in V(H_i)$, no new leg paths are created). This shortens the tail by at most $(2f_H(\ell)+6)\cdot t$ but it does not shorten leg paths, see Lemma~\ref{lem:connect-cats}.
\end{itemize}

We need to ensure that the length of each tail and leg path is at least $f_H(\ell)$ after all the operations.
For the tails, we must check:
 
$$  (f_H(\ell+1)-2t-2t^2-2f_H(\ell)-16)/3-(2f_H(\ell)+6)t \ge f_H(\ell)  $$

We can see that since $t\geq 1$ and $f_H(\ell)\ge 1$, we can progressively strengthen the condition:
\begin{align*}
    f_H(\ell+1)-2t-2t^2-2f_H(\ell)-16 &\ge 27\cdot t\cdot f_H(\ell)\\
    f_H(\ell+1) &\ge 49\cdot t^2\cdot f_H(\ell)
\end{align*}
The last inequality is satisfied by definition of $f_H$ (recall $f_H(\ell) = (3^{t+4}\cdot t^2)^{\ell}$).

No operation (aside from the final shortening) capable of shortening leg paths was used after the creation of new connections between centipedes, so for leg paths, we only need to check that the leg paths that existed already in $R_i^{\ell+1}$ are long enough. Recall that the only operations which shorten leg paths are moving and merging. Thus it is enough to ensure $f_H(\ell+1) - t - 4t^2 \geq f_H(\ell)$. However, it is easy to observe that the following stronger condition is satisfied: $f_H(\ell+1)\ge 5t^2f_H(\ell)$.

Last but not least, we never create a vertex with red degree higher than $\tww(G)+1$: for outer vertices, this is guaranteed by all operations defined in Subsection~\ref{sub:centipedes}, and each descendant $w$ of $H$ either is adjacent to at most one outer vertex or it has two such neighbors but that can happen only when merging two centipedes attached to $w$, which means that $w$ has a neighbor outside of $H_j$ also in $G_j$.
\end{proof} \fi

We now combine the previous two lemmas to obtain a contraction sequence of the representative for $G_1$---the first trigraph in $C_{BS}$ as well as in $C$. More generally:

\iflong \begin{lemma} \fi \ifshort \begin{lemma}[$\clubsuit$]\fi\label{lem:suffix_contract}
Let $G_i$ be a trigraph in $C_{BS}$ such that there are $\ell$ trigraphs after $G_i$ in $C_{BS}$. There is a contraction sequence of $R_i^\ell$ whose width is at most $\tww(G)+1$. 
\end{lemma}

\iflong \begin{proof}
We will prove the result by induction.
If $\ell=0$, then $R_i^0$ is either $K_1$---if the single vertex of the corresponding blueprint is isolated---or a red $P_3$ if it is non-isolated (since a centipede $\text{cen}(0, f_H(0)) \cong P_2$ is attached to $H_i$). Hence, in this case any contraction sequence of $R_i^\ell$ satisfies the requirement.

If $\ell\geq 1$, then there is a trigraph $G_j$ right after $G_i$ in $C_{BS}$.
Let $C_j$ be the contraction sequence for $R_j^{\ell-1}$ provided by the induction hypothesis.
If $H_i = H_j$, then by Lemma~\ref{lem:non-dec}, there is a partial contraction sequence $C^p$ from $R_i^\ell$ to $R_j^{\ell-1}$ such that $w(C^p) \le \tww(G)+1$, and we obtain our desired contraction sequence for $R_i^\ell$ by simply concatenating $C^p$ and $C_j$.
Otherwise, by Lemma~\ref{lem:dec}, there is a partial contraction sequence $C^p$ from $R_i^\ell$ to a pseudoinduced subtrigraph $R^*$ of $R_j^{\ell-1}$ such that $w(C^p) \le \tww(G) + 1$. Let $C_j^*$ be the contraction sequence of $R^*$ such that $w(C_j^*) \le w(C_j)$ obtained by using Observation~\ref{obs:induced}. Now we obtained our desired  contraction sequence for $R_i^\ell$ by concatenating $C^p$ and $C_j^*$.
\end{proof} \fi

One more thing we need to do is initialization: the trigraph we are interested in, i.e., the trigraph obtained from $G$ by shortening all dangling paths to bounded length, does not contain any centipedes. This is handled by (the proof of) Theorem~\ref{thm:paths_shortening} below, which also summarizes the outcome of this subsection.

\iflong
\begin{theorem}
\fi
\ifshort
\begin{theorem}[$\clubsuit$]
\fi
\label{thm:paths_shortening}
Let $G$ be a tidy $(H, \ca P)$-graph such that $\tww(G) \ge 3$ and all paths in $\ca P$ have length at least $3\cdot f_H(|2V(H)|^2) +9$ and let $G'=(H, \ca P')$ be any trigraph obtained from $G$ by shortening paths in $\ca P$ to arbitrary lengths no shorter than $3\cdot f_H(2|V(H)|^2) +9$. Then $\tww(G') \le \tww(G) + 1$.
\end{theorem}

\iflong \begin{proof}

Observe that $G'$ is a tidy $(H, \ca P')$-graph for some $\ca P'$.
We want to construct a contraction sequence $C'$ of width at most $\tww(G) + 1$ from an optimal contraction sequence $C$ of $G$.
Remark that for any optimal contraction $C$ of $G'$, it holds that
$|C_{BS}|\leq 2|V(H)|^2$ because $C_{BS}$ cannot contain two consecutive non-decisive steps, and each decisive step either adds a new red edge or contracts together two vertices of $H'$ (or both), thus the number of decisive steps is bounded by $|V(H)|(|V(H)|-1)/2 + (|V(H)|-1)\leq |V(H)|^2-1$.

Thus, using Lemma~\ref{lem:suffix_contract} there exists a positive integer $\ell \leq 2|V(H)|^2$, such that $R_1^\ell$ has twin-width at most $\tww(G)+1$. It is important to notice that here $R_1^\ell$ is well-defined (for a given $\ell$) even if we haven't fixed any contraction sequence $C$ to define $C_{BS}$. This comes from the fact that in any big step contraction sequence of $G$, the graph $G_1$ is exactly $G$ by definition, so we can define the first blueprint without knowing anything about the contraction sequence itself. We now consider $\ell$ fixed (even though we do not know it)
and will prove that we can construct a partial contraction sequence leading from $G'$ to $R_1^\ell$ of width at most $\tww(G)+1$.

Let $u \in V(H)$ be a connector in $G'$. By definition of tidiness, $u$ is adjacent to only one vertex of $\sump\ca P'$ in $G'$, $u$ has a unique neighbor $u'$ in $H'$, and $u'$ has positive black degree. Thus $u$ is the only vertex in its fully-red component, recall Definition~\ref{def:bluep}, and so each centipede in the representative $R_1^\ell$ of $G'$ has exactly one leg. 

Observe that since $\ell\le 2|V(H)|^2$, all paths in $\ca P'$ contain at least $3f_H(\ell)+9$ vertices. For each path $(u_0,..,u_x)$ in $\ca P'$, $C'$ first repeatedly contracts $u_{3\cdot f_H(\ell)+7}$ with the next vertex in line, until the end of the path. Then, $C'$ contracts together the pairs of vertices $(u_0,u_{f_H(\ell)+3})$, $(u_1,u_{f_H(\ell)+2})$, $(u_{2\cdot f_H(\ell)+4},u_{3\cdot f_H(\ell)+7})$ and finally $(u_{2\cdot f_H(\ell)+5},u_{3\cdot f_H(\ell)+6})$. Note that by symmetry, it does not matter in which direction we index the vertices of the path. Also note that all vertices throughout this process have red degree at most $4 \le \tww(G) + 1$.

Observe that the obtained trigraph is exactly $R_1^\ell$. %
Now it suffices to continue by applying Lemma~\ref{lem:suffix_contract}, and we obtain a contraction sequence for $G'$ whose width is at most $\tww(G)+1$.
\end{proof} \fi

\subsection{Putting Everything Together}

We are now ready to prove Theorem~\ref{thm:tww-3+}. 

\thmtwo*
 
 \ifshort
 \begin{proof}[Proof Sketch]
 We begin by handling the case where $\tww(G)\le 2$ by invoking Theorems~\ref{thm:tww1-deciding}~and~\ref{thm:tww2}. For the rest of the proof, we assume $\tww(G)\ge 3$.
Here, we first use Theorem~\ref{thm:pruning} to get in $n^{\bigoh(1)}$ time an original $(H, \ca{P})$-graph such that $|V(H)|\leq 16k$ and $|\ca{P}|\leq 4k$. Recall that this $(H, \ca P)$-graph has effectively the same twin-width as $G$, so any optimal contraction sequence for it can be lifted to an optimal one for $G$.
Using Corollary~\ref{cor:no-stumps}, we obtain---also in $n^{\bigoh(1)}$ time---a tidy $(H', \ca{P}')$-graph $G'$ such that $|V(H')|\leq 112k$, $|\ca{P}'|\leq 4k$, and $G'$ still has effectively the same twin-width as $G$.

At this point, we check the length of each path in $\ca P'$, whereas if we identify a path $P\in \ca P'$ whose length is below the bound required by Theorem~\ref{thm:paths_shortening} (w.r.t.\ the current size of $H'$), we add $P$ into $H'$ and update our choices of $H'$ and $\ca P'$ accordingly. After exhaustively completing the above check, we are guaranteed to have satisfied the conditions of Theorem~\ref{thm:paths_shortening}. We now begin constructing our contraction sequence for $G'$ as follows. First, we iteratively contract the paths which remain in $\ca P'$ until they have length precisely $3\cdot f_{H'}(2|V(H')|^2) +9$; recall that by Theorem~\ref{thm:paths_shortening}, we are guaranteed that the resulting graph $G^*$ has twin-width at most one larger than $G$ (and also $G'$). Moreover, the number of vertices in $G^*$ can be upper-bounded by a non-elementary function of our parameter, specifically $2^{2^{\dots^{2^{\bigoh(log(k))}}}}$ where the height of the tower of exponents is upper-bounded by $4k+3$. At this point, we apply Observation~\ref{obs:brute-force} to construct an optimal contraction sequence of $G^*$ and append it after the trivial sequence of contractions which produced $G^*$. The proof now follows by the fact that $G'$ has effectively the same twin-width as $G$.
 \end{proof}
 \fi
 
\iflong \begin{proof}
 We begin by handling the case where $\tww(G)\le 2$ by invoking Theorems~\ref{thm:tww1-deciding}~and~\ref{thm:tww2}. For the rest of the proof, we assume $\tww(G)\ge 3$.
Here, we first use Theorem~\ref{thm:pruning} to get in $n^{\bigoh(1)}$ time an original $(H, \ca{P})$ graph such that $|V(H)|\leq 16k$ and $|\ca{P}|\leq 4k$. Recall that this $(H, \ca P)$-graph has effectively the same twin-width as $G$, so any optimal contraction sequence for it can be lifted to an optimal one for $G$.
Using Corollary~\ref{cor:no-stumps}, we obtain---also in $n^{\bigoh(1)}$ time---a tidy $(H', \ca{P}')$-graph $G'$ such that $|V(H')|\leq 112k$, $|\ca{P}'|\leq 4k$, and $G'$ still has effectively the same twin-width as $G$.

From now on, we will use $t=|V(H')|$.
If some of the paths in $\ca{P}'$ are strictly shorter than $3\cdot f_{H'}(2t^2)+9$ (remember that $f_{H'}(\ell)= (3^{t+4}t^2)^\ell$), we update $H'$ by adding all vertices of those paths to it, and we remove them from $\ca{P}'$. It is factually the same graph $G'$, but now the $H'$ part is bigger. Note that the value of $f_{H'}(2t^2)$ is increasing each time a path is added to $H'$ because both the function and arguments are changing. Hence, after a path is added to $H'$, another paths might become too short.  We repeat this process until it holds that all paths remaining in $\ca P'$ are longer than the final value of $f_{H'}(2t^2)$.

Let us bound the size of $H'$ after this operation. When adding one single path to $H'$, the new size is at most $t+8+3\cdot f_{H'}(2t^2)=t+8+3(3^{t+4}t^2)^{2t^2}\leq 4(3t)^{6t^3}\leq t^{13t^3}\leq t^{t^6}=2^{2^{6log(t)}log(t)}$ if we assume that $t\geq 3$, which we can since $t\leq2 \Rightarrow k\leq1 \Rightarrow \tww(G)\leq 2$ since $H$ contains at least the endpoints of the feedback edge set, combined with Theorem~\ref{thm:fenone}. In the worst case, we add each of the paths in $\ca P'$ (at most $4k$) to $H'$, which leads to a maximum size of $H'$ of $2^{2^{\dots^{2^{\bigoh(log(k))}}}}$  where the tower's height is $4k+2$.

We now define $G^*=(H^*, \ca{P}^*)$ by setting $H^*=H'$ and shortening all the paths in $\ca(P')$ to a length of $3\cdot f_{H^*}(2|V(H^*)|^2)+9$ vertices. Note that to obtain a corresponding partial contraction sequence $C_{\rightarrow G^*}$ from $G'$ to $G^*$, it suffices to greedily contract arbitrary consecutive vertices in each path until it has the right size. This greedy contraction sequence only affects vertices from dangling paths, without ever increasing their degree beyond $2$, so the width of $C_{\rightarrow G^*}$ is at most the maximum between $2$ and the maximum red degree in $G^*$.

By Theorem~\ref{thm:paths_shortening}, it holds that $\tww(G^*)\leq \tww(G)+1$. Moreover, $|V(G^*)|$ is bounded by $2^{2^{\dots^{2^{\bigoh(log(k))}}}}$ where the tower's height is $4k+3$.
Now we apply Observation~\ref{obs:brute-force} to find an optimal contraction sequence of $G^*$ in time $2^{\bigoh(m\cdot\log m)}$, where $m = |V(G^*)|$.
Using Knuth's up-arrow notation, we can upper bound this running time by $2 \upuparrows \bigoh(k)$~\cite{Knuth1976}.

Now, to lift the solution $C^*$ obtained for $G^*$ to our initial graph $G$, we need first to prefix it with $C_{\rightarrow G^*}$ to obtain a contraction sequence of width at most $\tww(G)+1$ for $G'$. Moreover, we know that $G'$ has effectively the same twin-width as the original $(H, \ca P)$-graph it was created from, which itself has effectively the same twin-width as the initial graph $G$. By Definition~\ref{def:effective}, we can obtain in polynomial time a contraction sequence $C$ of width at most $\tww(G)+1$ for $G$.

The total running time is then $\left( 2 \upuparrows \bigoh(k)\right) \cdot n^{\bigoh(1)}$.
\end{proof} \fi

\section{Concluding Remarks}

While the feedback edge number parameterization employed by our algorithms is highly restrictive, we believe Theorems~\ref{thm:tww2} and~\ref{thm:tww-3+} represent a tangible and important first step towards more general algorithms for computing near-optimal contraction sequences, with the ``holy grail'' being a fixed-parameter algorithm for computing near-optimal contraction sequences parameterized by twin-width itself. The natural next goals in this line of research would be to obtain fixed-parameter algorithms for the problem when parameterized by treedepth~\cite{sparsity} and then by treewidth~\cite{RobertsonS86}. 

Towards this direction, we note that it is not at all obvious how one could apply classical tools such as \emph{typical sequences}~\cite{BodlaenderGHK95,DBLP:conf/innovations/EibenGHJK22,BodlaenderJT23} in the context of computing contraction sequences. At least for treedepth, it may be possible to employ the general approach developed in Section~\ref{sec:tww-3+}---in particular, establishing the existence of a near-optimal but ``well-structured'' contraction sequence and using that to identify safe reduction rules---but the details and challenges arising there seem to differ significantly from the ones handled in this article. 

Last but not least, we remark that the algorithms developed here rely on reduction rules which are provably safe, simple to implement, and run in polynomial time; we believe these may potentially be of interest for heuristic and empirical purposes. We also believe that the additive error of $1$ incurred by Theorem~\ref{thm:tww-3+} is avoidable, albeit this may perhaps be seen as a less pressing question than settling the approximability of twin-width under the parameterizations outlined in the previous paragraph.

\bibliographystyle{plainurl}
\bibliography{main.bib}

\begin{thebibliography}{10}

\bibitem{DBLP:conf/iwpec/BalabanH21}
Jakub Balab{\'{a}}n and Petr Hlinen{\'{y}}.
\newblock Twin-width is linear in the poset width.
\newblock In Petr~A. Golovach and Meirav Zehavi, editors, {\em 16th
  International Symposium on Parameterized and Exact Computation, {IPEC} 2021,
  September 8-10, 2021, Lisbon, Portugal}, volume 214 of {\em LIPIcs}, pages
  6:1--6:13. Schloss Dagstuhl - Leibniz-Zentrum f{\"{u}}r Informatik, 2021.
\newblock \href {https://doi.org/10.4230/LIPIcs.IPEC.2021.6}
  {\path{doi:10.4230/LIPIcs.IPEC.2021.6}}.

\bibitem{DBLP:conf/wg/BalabanHJ22}
Jakub Balab{\'{a}}n, Petr Hlinen{\'{y}}, and Jan Jedelsk{\'{y}}.
\newblock Twin-width and transductions of proper k-mixed-thin graphs.
\newblock In Michael~A. Bekos and Michael Kaufmann, editors, {\em
  Graph-Theoretic Concepts in Computer Science - 48th International Workshop,
  {WG} 2022, T{\"{u}}bingen, Germany, June 22-24, 2022, Revised Selected
  Papers}, volume 13453 of {\em Lecture Notes in Computer Science}, pages
  43--55. Springer, 2022.
\newblock \href {https://doi.org/10.1007/978-3-031-15914-5\_4}
  {\path{doi:10.1007/978-3-031-15914-5\_4}}.

\bibitem{BannisterCE18}
Michael~J. Bannister, Sergio Cabello, and David Eppstein.
\newblock Parameterized complexity of 1-planarity.
\newblock {\em J. Graph Algorithms Appl.}, 22(1):23--49, 2018.
\newblock \href {https://doi.org/10.7155/jgaa.00457}
  {\path{doi:10.7155/jgaa.00457}}.

\bibitem{DBLP:conf/icalp/BergeBD22}
Pierre Berg{\'{e}}, {\'{E}}douard Bonnet, and Hugues D{\'{e}}pr{\'{e}}s.
\newblock Deciding twin-width at most 4 is np-complete.
\newblock In Mikolaj Bojanczyk, Emanuela Merelli, and David~P. Woodruff,
  editors, {\em 49th International Colloquium on Automata, Languages, and
  Programming, {ICALP} 2022, July 4-8, 2022, Paris, France}, volume 229 of {\em
  LIPIcs}, pages 18:1--18:20. Schloss Dagstuhl - Leibniz-Zentrum f{\"{u}}r
  Informatik, 2022.
\newblock \href {https://doi.org/10.4230/LIPIcs.ICALP.2022.18}
  {\path{doi:10.4230/LIPIcs.ICALP.2022.18}}.

\bibitem{DBLP:journals/algorithmica/BergougnouxEGOR21}
Benjamin Bergougnoux, Eduard Eiben, Robert Ganian, Sebastian Ordyniak, and
  M.~S. Ramanujan.
\newblock Towards a polynomial kernel for directed feedback vertex set.
\newblock {\em Algorithmica}, 83(5):1201--1221, 2021.
\newblock \href {https://doi.org/10.1007/s00453-020-00777-5}
  {\path{doi:10.1007/s00453-020-00777-5}}.

\bibitem{Bodlaender96}
Hans~L. Bodlaender.
\newblock A linear-time algorithm for finding tree-decompositions of small
  treewidth.
\newblock {\em {SIAM} J. Comput.}, 25(6):1305--1317, 1996.
\newblock \href {https://doi.org/10.1137/S0097539793251219}
  {\path{doi:10.1137/S0097539793251219}}.

\bibitem{BodlaenderGHK95}
Hans~L. Bodlaender, John~R. Gilbert, Hj{\'{a}}lmtyr Hafsteinsson, and Ton
  Kloks.
\newblock Approximating treewidth, pathwidth, frontsize, and shortest
  elimination tree.
\newblock {\em J. Algorithms}, 18(2):238--255, 1995.
\newblock \href {https://doi.org/10.1006/jagm.1995.1009}
  {\path{doi:10.1006/jagm.1995.1009}}.

\bibitem{BodlaenderJT23}
Hans~L. Bodlaender, Lars Jaffke, and Jan~Arne Telle.
\newblock Typical sequences revisited - computing width parameters of graphs.
\newblock {\em Theory Comput. Syst.}, 67(1):52--88, 2023.
\newblock \href {https://doi.org/10.1007/s00224-021-10030-3}
  {\path{doi:10.1007/s00224-021-10030-3}}.

\bibitem{DBLP:journals/siamdm/BodlaenderJK13}
Hans~L. Bodlaender, Bart M.~P. Jansen, and Stefan Kratsch.
\newblock Preprocessing for treewidth: {A} combinatorial analysis through
  kernelization.
\newblock {\em {SIAM} J. Discret. Math.}, 27(4):2108--2142, 2013.
\newblock \href {https://doi.org/10.1137/120903518}
  {\path{doi:10.1137/120903518}}.

\bibitem{DBLP:conf/soda/BonnetGKTW21}
{\'{E}}douard Bonnet, Colin Geniet, Eun~Jung Kim, St{\'{e}}phan Thomass{\'{e}},
  and R{\'{e}}mi Watrigant.
\newblock Twin-width {II:} small classes.
\newblock In D{\'{a}}niel Marx, editor, {\em Proceedings of the 2021 {ACM-SIAM}
  Symposium on Discrete Algorithms, {SODA} 2021, Virtual Conference, January 10
  - 13, 2021}, pages 1977--1996. {SIAM}, 2021.
\newblock \href {https://doi.org/10.1137/1.9781611976465.118}
  {\path{doi:10.1137/1.9781611976465.118}}.

\bibitem{DBLP:conf/icalp/BonnetG0TW21}
{\'{E}}douard Bonnet, Colin Geniet, Eun~Jung Kim, St{\'{e}}phan Thomass{\'{e}},
  and R{\'{e}}mi Watrigant.
\newblock Twin-width {III:} max independent set, min dominating set, and
  coloring.
\newblock In Nikhil Bansal, Emanuela Merelli, and James Worrell, editors, {\em
  48th International Colloquium on Automata, Languages, and Programming,
  {ICALP} 2021, July 12-16, 2021, Glasgow, Scotland (Virtual Conference)},
  volume 198 of {\em LIPIcs}, pages 35:1--35:20. Schloss Dagstuhl -
  Leibniz-Zentrum f{\"{u}}r Informatik, 2021.
\newblock \href {https://doi.org/10.4230/LIPIcs.ICALP.2021.35}
  {\path{doi:10.4230/LIPIcs.ICALP.2021.35}}.

\bibitem{DBLP:conf/stoc/BonnetGMSTT22}
{\'{E}}douard Bonnet, Ugo Giocanti, Patrice~Ossona de~Mendez, Pierre Simon,
  St{\'{e}}phan Thomass{\'{e}}, and Szymon Torunczyk.
\newblock Twin-width {IV:} ordered graphs and matrices.
\newblock In Stefano Leonardi and Anupam Gupta, editors, {\em {STOC} '22: 54th
  Annual {ACM} {SIGACT} Symposium on Theory of Computing, Rome, Italy, June 20
  - 24, 2022}, pages 924--937. {ACM}, 2022.
\newblock \href {https://doi.org/10.1145/3519935.3520037}
  {\path{doi:10.1145/3519935.3520037}}.

\bibitem{DBLP:conf/stacs/BonnetGMT23}
{\'{E}}douard Bonnet, Ugo Giocanti, Patrice~Ossona de~Mendez, and St{\'{e}}phan
  Thomass{\'{e}}.
\newblock Twin-width {V:} linear minors, modular counting, and matrix
  multiplication.
\newblock In Petra Berenbrink, Patricia Bouyer, Anuj Dawar, and
  Mamadou~Moustapha Kant{\'{e}}, editors, {\em 40th International Symposium on
  Theoretical Aspects of Computer Science, {STACS} 2023, March 7-9, 2023,
  Hamburg, Germany}, volume 254 of {\em LIPIcs}, pages 15:1--15:16. Schloss
  Dagstuhl - Leibniz-Zentrum f{\"{u}}r Informatik, 2023.
\newblock \href {https://doi.org/10.4230/LIPIcs.STACS.2023.15}
  {\path{doi:10.4230/LIPIcs.STACS.2023.15}}.

\bibitem{DBLP:conf/soda/BonnetKRT22}
{\'{E}}douard Bonnet, Eun~Jung Kim, Amadeus Reinald, and St{\'{e}}phan
  Thomass{\'{e}}.
\newblock Twin-width {VI:} the lens of contraction sequences.
\newblock In Joseph~(Seffi) Naor and Niv Buchbinder, editors, {\em Proceedings
  of the 2022 {ACM-SIAM} Symposium on Discrete Algorithms, {SODA} 2022, Virtual
  Conference / Alexandria, VA, USA, January 9 - 12, 2022}, pages 1036--1056.
  {SIAM}, 2022.
\newblock \href {https://doi.org/10.1137/1.9781611977073.45}
  {\path{doi:10.1137/1.9781611977073.45}}.

\bibitem{DBLP:journals/algorithmica/BonnetKRTW22}
{\'{E}}douard Bonnet, Eun~Jung Kim, Amadeus Reinald, St{\'{e}}phan
  Thomass{\'{e}}, and R{\'{e}}mi Watrigant.
\newblock Twin-width and polynomial kernels.
\newblock {\em Algorithmica}, 84(11):3300--3337, 2022.
\newblock \href {https://doi.org/10.1007/s00453-022-00965-5}
  {\path{doi:10.1007/s00453-022-00965-5}}.

\bibitem{BonnetKTW22}
{\'{E}}douard Bonnet, Eun~Jung Kim, St{\'{e}}phan Thomass{\'{e}}, and
  R{\'{e}}mi Watrigant.
\newblock Twin-width {I:} tractable {FO} model checking.
\newblock {\em J. {ACM}}, 69(1):3:1--3:46, 2022.
\newblock \href {https://doi.org/10.1145/3486655} {\path{doi:10.1145/3486655}}.

\bibitem{CyganFKLMPPS15}
Marek Cygan, Fedor~V. Fomin, Lukasz Kowalik, Daniel Lokshtanov, D{\'{a}}niel
  Marx, Marcin Pilipczuk, Michal Pilipczuk, and Saket Saurabh.
\newblock {\em Parameterized Algorithms}.
\newblock Springer, 2015.
\newblock \href {https://doi.org/10.1007/978-3-319-21275-3}
  {\path{doi:10.1007/978-3-319-21275-3}}.

\bibitem{Diestel}
Reinhard Diestel.
\newblock {\em Graph Theory, 4th Edition}, volume 173 of {\em Graduate texts in
  mathematics}.
\newblock Springer, 2012.

\bibitem{DowneyF13}
Rodney~G. Downey and Michael~R. Fellows.
\newblock {\em Fundamentals of Parameterized Complexity}.
\newblock Texts in Computer Science. Springer, 2013.
\newblock \href {https://doi.org/10.1007/978-1-4471-5559-1}
  {\path{doi:10.1007/978-1-4471-5559-1}}.

\bibitem{DBLP:conf/innovations/EibenGHJK22}
Eduard Eiben, Robert Ganian, Thekla Hamm, Lars Jaffke, and O{-}joung Kwon.
\newblock A unifying framework for characterizing and computing width measures.
\newblock In Mark Braverman, editor, {\em 13th Innovations in Theoretical
  Computer Science Conference, {ITCS} 2022, January 31 - February 3, 2022,
  Berkeley, CA, {USA}}, volume 215 of {\em LIPIcs}, pages 63:1--63:23. Schloss
  Dagstuhl - Leibniz-Zentrum f{\"{u}}r Informatik, 2022.
\newblock \href {https://doi.org/10.4230/LIPIcs.ITCS.2022.63}
  {\path{doi:10.4230/LIPIcs.ITCS.2022.63}}.

\bibitem{EppsteinConfluent}
David Eppstein.
\newblock The widths of strict outerconfluent graphs.
\newblock {\em CoRR}, abs/2308.03967, 2023.
\newblock \href {https://arxiv.org/abs/2308.03967} {\path{arXiv:2308.03967}}.

\bibitem{FichteGHSO23}
Johannes~Klaus Fichte, Robert Ganian, Markus Hecher, Friedrich Slivovsky, and
  Sebastian Ordyniak.
\newblock Structure-aware lower bounds and broadening the horizon of
  tractability for {QBF}.
\newblock In {\em {LICS}}, pages 1--14, 2023.
\newblock \href {https://doi.org/10.1109/LICS56636.2023.10175675}
  {\path{doi:10.1109/LICS56636.2023.10175675}}.

\bibitem{DBLP:conf/stoc/FominK22}
Fedor~V. Fomin and Tuukka Korhonen.
\newblock Fast fpt-approximation of branchwidth.
\newblock In Stefano Leonardi and Anupam Gupta, editors, {\em {STOC} '22: 54th
  Annual {ACM} {SIGACT} Symposium on Theory of Computing, Rome, Italy, June 20
  - 24, 2022}, pages 886--899. {ACM}, 2022.
\newblock \href {https://doi.org/10.1145/3519935.3519996}
  {\path{doi:10.1145/3519935.3519996}}.

\bibitem{GanianH10}
Robert Ganian and Petr Hlinen{\'{y}}.
\newblock On parse trees and myhill-nerode-type tools for handling graphs of
  bounded rank-width.
\newblock {\em Discret. Appl. Math.}, 158(7):851--867, 2010.
\newblock \href {https://doi.org/10.1016/j.dam.2009.10.018}
  {\path{doi:10.1016/j.dam.2009.10.018}}.

\bibitem{GanianK21}
Robert Ganian and Viktoriia Korchemna.
\newblock The complexity of bayesian network learning: Revisiting the
  superstructure.
\newblock In Marc'Aurelio Ranzato, Alina Beygelzimer, Yann~N. Dauphin, Percy
  Liang, and Jennifer~Wortman Vaughan, editors, {\em Advances in Neural
  Information Processing Systems 34: Annual Conference on Neural Information
  Processing Systems 2021, NeurIPS 2021, December 6-14, 2021, virtual}, pages
  430--442, 2021.
\newblock URL:
  \url{https://proceedings.neurips.cc/paper/2021/hash/040a99f23e8960763e680041c601acab-Abstract.html}.

\bibitem{GanianO21}
Robert Ganian and Sebastian Ordyniak.
\newblock The power of cut-based parameters for computing edge-disjoint paths.
\newblock {\em Algorithmica}, 83(2):726--752, 2021.
\newblock \href {https://doi.org/10.1007/s00453-020-00772-w}
  {\path{doi:10.1007/s00453-020-00772-w}}.

\bibitem{DBLP:conf/sat/GanianPSSS22}
Robert Ganian, Filip Pokr{\'{y}}vka, Andr{\'{e}} Schidler, Kirill Simonov, and
  Stefan Szeider.
\newblock Weighted model counting with twin-width.
\newblock In Kuldeep~S. Meel and Ofer Strichman, editors, {\em 25th
  International Conference on Theory and Applications of Satisfiability
  Testing, {SAT} 2022, August 2-5, 2022, Haifa, Israel}, volume 236 of {\em
  LIPIcs}, pages 15:1--15:17. Schloss Dagstuhl - Leibniz-Zentrum f{\"{u}}r
  Informatik, 2022.
\newblock \href {https://doi.org/10.4230/LIPIcs.SAT.2022.15}
  {\path{doi:10.4230/LIPIcs.SAT.2022.15}}.

\bibitem{HlinenyO08}
Petr Hlinen{\'{y}} and Sang{-}il Oum.
\newblock Finding branch-decompositions and rank-decompositions.
\newblock {\em {SIAM} J. Comput.}, 38(3):1012--1032, 2008.
\newblock \href {https://doi.org/10.1137/070685920}
  {\path{doi:10.1137/070685920}}.

\bibitem{Knuth1976}
Donald~E. Knuth.
\newblock Mathematics and computer science: Coping with finiteness.
\newblock {\em Science}, 194(4271):1235--1242, December 17, 1976.

\bibitem{DBLP:conf/iwpec/KobayashiT16}
Yasuaki Kobayashi and Hisao Tamaki.
\newblock Treedepth parameterized by vertex cover number.
\newblock In Jiong Guo and Danny Hermelin, editors, {\em 11th International
  Symposium on Parameterized and Exact Computation, {IPEC} 2016, August 24-26,
  2016, Aarhus, Denmark}, volume~63 of {\em LIPIcs}, pages 18:1--18:11. Schloss
  Dagstuhl - Leibniz-Zentrum f{\"{u}}r Informatik, 2016.
\newblock \href {https://doi.org/10.4230/LIPIcs.IPEC.2016.18}
  {\path{doi:10.4230/LIPIcs.IPEC.2016.18}}.

\bibitem{DBLP:conf/focs/Korhonen21}
Tuukka Korhonen.
\newblock A single-exponential time 2-approximation algorithm for treewidth.
\newblock In {\em 62nd {IEEE} Annual Symposium on Foundations of Computer
  Science, {FOCS} 2021, Denver, CO, USA, February 7-10, 2022}, pages 184--192.
  {IEEE}, 2021.
\newblock \href {https://doi.org/10.1109/FOCS52979.2021.00026}
  {\path{doi:10.1109/FOCS52979.2021.00026}}.

\bibitem{sparsity}
Jaroslav Nesetril and Patrice~Ossona de~Mendez.
\newblock {\em Sparsity - Graphs, Structures, and Algorithms}, volume~28 of
  {\em Algorithms and combinatorics}.
\newblock Springer, 2012.
\newblock \href {https://doi.org/10.1007/978-3-642-27875-4}
  {\path{doi:10.1007/978-3-642-27875-4}}.

\bibitem{Oum05}
Sang{-}il Oum.
\newblock Rank-width and vertex-minors.
\newblock {\em J. Comb. Theory, Ser. {B}}, 95(1):79--100, 2005.
\newblock \href {https://doi.org/10.1016/j.jctb.2005.03.003}
  {\path{doi:10.1016/j.jctb.2005.03.003}}.

\bibitem{RobertsonS86}
Neil Robertson and Paul~D. Seymour.
\newblock Graph minors. {II.} algorithmic aspects of tree-width.
\newblock {\em J. Algorithms}, 7(3):309--322, 1986.
\newblock \href {https://doi.org/10.1016/0196-6774(86)90023-4}
  {\path{doi:10.1016/0196-6774(86)90023-4}}.

\bibitem{DBLP:conf/ijcai/SchidlerS23}
Andr{\'{e}} Schidler and Stefan Szeider.
\newblock Computing twin-width with {SAT} and branch {\&} bound.
\newblock In {\em Proceedings of the Thirty-Second International Joint
  Conference on Artificial Intelligence, {IJCAI} 2023, 19th-25th August 2023,
  Macao, SAR, China}, pages 2013--2021. ijcai.org, 2023.
\newblock \href {https://doi.org/10.24963/ijcai.2023/224}
  {\path{doi:10.24963/ijcai.2023/224}}.

\bibitem{UhlmannW13}
Johannes Uhlmann and Mathias Weller.
\newblock Two-layer planarization parameterized by feedback edge set.
\newblock {\em Theor. Comput. Sci.}, 494:99--111, 2013.
\newblock \href {https://doi.org/10.1016/j.tcs.2013.01.029}
  {\path{doi:10.1016/j.tcs.2013.01.029}}.

\end{thebibliography}

\end{document}